\documentclass[11pt,a4paper]{article}
\usepackage[linesnumbered,ruled]{algorithm2e}
\usepackage{titlesec}
\SetAlCapNameFnt{\small	}
\SetAlCapFnt{\small	}
\usepackage[title]{appendix}
\usepackage{soul}
\usepackage{amsmath,amssymb,amsbsy,amsthm,calc,cases,enumerate,url,framed,fancyhdr,lipsum,tikz,enumitem,natbib,etoolbox,bbm}
\usepackage{graphicx}
\usepackage[pageanchor,hypertexnames=false]{hyperref}

\usepackage[font=small,skip=0pt]{caption}

\newcommand{\unparskip}{\vspace{-\parskip}}

\titlespacing\section{0pt}{12pt plus 4pt minus 2pt}{0pt plus 2pt minus 2pt}
\titlespacing\subsection{0pt}{12pt plus 4pt minus 2pt}{0pt plus 2pt minus 2pt}
\titlespacing\subsubsection{0pt}{12pt plus 4pt minus 2pt}{0pt plus 2pt minus 2pt}

\newtheorem{theorem}{Theorem}
\newtheorem{corollary}{Corollary}
\usepackage{bbm}
\usepackage{stmaryrd}
\newtheorem{proposition}{Proposition}%
\newtheorem{assumption}{Assumption}%
\newenvironment{proof_outline}{\begin{proof}[Proof outline]}{\end{proof}}

\newtheorem{thm}{Theorem}
\newtheorem{prop}{Proposition}
\newtheorem{lem}{Lemma}
\newtheorem{cor}{Corollary}
\newtheorem{as}{Assumption}

\BeforeBeginEnvironment{proof}{\unparskip}
\BeforeBeginEnvironment{proof_outline}{\unparskip}
  
\BeforeBeginEnvironment{theorem}{\vspace{0.5\parskip}}
\BeforeBeginEnvironment{corollary}{\vspace{0.5\parskip}}
\BeforeBeginEnvironment{lemma}{\vspace{0.5\parskip}}
\BeforeBeginEnvironment{proposition}{\vspace{0.5\parskip}}
  
\BeforeBeginEnvironment{thm}{\vspace{0.5\parskip}}
\BeforeBeginEnvironment{cor}{\vspace{0.5\parskip}}
\BeforeBeginEnvironment{prop}{\vspace{0.5\parskip}}
\BeforeBeginEnvironment{lem}{\vspace{0.5\parskip}}

\newcommand{\iid}{\overset{\mathrm{iid}}{\sim}}

\DeclareMathOperator*{\argmin}{arg\,min}

\newcommand{\R}{\mathbb{R}}

\newcommand{\E}{\mathbb{E}}
\renewcommand{\Pr}{\mathbb{P}}
\newcommand{\Ind}{\mathbbm{1}}

\newcommand{\dotcup}

\usepackage{xr}

\DeclareRobustCommand\full {\tikz[baseline=-0.5ex]\draw[blue,thick] (0,0)--(0.5,0);}
\DeclareRobustCommand\dotted{\tikz[baseline=-0.5ex]\draw[black,opacity = 0.4,thick,dotted] (0.06,0)--(0.5,0);}
\DeclareRobustCommand\dashed{\tikz[baseline=-0.5ex]\draw[blue!60!black, opacity = 0.7,thick,dashed] (0,0)--(0.52,0);}

\DeclareRobustCommand\full {\tikz[baseline=-0.5ex]\draw[thick,color=blue!60!black] (0,0)--(0.5,0);}

\DeclareRobustCommand\fullmid  {\tikz[baseline=-0.5ex]\draw[thick,color=blue!60!black,opacity = 0.5] (0,0)--(0.5,0);}
\DeclareRobustCommand\fulllow  {\tikz[baseline=-0.5ex]\draw[thick,color=blue!60!black, opacity = 0.25] (0,0)--(0.5,0);}

\usepackage[normalem]{ulem}

\begin{document}

\title{Bayesian Quantile Estimation and Regression\\ with Martingale Posteriors}

\author{Edwin Fong$^{1,}$\thanks{Corresponding author. Email: chefong@hku.hk} \ and Andrew Yiu$^{2}$  \\ \\
$^1$Department of Statistics and Actuarial Science, University of Hong Kong\\
$^2$Department of Statistics, University of Oxford\\
}
\date{}

\maketitle
\vspace{-5mm}
\begin{abstract}
Quantile estimation and regression within the Bayesian framework is challenging as the choice of likelihood and prior is not obvious. 
In this paper, we introduce a novel Bayesian nonparametric method for quantile estimation and regression based on the recently introduced martingale posterior (MP) framework. The core idea of the MP is that posterior sampling is equivalent to predictive imputation, which allows us to break free of the stringent likelihood-prior specification. We  demonstrate that a recursive estimate of a smooth quantile function, subject to a martingale condition, is entirely sufficient for full nonparametric Bayesian inference. We term the resulting posterior distribution as the quantile martingale posterior (QMP), which arises from an implicit generative predictive distribution. Associated with the QMP is an expedient, MCMC-free and parallelizable posterior computation scheme, which can be further accelerated with an asymptotic approximation based on a Gaussian process. Furthermore, the well-known issue of monotonicity in quantile estimation is naturally alleviated through increasing rearrangement due to the connections to the Bayesian bootstrap.  Finally, the QMP has a particularly tractable form that allows for comprehensive theoretical study, which forms a main focus of the work. We demonstrate the ease of posterior computation in simulations and real data experiments. 
\end{abstract}

\section{Introduction}
Quantile estimation and regression has wide applications in fields such as econometrics and biostatistics \citep{Koenker1978}. 
{The Bayesian approach has garnered attention due to the ability to fully quantify uncertainty through the posterior distribution.}
However, a Bayesian equivalent is not immediately obvious as the need to specify a likelihood is challenging.  \cite{Yu2001,Yang2016} and related works utilize a `working likelihood' based on the asymmetric Laplace distribution, where the quantile of interest parameterizes a potentially misspecified likelihood.  Within the Bayesian nonparametric literature, {the challenge lies in eliciting a valid nonparametric prior}. \cite{Hjort2009} introduced the quantile pyramid, which is a nonparametric prior with support on piecewise linear quantile functions. \cite{Rodrigues2019} and \cite{An2024} extend the quantile pyramid to allow for the introduction of covariate dependence. \cite{Tokdar2012} introduce a semiparametric prior for linear quantile regression which has support on monotone curves; \cite{Yang2017} and \cite{Chen2021} then extend this to more complex covariate spaces. In general, constructing prior distributions for quantile functions is nontrivial, and posterior inference in all cases require the use of Markov chain Monte Carlo (MCMC) techniques, which can often be computationally demanding. 

A recent promising class of approaches that avoids the need to work directly with a likelihood are methods which \textit{generalize} Bayesian inference. One direction is the  generalized Bayesian update of \cite{Bissiri2016}, which relies on a loss function instead of a likelihood, and motivates updating through coherence. The asymmetric Laplace likelihood can be motivated in this fashion, as the likelihood is indeed proportional to the exponentiated check loss function. Another approach is to view Bayesian inference as a \textit{predictive} task by taking advantage of connections between posterior and predictive inference, which has been explored in \cite{Berti2020,Fong2023a,Fortini2020,Fortini2023} and others. Particularly close to our work is the \textit{martingale posterior} (MP) of \cite{Fong2023a}, where the traditional likelihood-prior construct of Bayesian inference is replaced with the elicitation of a sequence of predictive densities, which shares the motivation of the prequential approach of \cite{Dawid1984}. Given observations $Y_{1:n}$, the sequence of predictives is utilized to impute the remainder of the population, $Y_{n+1:\infty}$, from which an estimand can be computed and is then distributed according to the MP.

\subsection{Our contribution}
In this work, we introduce a Bayesian nonparametric method for quantile estimation and regression, motivated from the purely predictive framework of the MP \citep{Fong2023a}. The core idea is to utilize a recursive {estimate} of the quantile function as a \textit{generative predictive}, which is then sampled from and updated to impute $Y_{n+1:\infty}$. 
We will differentiate between quantile functions and quantile function \textit{estimates}, where the first is monotonically increasing but the latter may not be.
 The distribution of the resulting random quantile function of $Y_{n+1:\infty}$ is then termed the \textit{quantile martingale posterior} (QMP).
The generative predictive is essentially a stochastic approximation of the quantile function with an additional coherence condition. 

{The QMP inherits many advantages of the MP framework. Firstly, exact posterior computation is simple and expedient, as MCMC is not required at all. We will see later that a highly accurate approximate posterior sampling reduces computation time even further, making the imputation step negligible in time. Secondly, in many situations, we may not have strong prior information despite wanting to quantify posterior uncertainty. The prior distribution can thus be a nuisance to specify, motivating noninformative priors \citep{Berger2009}. This is particularly true in Bayesian nonparametrics, where the specification of the prior is both technically demanding and challenging to interpret. In contrast to traditional Bayes, the QMP is entirely data-driven and prior-free, and the model is simple to interpret due to connections to stochastic approximation.}

{
The QMP also has unique advantages within the context of quantile estimation.} The issue of monotonicity or quantile crossing is handled automatically by the imputation step in the QMP, and we rely heavily on the useful theory of increasing rearrangements. This is another benefit of working with the predictive framework and specifically with a generative predictive as in our case.
Extensions to incorporate covariate dependence, e.g. for linear quantile regression, is then straightforward again due to connections to stochastic approximation, especially when compared to traditional Bayesian nonparametric priors. Finally, we will be extending beyond the c.i.d. condition required for the original MP, which greatly expands the possible set of models for Bayesian nonparametric inference.

In exchange for these benefits, we will immediately be faced with theoretical challenges, for which solutions form the bulk of this work. In general, theoretical study of the MP 
is challenging due to the inability to rely on standard tools for Bayesian asymptotics, and we now cannot even rely on results from the c.i.d. literature. To study the existence and support of the QMP, we will leverage new tools from the Banach space valued martingale literature, which will aid us greatly. In addition, we will be able to study the weak convergence of the QMP, as well as  posterior consistency and contraction in the frequentist sense, which is novel for MPs. The theoretical results have strong practical implications as they guide model elicitation, hyperparameter setting and approximate sampling.  We hope these methods and tools used are also of independent interest and will be useful for future research in MPs and Bayesian inference in general. {We speculate that the aforementioned theory may also be adapted to the Bayesian estimation of more general monotone functions \citep[e.g.][]{Chakraborty2021}.}

We now provide an outline the paper. In Section \ref{sec:intro}, we will review the role of increasing rearrangement in quantile estimation and the MP framework from \cite{Fong2023a}. We then introduce the QMP in the unconditional setting, and provide intuition as to the various model components and sampling algorithm. Section \ref{sec:theory} will then cover the bulk of the theory, with most derivations postponed for the Appendix. Section \ref{sec:implications} will discuss the practical implications of the theory, with a focus on the setting of a few key hyperparameters and an expedient approximate posterior sampling scheme. Section \ref{sec:quantreg} then extends the QMP for quantile regression, covering similar theory and practical discussions. 
Section \ref{sec:illustrations} demonstrates the QMP in a simulation and real data example, and Section \ref{sec:discussion} concludes with future directions.

\section{Quantile martingale posteriors}\label{sec:intro}
For ease of exposition, we first introduce the quantile martingale posterior without covariate dependence, and extend it to the  quantile regression case in Section \ref{sec:quantreg}. For the remainder of this section, let  $Y_{1:n}$ be $n$ i.i.d. copies of the r.v. $Y \in \R$ from an unknown sampling distribution $P^*$ with cumulative distribution function (CDF) $P^*(y)$. 

\subsection{Quantile functions and increasing rearrangement}
To begin, we outline some prerequisites on the quantile function and its estimators, with a particular focus on \textit{increasing rearrangement} \citep{Chernozhukov2010}. 
 The quantile function $Q^*: (0,1) \to \R$  is the left-continuous, monotonically increasing function defined as
\begin{align*}
    Q^*(u) = \inf\left\{y \in \R: u \leq P^*(y)\right\}.
\end{align*}
The quantile function is particularly useful for inverse-transform sampling from $P^*$, which we strongly leverage in our work. In particular, given a uniform r.v. $V \sim \mathcal{U}(0,1)$, we have that $Q^*(V) \sim P^*$. 
{This is due to the key property that $Q^*(u) \leq y$ if and only if $u \leq P^*(y)$ for all $u \in (0,1)$.} A detailed summary of properties of quantile functions can be found in \citet{Embrechts2013}. For the remainder of the paper, we will assume that that both $P^*(y)$ and $Q^*(u)$ are continuous.

Let $Q_n$ be an \textit{estimate} of the quantile function $Q^*$ from $Y_{1:n}$. A well-known problem in quantile estimation is that $Q_n(u)$ may not be monotonically increasing on $u\in (0,1)$, so it is not a valid quantile function. In the case of quantile regression, this is known as the \textit{quantile crossing problem} \citep{Bassett1982,He1997,Chernozhukov2010}, where the lack of monotonicity causes quantile curves as functions of the covariates to cross one another for different values of $u$. Many solutions to this problem have been proposed, but we will focus particularly on increasing rearrangements, as this occurs naturally under the MP framework. 

For the remainder of the paper, we will denote a potentially non-monotone quantile function estimate as $Q_n$. 
Let $Q_n^{\dagger}$ denote the increasing rearrangement of $Q_n$, which is defined as follows:
\begin{equation}\label{eq:rearrangement}
\begin{aligned}
    P_n(y) &= \int_0^1 \Ind(Q_n(u) \leq y)\, du,\quad 
    Q_n^{\dagger}(u) =  \inf\left\{y \in \R: u \leq P_n(y)\right\}.
\end{aligned}
\end{equation}
$Q_n^{\dagger}$  is then a proper quantile function, where one can see the monotonicity as follows. For $V \sim \mathcal{U}(0,1)$, the function $P_n$ is the CDF of $Q_n\left(V \right)$, so $Q_n^{\dagger}$ is a valid quantile function and must be monotonically increasing. The connection to the bootstrap is hence obvious and of key importance - the quantile estimate $Q_n$ gives us a means to simulate from $P_n$ (or equivalently $Q_n^{\dagger}$) through the inverse transform, which forms the basis of our work. In {Figure \ref{fig:rearrangement} (left)}, we show an example of rearranging a non-monotone $Q_n$ into a monotonically increasing $Q_n^{\dagger}$, with corresponding $P_n$ in {Figure \ref{fig:rearrangement} (right)}. We can see that $Q_n^{\dagger}$ agrees with $Q_n$ in some regions, and preserves continuity. A detailed discussion on properties of rearrangement for quantile estimation can be found in {\cite{Chernozhukov2010}}.

There is also a close connection to rearrangement inequalities \citep{Hardy1952}, which have previously been leveraged in estimation by \cite{Chernozhukov2009} and specifically in quantile estimation/regression by \cite{Chernozhukov2010}.  Many useful properties of $Q_n^{\dagger}$ have also been shown in \cite{Chernozhukov2010}, and we will outline and utilize this theory in Section \ref{sec:theory}. In particular, one can show that $Q_n^{\dagger}$ is always a better estimate of $Q^*$ in terms of $L^p$ distance as a result. Increasing rearrangement also preserves continuity properties, which will be useful for us.

\subsection{Martingale posterior distributions}

The MP is a generalization of the Bayesian framework introduced by \cite{Fong2023a}. The key notion is that Bayesian uncertainty on a parameter of interest $\theta$  arises from the unknown remainder of the population $Y_{n+1:\infty}$ that has yet to be observed. \cite{Fong2023a} show that posterior sampling is equivalent to the predictive imputation of $Y_{n+1:\infty}$ given $Y_{1:n}$, followed by the computation of $\theta$ as an estimand from $Y_{1:\infty}$. This procedure is termed as \textit{predictive resampling}, where the sequence of predictive distributions, $P_n(y) = P(Y_{n+1} \leq y \mid Y_{1:n})$, is used to sequentially impute $Y_{n+1:\infty}$, which is outlined in {Algorithm \ref{alg:predictive_resampling}}.

Armed with this interpretation of Bayesian inference, the MP then generalizes Bayes by eliciting a general sequence of predictive distributions $\{P_n,P_{n+1},\ldots\}$ directly as the statistical model, removing the need for a likelihood and prior, and instead relying on predictive resampling to obtain a posterior distribution on a parameter of interest. In order for the MP on $\theta$ to exist, we require the sequence $P_N$ to converge almost surely to a random probability measure $P_\infty$ when predictive resampling, which is ensured through a martingale condition. In particular, \cite{Fong2023a} requires the following predictive coherence condition, $\E\left[P_{N+1}(y) \mid Y_{1:N}\right] = P_N(y)$
for each $y \in \R$ and all $N \geq n$. 
This then implies that the sequence of imputed observations $Y_{n+1:\infty}$ is \textit{conditionally identically distributed} (c.i.d.), and \cite{Berti2004} show that the c.i.d. condition is sufficient for the existence of a $P_\infty$ which $P_N \to P_\infty$ weakly almost surely. This c.i.d. condition unfortunately greatly constrains the class of predictive distributions one can use for the MP. The MP also has close connections to the Bayesian bootstrap of \cite{Rubin1981}, which has recently had a resurgence in popularity, e.g. \cite{Fong2019a, Nie2023}. Other nonparametric MPs have been suggested in \cite{Cui2023,Cui2024} and \cite{Walker2024}. Parametric versions of the MP have also been introduced in \cite{Walker2022,Holmes2023}, where a  parametric predictive distribution is utilized for predictive resampling. The martingale is now directly the parameter of interest $\theta$, ensuring convergence of an estimator $\theta_N \to \theta_\infty$ instead of $P_N \to P_\infty$, which relaxes the c.i.d. condition.

\cite{Fong2023a} enforce the c.i.d. condition using a nonparametric recursive update for $P_N$ based on the bivariate copula as introduced in \cite{Hahn2018}. This recursive update is inspired by the Dirichlet process mixture model, and takes the form
\begin{align}\label{eq:copula_update}
    P_{N+1}(y) = (1-\alpha_{N+1})P_N(y) + \alpha_{N+1}H_\rho\left(P_N(y), P_N(Y_{N+1})\right),
\end{align}
where $H_\rho(u,v)$ is the conditional distribution of the bivariate Gaussian copula of the form
\begin{align}\label{eq:condit_copula}
H_\rho(u,v) = \Phi\left\{\frac{\Phi^{-1}\left(u\right) - \rho \Phi^{-1}(v)}{\sqrt{1-\rho^2}}\right\},
\end{align}
and  $\rho \in (0,1)$ is the correlation term and $\Phi$ and $\Phi^{-1}$ are the standard normal CDF and its inverse respectively. The weights are usually chosen $\alpha_{N} = O(N^{-1})$ in order for the update to approach the independence copula as $N \to \infty$. Intuitively, the second term in the sum is akin to a kernel centred at $Y_{N+1}$ as in the traditional kernel density estimate, but the main difference is that the kernel is adaptive as it depends on $P_N$. 

The nonparametric MP based on (\ref{eq:copula_update}) faces a few challenges. 
Firstly, estimating a probability density constrains the update due to the need to integrate to 1.
Secondly, although extensions to conditional density estimation are provided in \cite{Fong2023a}, it is challenging to incorporate structure in the regression setting (e.g. linearity), due to the stringent c.i.d. condition. Finally, studying the asymptotic properties of the nonparametric MP based on the copula is challenging, due to working in the space of probability measures \citep{Berti2004}. We will see that the QMP alleviates these challenges faced by the nonparametric MP outlined in the previous section as the space of quantile function estimates is much easier to handle.\vspace{2mm}

\begin{figure}[ht]
\small
  \centering
  \begin{minipage}{.44\linewidth}
\begin{algorithm}[H]
{Compute $P_n$ from the observed data $Y_{1:n}$}\\
  \For{$b \gets 1$ \textnormal{\textbf{to}} $B$}{
  \For{$i \gets n+1$ \textnormal{\textbf{to}} $N$} {
   Sample $Y_{i}  \sim {P}_{i-1}$\\
   Update  $P_{i} \mapsfrom \left\{P_{i-1}, Y_{i}\right\} $
  }
  Evaluate   ${\theta}^{(b)}_N =\theta(Y_{1:N})$ or  $\theta(P_N)$ }
 {Return $\{\theta_N^{(1)},\ldots,\theta_N^{(B)} \}$}
\caption{Predictive resampling}\label{alg:predictive_resampling}
\end{algorithm}
  \end{minipage}\hspace{2mm}
    \begin{minipage}{.54\linewidth}
\begin{algorithm}[H]
{Compute $Q_n$ from the observed data $Y_{1:n}$}\\
  \For{$b \gets 1$ \textnormal{\textbf{to}} $B$}{
  \For{$i \gets n+1$ \textnormal{\textbf{to}} $N$} {
   Sample $V_{i}  \sim \mathcal{U}(0,1)$; compute $Y_i = Q_{i-1}(V_i)$\\
   Update  $Q_{i} \mapsfrom \left\{Q_{i-1}, Y_{i}\right\} $
  }
  Evaluate   ${\theta}^{(b)}_N =\theta(Y_{1:N})$ or   $\theta(Q_N^{\dagger})$  }
 {Return $\{\theta_N^{(1)},\ldots,\theta_N^{(B)} \}$}
\caption{Quantile predictive resampling}\label{alg:quantile_predictive_resampling}
\end{algorithm}
  \end{minipage}
\end{figure}

\subsection{Quantile predictive resampling}\label{sec:PR}
In this section, we introduce the \textit{quantile martingale posterior} framework, which builds on the ideas of \cite{Fong2023a} to address quantile estimation. The core idea is to utilize a recursive update for an estimate of the quantile function, which serves as our predictive imputation machine. For now, assume that we have an estimate of the quantile function, $Q_n:(0,1) \to \R$, computed from the i.i.d. observations $Y_{1:n}$. We will address how to obtain $Q_n$ later, and will assume that $Q_n$ is continuous and bounded, but not necessarily monotonic. Given $Q_n$, consider the following sampling scheme:
\begin{enumerate}
    \item Simulate $V_{n+1} \sim \mathcal{U}(0,1)$
    \item Compute $Y_{n+1} = Q_n(V_{n+1})$.
\end{enumerate}
Viewed in this manner, $Q_n$ is simply a tool for simulating $Y_{n+1}$, and can thus be viewed as a \textit{generative predictive sampler}. This is analogous to the approach of the generative adversarial network \citep{Goodfellow2020}, where accurate samples are generated by passing noise through a neural network instead of estimating the density. It is also not challenging to see that $Y_{n+1}$ is in fact distributed according to $P_n$ with the corresponding rearranged quantile function $Q^{\dagger}_n$, which is indeed monotonic. The quantile function estimate $Q_n$ thus provides us a means to simulate from the rearranged predictive distribution directly, without the need to actually compute the rearrangement operator (\ref{eq:rearrangement}). This procedure is illustrated in {Figure \ref{fig:rearrangement} (left)}, where we draw $V_{n+1} \sim \mathcal{U}(0,1)$ and read off the corresponding value $Q_n(V_{n+1})$ to get a sample. The quantile function  and CDF of $Y_{n+1}$ is then $Q_n^{\dagger}$ and $P_n$, as shown in red in {Figures \ref{fig:rearrangement} (left)} and {\ref{fig:rearrangement} (right)} respectively. We will also refer to $Q_n^{\dagger}$ and $P_n$ as the \textit{implicit} quantile function and CDF respectively.
We provide more intuition as to what rearrangement implies for the resulting QMP in Section \ref{sec:theory}.

Given the further specification of a recursive update $\left(Q_n,Y_{n+1}\right) \to Q_{n+1}$, and assuming appropriate conditions on the update, we will then have all the ingredients needed to sample from the QMP, which is outlined in {Algorithm \ref{alg:quantile_predictive_resampling}}. The main difference to the original MP is that we keep track of a quantile function estimate, which can be interpreted as a generative predictive sampler, and it does not need to satisfy the monotonicity property. For now, we leave the update unspecified, but we will investigate the appropriate elicitation of the update function in detail starting in Section \ref{sec:recursive}.  Compared to the original MP, the class of possible predictives for the QMP is much broader, as we only require $Q_n$ to be bounded and continuous, whereas the original MP requires estimating a probability density function. We will see in Section \ref{sec:theory} that this relaxation allows for comprehensive theoretical study of the QMP, and 
Section \ref{sec:quantreg} will illustrate the simplicity of incorporating covariates for conditional quantile estimation.\vspace{5mm}

\begin{figure}[!h]
\begin{center}
\includegraphics[width=0.9\textwidth]{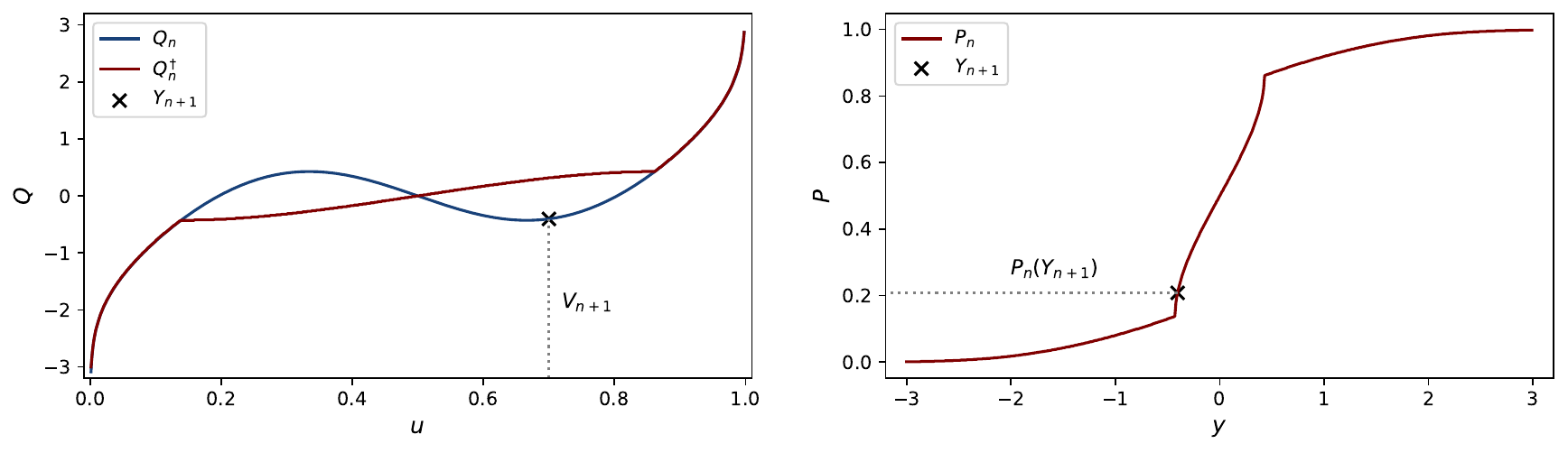}
\end{center}
\caption{Plot of  (Left) $Q_n$ and rearranged $Q_n^{\dagger}$ and (Right) implicit $P_n$; for  $V_{n+1} =0.7$ which gives $Y_{n+1}\approx -0.4$. Note that $P_n(Y_{n+1}) \neq V_{n+1}$, although $P_n(Y_{n+1}) \overset{d}{=} V_{n+1}$ as both are distributed according to $\mathcal{U}(0,1)$.} 
\label{fig:rearrangement}
\end{figure}

\subsection{Martingale condition and coherence}\label{sec:martingale}
In order for the QMP to be well-specified under the scheme of {Algorithm \ref{alg:quantile_predictive_resampling}}, we will require an analogous condition to the c.i.d. property for the original nonparametric MP. Unsurprisingly, we find that a martingale condition is once again sufficient for existence of the MP, which corresponds to an interesting coherence property on the generative predictive.

While we will leave the technical details for Section \ref{sec:theory}, we briefly outline the martingale condition here. In particular, we require a similar condition on the {estimate} of the quantile:
\begin{align}\label{eq:martingale}
   \E\left[Q_{N+1}(u)  \mid Y_{1:N} \right] = Q_N(u)
\end{align}
for each $u\in (0,1)$ for all $N \geq n$. Here, the conditional expectation is over $Y_{N+1} = Q_N(V_{N+1})$, so we are averaging over the r.v. $V_{N+1}\sim \mathcal{U}(0,1)$. Under assumptions on the recursive update, we show in Section \ref{sec:theory} that the limiting empirical distribution of $Y_{n+1:\infty}$ converges to some $P_\infty$ weakly almost surely, which has a corresponding random quantile function
 $Q_\infty^{\dagger}$. This kind of convergence also has close connections to exchangeability.
The QMP is then the distribution of $Q_\infty^{\dagger}$ or $P_\infty$ (or appropriate functionals thereof). The theory requires technical tools from the function-valued martingales and rearrangement operator literature, but intuitively, the above weak convergence implies that the QMP over the unknown quantile function exists. Furthermore, we will see that the additional flexibility gained in working with quantile functions instead of CDFs will allow us to quantify the convergence of $Q_N^{\dagger}$ to $Q_{\infty}^{\dagger}$ more precisely.

Previously, \cite{Fong2023a} highlighted that the c.i.d. condition was equivalent to predictive coherence, as the posterior mean of the predictive CDF $P_\infty(y)$ is equal to the initial estimate $P_n(y)$. In the QMP case, we will instead have a kind of \textit{generative} coherence. To interpret this, suppose we are interested in drawing a sample from $Y \sim P_{\infty}$. One can draw $V \sim \mathcal{U}(0,1)$ and plug it into the limiting generative predictive $Y= Q_\infty(V)$, which then gives $Y \sim P_\infty$. From (\ref{eq:martingale}), we have that $\E\left[Y \mid Y_{1:n}\right] = Q_n(V)$,
which suggests that the posterior mean of a sample from $P_{\infty}$ is equal to a sample from $P_n$ almost surely. We thus have not introduced any bias in samples from $P_n$ through our recursive update, which amounts to a generative coherence property.

\section{Recursive quantile estimator}\label{sec:recursive}
\subsection{Stochastic approximation}
We now introduce a novel recursive update to estimate continuous quantile functions. Recursive updates are particularly well-suited for the QMP, as it gives us both a means for predictive resampling and for ensuring the necessary martingale condition, which will we discuss in depth shortly. The motivation is based on the connection between recursive methods and \textit{stochastic approximation} \citep{Lai2003}. \cite{Hahn2018, Fong2023a} highlight the interpretation of (\ref{eq:copula_update}) as a stochastic approximation of the CDF/density, and the parametric MP of \cite{Walker2022,Holmes2023}
relies on a stochastic gradient descent approach to update the parameter $\theta_N$.

We take a similar approach here, leveraging a stochastic approximation estimate of the quantile function, which has also been investigated in works such as \cite{Aboubacar2014,Kohler2014} and \cite{Chen2023} in the non-Bayesian setting. One can define the quantile at $u \in (0,1)$ as
 $Q^*(u) = \argmin_q \int \rho_u(y - q) \, dP^*(y)$
where $\rho_u(z) =  z\left(u - \mathbbm{1}\left(z \leq 0 \right) \right)$ is the familiar check loss. Although the check loss is not differentiable at $z = 0$, one can still utilize the sub-gradient, and define the recursive update
\begin{align}\label{eq:recursive_quantile_freq}
    Q_{n+1}(u) = Q_{n}(u) + \alpha_{n+1}\left[u - \mathbbm{1}\left(Y_{n+1} \leq Q_{n}(u)\right)\right]
\end{align}
where $\alpha_n$ is a sequence of decreasing weights chosen so that 
\begin{align}\label{eq:alpha_conditions}
    \sum_{i = 1}^\infty \alpha_i = \infty, \quad \sum_{i = 1}^\infty \alpha_i^2 < \infty
\end{align}
as is standard in stochastic approximation.
One can show that this is indeed a consistent estimator under some assumptions, as the second condition on $\alpha_i$ ensures the algorithm converges, and the first condition ensures initial conditions are forgotten and we converge to the minimizer. 

There are however two main issues with (\ref{eq:recursive_quantile_freq}) that cause it to be unsuitable for the QMP, which we address now. Firstly, we are interested in the case where $Q^*$ is continuous, whilst (\ref{eq:recursive_quantile_freq}) will recover a discontinuous estimate of the quantile. Secondly, for the purposes of the QMP, there is the subtle but important point that (\ref{eq:recursive_quantile_freq}) does not imply a martingale for $Q_n(u)$ under the quantile predictive resampling, which will be important for showing the existence of the QMP.

\subsection{Recursive copula update}
We now describe a recursive estimate of the quantile function which returns both continuous curves and satisfies the required martingale condition. To begin, we highlight the connection between the recursive update of the predictive CDF based on the bivariate Gaussian copula as shown in (\ref{eq:copula_update}) and the empirical distribution and Bayesian bootstrap. The empirical distribution can be written recursively:
\begin{align*}
    P_{N+1}(y) = (1-\alpha_{N+1})P_N(y) + \alpha_{N+1}\mathbbm{1}\left(Y_{N+1} \leq y\right)
\end{align*}
where $\alpha_{N} = N^{-1}$. By comparing the above update to (\ref{eq:copula_update}), we see that the indicator term $\mathbbm{1}\left(Y_{N+1} \leq y\right)$ corresponds to the term $H_\rho\left(P_N(y), P_N(Y_{N+1})\right)$. In fact, we have that $\lim_{\rho \to 1} H_\rho(u,v) = \mathbbm{1}(v \leq u)$.
As a result, (\ref{eq:copula_update}) can be viewed as a smoothed version of the empirical distribution update.

Inspired by this connection, we apply the same intuition to extend (\ref{eq:recursive_quantile_freq}) into a form that is suitable for the QMP. Our suggested recursive update of the quantile function estimate is then
\begin{align}\label{eq:quantile_copula}
    Q_{N+1}(u) = Q_N(u) + \alpha_{N+1}\left[u - H_{\rho_{N+1}}\left(u, P_N(Y_{N+1})\right)\right],
\end{align}
where $P_N$ is the rearranged CDF function of $Q_N$, and $\alpha_N$ satisfies (\ref{eq:alpha_conditions}). We postpone discussion on the sequence $\rho_N \in (0,1)$ except for requiring that $\rho_N \to 1$ as $N$ increases, which is a key difference between the QMP and the regular MP, as the bandwidth $\rho$ is kept fixed in the latter.
Intuitively, the update (\ref{eq:quantile_copula}) is akin to a Bayesian analogue of a recursive kernel-smoothed quantile estimator (e.g. \cite{Aboubacar2014}) which arises naturally from a stochastic optimization viewpoint.
 {Figure \ref{fig:quant_copula} (left)} illustrates the form of $[u - H_\rho(u,v)]$ for increasing values of $\rho$, which we see approaches the limiting case $[u-\mathbbm{1}(v \leq u)]$. We can thus directly view $[u - H_\rho(u,v)]$ as a continuous relaxation of $[u-\mathbbm{1}(v \leq u)]$. {Figure \ref{fig:quant_copula} (right)} then illustrates the effect of updating with an observation with $P_N(Y_{N+1}) = v$.

\begin{figure}[!h]
\begin{center}
\includegraphics[width=0.9\textwidth]{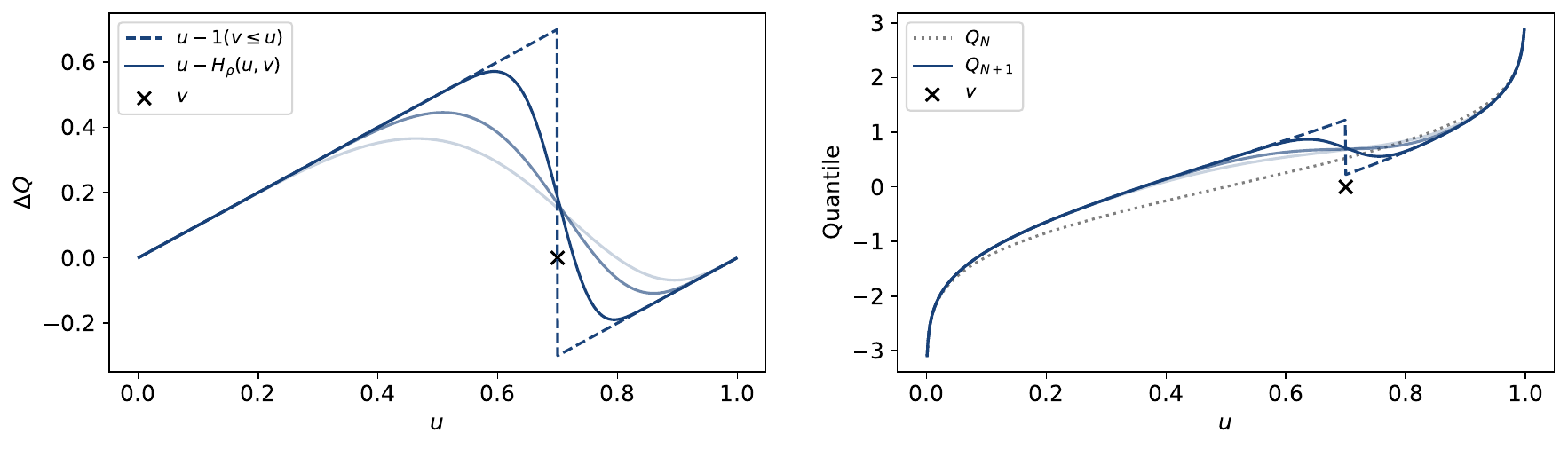}
\end{center}
\caption{Plot of  (Left)  $[u - H_\rho\left(u, v\right)]$ and (Right) updated quantile estimate $Q_{N+1}(u)$ and old $Q_{N}(u)$ (\dotted); for $v =0.7$ with $\rho = (0.9,0.95,0.99)$ (\fulllow,\fullmid,\full) and $[u - \mathbbm{1}(v \leq u)]$ (\dashed).  } 
\label{fig:quant_copula}
\end{figure}
Unlike the non-Bayesian case, much care is needed to ensure the coherence condition discussed in Section \ref{sec:martingale} is satisfied.
To this end, the rearrangement step is crucial for obtaining the martingale under predictive resampling, as $Q_N$ may not be monotonic. We highlight the key property that both  $[u - H_\rho(u,v)]$ and $[u-\mathbbm{1}(v \leq u)]$ are not monotonic, so it is possible for $Q_{N+1}$ to not be monotonic even if $Q_N$ is. This is illustrated in {Figure \ref{fig:quant_copula} (right)}, where for $\rho$  close to 1, we have non-monotonicity of the updated $Q_{N+1}$.

To understand the importance of rearrangement for the martingale condition, we focus on the step function case, and contrast between $\left[u - \mathbbm{1}(Y_{N+1} \leq Q_{N}(u))\right]$ versus $[u - \mathbbm{1}(P_N(Y_{N+1}) \leq u)]$, where the first case is from (\ref{eq:recursive_quantile_freq}) and the latter is from (\ref{eq:quantile_copula}) with $\rho \to 1$. If $Q_N$ is is a proper quantile function, i.e. it is motonically increasing and left-continuous, then we have $Y_{N+1} \leq Q_N(u) \Leftrightarrow P_N(Y_{N+1}) \leq u$. In this case, it is thus clear that the two updates are equivalent. However, when $Q_N$ is not monotonic, the two updates will differ. To see why the latter update is more suitable, consider $Y_{N+1} \sim P_N$ where $P_N$ is continuous. Under predictive resampling, we have $P_N(Y_{N+1}) \sim \mathcal{U}(0,1)$, so in the latter case we have
\begin{align*}
    \E\left[u - \mathbbm{1}(P_N(Y_{N+1}) \leq u) \mid Y_{1:N}\right] = u - \int_0^1 \mathbbm{1}\left(v \leq u\right)\, dv = 0.
\end{align*}
In the first case however, we have
$\E\left[u - \mathbbm{1}(Y_{N+1} \leq Q_{N}(u)) \mid Y_{1:N}\right]  = u - P_N\left(Q_N(u)\right) \neq 0.$
The issue arises as $P_N\left(Q_N(u)\right)\neq u$ when $Q_N$ is not monotonic, and the size of the deviation is related to how non-monotonic $Q_N$ is. Finally, the above logic extends to the smooth case, where one can show that
\begin{align*}
    \E\left[u - H_{\rho_{N+1}}(u,P_N(Y_{N+1})) \mid Y_{1:N}\right] = u - \int_0^1 H_{\rho_{N+1}}(u,v)\, dv = 0.
\end{align*}
This follows as $\int_0^{v'} H_\rho(u,v) \, dv = C_\rho(u,v')$ where $C_\rho$ is the bivariate Gaussian copula, and taking $v' \to 1$ returns $C_\rho(u,1) = u$. As a result, the recursive update (\ref{eq:quantile_copula}) satisfies the required martingale condition from Section \ref{sec:martingale} when $P_N$ is continuous. This once again highlights the bivariate copula as a versatile building block for Bayesian nonparametrics, especially for smooth functions.

\subsection{Posterior sampling from the QMP}
A nice property of the QMP is that rearrangement is automatically handled during predictive resampling. To see this, we revisit the quantile predictive resampling scheme, where $Y_{N+1} = Q_N(V_{N+1})$ for $V_{N+1} \sim\mathcal{U}(0,1)$, resulting in $Y_{N+1} \sim P_N$. The recursive quantile update only relies on $Y_{N+1}$ through $P_N\left(Y_{N+1}\right)$, 
and again we have $P_N\left(Y_{N+1}\right) \sim \mathcal{U}(0,1)$ if $P_N$ is continuous. 
To carry out one step of predictive resampling, it is then simply a matter of simulating $V_{N+1} \sim \mathcal{U}(0,1)$ and computing
\begin{align}\label{eq:quantile_copula_PR}
    Q_{N+1}(u) = Q_N(u) + \alpha_{N+1}\left[u - H_{\rho_{N+1}}\left(u, V_{N+1}\right)\right].
\end{align}
Once again, posterior sampling only depends on the simulation of uniform r.v.s, which is extremely cheap, and does not require complex MCMC schemes. In practice, the update (\ref{eq:quantile_copula_PR}) truncated at some reasonably large $N \gg n$ is sufficient for convergence to $Q_\infty$. Looking ahead, we will shortly see that a truncation may not even be necessary as we can identify the limiting law of $Q_{\infty} - Q_N$.

The advantages of predictive resampling over traditional Bayes is clear and outlined in \cite{Fong2023a}, which we now recap. Firstly, we can completely avoid issues of mixing and serial computation that faces MCMC, relying only on uniform r.v.s and simple computations to provide i.i.d. posterior samples. The update (\ref{eq:quantile_copula_PR}) is also particularly easy to parallelize, both across samples and across different values of $u$, allowing us to easily take advantage of modern GPU compute.  Finally, the only source of approximation comes from a relatively harmless truncation step.
Interestingly, sampling from the QMP has additional advantages over the regular MP due to working in the space of quantile function estimates. The first is that samples from $P_{\infty}$ can be obtained directly when quantile predictive resampling, as we are working with a generative predictive. Specifically, if we replace $u$ in (\ref{eq:copula_update}) with $U\sim \mathcal{U}(0,1)$, then computing the recursive update will transform $Y_N = Q_N(U) \sim P_N$ into $Y_{N+1} = Q_{N+1}(U) \sim P_{N+1}$. In practice, we can thus initialize a vector $U_{1:T}\iid \mathcal{U}(0,1)$ and pass it through the update (\ref{eq:quantile_copula_PR}) for $N\geq n+1$. 
A second benefit is that theoretical study will be more comprehensive compared to the regular MP, as the QMP only requires $u$ and $V_{N+1}$ as inputs for the update. In particular, we will be able to show a weak convergence result which allows even faster approximate sampling from the QMP based on a Gaussian process (GP). 

\subsection{Initial estimate $Q_n$}\label{sec:initial}
Up until now, we have not discussed how one would obtain the initial estimate $Q_n$ from the observed data $Y_{1:n}$, from which predictive resampling begins. We emphasize here that $n$ is the number of observed i.i.d. samples from $P^*$ or $Q^*$, whereas $N$ is used to index predictively resampled future samples. Of course, $Q_n$ is of utmost importance as it governs the central tendency of the QMP. 
Following \cite{Fong2023a}, it is the most coherent if $Q_n$ is obtained by applying the update (\ref{eq:quantile_copula}) to the i.i.d. observations $Y_{1:n}$, starting from some initial $Q_0$ (e.g. the quantile function of $\mathcal{U}[a,b]$ where $[a,b]$ depends on the dataset). The entire statistical model is then governed by the update (\ref{eq:quantile_copula}) and $Q_0$, which is closely connected to the prequential framework of \cite{Dawid1984}. One can then interpret Bayesian inference as applying the update (\ref{eq:quantile_copula}) up until the final observation $Y_n$, then imputing $Y_{n+1:\infty}$ from $Q_n$ once we `run out' of observed data points.

There is however a slight intricacy specific to the quantile estimation case, which suggests that a variant of (\ref{eq:quantile_copula}) may be more desirable when estimating $Q_n$ from i.i.d. observations. Consider the \textit{rearranged} update:
\begin{align}\label{eq:rearr_quantile_copula}
    Q_{i+1}(u) &= Q^{\dagger}_i(u) + \alpha_{i+1}\left[u - H_{\rho_{i+1}}\left(u, P_i(Y_{i+1})\right)\right],
\end{align}
for $i = 1,\ldots,n-1$. The key difference is that we require an additional rearrangement of $Q_i$ after each update. We will see in Section \ref{sec:consistency} that applying the update (\ref{eq:rearr_quantile_copula})  gives us consistency of $Q_n^{\dagger}$ at $Q^*$, which we have been unable to show for (\ref{eq:quantile_copula}). The intuition for this discrepancy lies in the stochastic gradient descent interpretation of the update. 
In both (\ref{eq:quantile_copula}) and (\ref{eq:rearr_quantile_copula}), the `gradient' is computed at the rearranged version of $Q_i^{\dagger}$ (through $P_i$), but the gradient update is applied to $Q_i$ instead of $Q_i^{\dagger}$ in (\ref{eq:quantile_copula}), which may impede consistency. 

Under predictive resampling however, the non-linear rearrangement step in (\ref{eq:rearr_quantile_copula}) would cause $Q_N$ to no longer be a martingale, which is undesirable for the QMP. As a result, we recommend using the update $(\ref{eq:rearr_quantile_copula})$ for the i.i.d. observations to obtain the initial  $Q_n^{\dagger}$, then carrying out predictive resampling with $(\ref{eq:quantile_copula_PR})$ for imputing $Y_{n+1},Y_{n+2},\ldots$. Under this scheme, $(\ref{eq:rearr_quantile_copula})$ will ensure frequentist consistency while $(\ref{eq:quantile_copula_PR})$ will ensure that the QMP exists under predictive resampling. This slight incoherency appears to be the price that we pay for working with quantile function estimates, which are well known to have issues related to monotonicity of estimates as we discussed. Fortunately, we find that in practice there is not too much difference in the estimated $Q_n^{\dagger}$ obtained through (\ref{eq:quantile_copula_PR}) or (\ref{eq:rearr_quantile_copula}) as long as $\rho_N$ is chosen to not approach $1$ too quickly (which we discuss shortly), so the above concern is perhaps more theoretical in nature.

\subsection{Algorithm}
We now
summarize the QMP method, and  postpone the setting of $\rho_i$ and $\alpha_i$ and approximate sampling to Section \ref{sec:implications}.
{Algorithms \ref{alg:fit}} and {\ref{alg:QMP}} below illustrate the full process of obtaining the QMP. Like with the regular MP, there is a distinct separation of estimation and obtaining uncertainty, which is more akin to frequentist methods. In practice, it may be desirable to average the output of {Algorithm \ref{alg:fit}} over multiple permutations of the data (e.g. $10$) if it is desirable for the initial estimate of the QMP $Q_n$ to be permutation invariant. Due to the expediency of the update, this is not too restrictive computationally, and no permutation-averaging is required for predictive resampling due to asymptotic exchangeability (discussed in Section \ref{sec:prob}).
We will require a grid of $u$-values on which we compute the quantile estimates, and this also governs the `resolution' of our samples. We find that a grid of $200$ evenly spaced points from $[0,1]$ works well in practice.
The number of future samples $N$ can be set by monitoring the convergence of $Q_N$, and we see that $N \approx n + 5000$ is sufficient in practice. {Algorithm \ref{alg:QMP}} can be easily executed on a GPU,  as sampling consists of many simple operations which can be computed in parallel. However, we will see in Section \ref{sec:approx} than {Algorithm \ref{alg:QMP}}
can be approximated even more quickly using a GP. 

\begin{figure}[!ht]
\small
  \centering
  \begin{minipage}{.44\linewidth}
\begin{algorithm}[H]
{Initialize $Q_0$}\\
{Data is $Y_1,\ldots,Y_n$}\\
\For{$i \gets 1$ \textnormal{\textbf{to}} $n$ }{
Compute $V_i = P_{i-1}(Y_{i})$\\
{$Q_{i}(u) = Q^{\dagger}_{i-1}(u) + \alpha_{i}\left[u - H_{\rho_{i}}\left(u, V_i\right)\right]$}
}
 {Return $Q_n^{\dagger}$}
\caption{Estimation of quantile function}\label{alg:fit}
\end{algorithm}
\end{minipage}
\hspace{3mm}
\begin{minipage}{.52\linewidth}
\begin{algorithm}[H]
{Initialize $Q_n$ from {Algorithm \ref{alg:fit}}}\\
\For{$b \gets 1$ \textnormal{\textbf{to}} $B$}{
\For{$i \gets n+1$ \textnormal{\textbf{to}} $N$}{
{Draw $V^{(b)}_i \sim \mathcal{U}(0,1)$}\\
{{$Q^{(b)}_{i}(u) = Q^{(b)}_{i-1}(u) + \alpha_{i}\left[u - H_{\rho_{i}}\left(u, V^{(b)}_i\right)\right]$}   
}}}
 {Return $\left\{{Q_N^{\dagger}}^{(1)},\ldots,{Q_N^{\dagger}}^{(B)}\right\}$}
\caption{QMP sampling}\label{alg:QMP}
\end{algorithm}
\end{minipage}
\end{figure}\vspace{-5mm}

\section{Theory}\label{sec:theory}

For the original MP, asymptotic theory was challenging due to the complex dependence in the update. Interestingly, the lack of dependence on the predictive of the first input into $H_\rho(u,v)$ helps to simplify the theory. We distinguish between two asymptotic regimes under the MP framework. The first is the convergence of $P_N \to P_\infty$ from predictive resampling, starting at $N = n+1$, which we term \textit{predictive} asymptotics. This is closely connected to Doob's consistency theorem \citep{Doob1949}, and is discussed nicely in \cite{Fortini2024}. The second is the classical \textit{frequentist} asymptotics, where we study the convergence of the MP or relevant estimates (such as $P_n$) as $n \to \infty$, where $n$ is the number of i.i.d. observations $Y_{1:n} \iid P^*$.  We will now investigate both for the QMP. Full derivations are postponed to the Appendix, although we provide proof outlines when they are particularly insightful. 

\subsection{Predictive asymptotics}
To study the predictive asymptotics of the QMP, we will rely on the theory of \textit{function-valued} martingales \citep{Pisier2016}. Although the theory is  technical, the results and conditions are  insightful and simple to interpret.
We begin this subsection with prerequisite theory from functional analysis, with details deferred to Section \ref{app:sec_prereq} in the Appendix. As the space of possible  of quantile function {estimates} $Q_N$ is quite large due to not requiring monotonicity, we will have sufficient structure to borrow powerful results from functional analysis. Let $B$ be a Banach space of real-valued functions $f:(0,1) \to \R$ with norm $\|\cdot \|_B$, which $Q_N$ will belong to. In particular, we will work with two very useful spaces that lend themselves to easy study of recursive updates for $Q_N$. The first is the $L^2((0,1))$ space, which consists of square-integrable functions with norm $\|f\|_{2} = \sqrt{\int f(u)^2 \, du}$. 
We write the $L^2$ distance between two elements $f, g\in L^2((0,1))$ as $d_2(f,g) = \|f-g\|_2$. The second is the Sobolev space $H^1((0,1))$ consisting of functions $f \in L^2((0,1))$  which are weakly differentiable with weak derivative $f' \in L^2((0,1))$, which shares properties with the regular derivative. 
A very useful property in the 1-dimensional case is that if $f \in H^1((0,1))$, then $f$ is equal almost everywhere to an absolutely continuous function. 
The norm in the Sobolev space $H^1((0,1))$ is then
$\|f\|_{1,2} = \sqrt{\| f\|_2^2 + \|f'\|_2^2}$, with corresponding distance $d_{1,2}(f,g) = \|f-g\|_{1,2}$. Both $L^2$ and $H^1$ are Hilbert spaces, which will allow us to apply function-valued martingale convergence theorems easily.

Through {Algorithm \ref{alg:QMP}}, $Q_N$ will evolve randomly, so we require a probability space on $B$-valued objects. Let $(\Omega,\mathcal{F},\mathbb{P})$ denote the probability space. A r.v. in this case is a function $f: \Omega \to B$ which is Bochner measurable and takes values in $B$, so realizations of the r.v. are functions, that is $f(\omega)\in B$ for $\omega \in \Omega$. We write $L^p\left(\Omega, \mathcal{F}, \mathbb{P}; B\right)$ or $L^p(B)$ as the space of  Bochner measurable functions with $\E\left[\|f\|_B^p\right] =\int \|f\|_B^p d\mathbb{P} < \infty$ for some $1 \leq p < \infty$, where we will mostly be using $p = 2$. The norm of this space is defined as $\|f\|_{L^p(B)} = \left(\E\left[ \|f\|_B^p \right]\right)^{1/p}$, and functions that are equal a.e. are identified. In our use cases, expectations within this  space can be evaluated pointwise on the function, so
the condition (\ref{eq:martingale}) is enough to ensure $Q_N$ is a function-valued martingale. Details regarding (conditional) expectations are in Section \ref{app:sec_ban_mart} of the Appendix. 

\subsubsection{Existence and support of the QMP}
 We now study the convergence of the sequence $Q_{n+1}, Q_{n+2},\ldots$ under quantile predictive resampling with {Algorithm \ref{alg:QMP}}, which will inform us on properties of the QMP. The main theorem we will use is the  convergence theorem for Banach space valued martingales, which we cover in detail in Section \ref{app:sec_ban_mart} of the Appendix.
 
We will need the following assumptions on $Q_n:(0,1)\to \R$, which is the initial estimate of the quantile function that we predictive resample from, as well as an assumption on the copula update.

\begin{assumption}[Bounded in $L^2$]\label{as:L2}
    $Q_n$ satisfies $\|Q_n\|_2 < \infty$.
\end{assumption}
\begin{assumption}[Weak derivatives bounded in $L^2$]\label{as:weak_deriv} 
$Q_n$ is weakly differentiable with weak derivative $q_n$ which satisfies  $\|q_n\|_2 < \infty$, so $\|Q_n\|_{1,2} <\infty$.
\end{assumption}
\begin{assumption}[Learning rate]\label{as:alpha}
    The learning rate sequence takes the form $\alpha_i= a(i+1)^{-1}$ for some $a \in (0,\infty)$ for $i \geq 1$.
\end{assumption}
\begin{assumption}[Bandwidth]\label{as:bandwidth}
    The bandwidth sequence takes the form  $\rho_i =\sqrt{1-ci^{-k}}$  where $0 < k < 1$ and $0 < c< 1$ for $i \geq 1$. 
\end{assumption}

Intuitively, {Assumptions \ref{as:L2}} and {\ref{as:weak_deriv}} ensure that the initial sampler $Q_n$ is sufficiently well-behaved. {Assumption \ref{as:alpha}} satisfies (\ref{eq:alpha_conditions}) which is standard for stochastic approximation. {Assumption \ref{as:bandwidth}} ensures that $\rho_N$ does not approach $1$ too quickly, i.e. the smoothness of the update function does not decrease too quickly.
 
\begin{proposition}\label{prop:L2_mart}
   Under {Assumptions \ref{as:L2}} and {\ref{as:alpha}}, there exists a random function $Q_\infty$ with realizations in $L^2((0,1))$   such that 
    $d_2(Q_N,Q_\infty) \to 0$ a.s. 
    
\end{proposition}
\begin{proof_outline}
    We rely on the martingale convergence theorem for Banach spaces as given in {Theorem \ref{app:thm_hilbert_mart}} in the Appendix. By construction, we have (\ref{eq:martingale}) so $Q_N$ is a martingale. The main condition to check is that $\sup_{N \geq n}\E\left[\|Q_N\|_{2}^2\right] < \infty$, which is detailed in the Appendix.
\end{proof_outline}
The above proposition thus guarantees the existence of the QMP, which is the distribution of $Q_\infty$. Under relatively weak constraints on the predictive update, we can say much more about the support of the QMP.

\begin{theorem}\label{th:sobolev_mart}
     Under {Assumptions \ref{as:L2}}-{\ref{as:bandwidth}}, there exists a random function $Q_\infty$ with realizations in $H^1((0,1))$ such that $d_{1,2}\left(Q_N,Q_\infty\right) \to 0$ a.s. 
\end{theorem}
\begin{proof_outline}
The key here is that the {Assumption \ref{as:bandwidth}} on the bandwidth prevents the expected Sobolev norm from diverging to infinity, i.e. $\sup_{N \geq n} \E\left[\|Q_N\|^2_{1,2}\right]<\infty$. This allows us to apply {Theorem \ref{app:thm_hilbert_mart}} as we did in {Proposition \ref{prop:L2_mart}}.
\end{proof_outline}
\begin{corollary}\label{cor:continuous}
Under {Assumptions \ref{as:L2}}-{\ref{as:bandwidth}}, realizations of $Q_\infty$ are absolutely continuous on $(0,1)$ a.s., up to the equivalence class of $H^1((0,1))$. 
\end{corollary}
In other words, the above theorem and corollary implies that samples of $Q_\infty$ from the QMP are  absolutely continuous and thus differentiable almost everywhere a.s. We have thus managed to identify the support of the QMP by leveraging the Sobolev space, which is crucial if absolute continuity of the quantile function estimate is desired. However, we have only studied the quantile estimate $Q_N$, which may not be monotonic. Since the actual object of interest is the implicit quantile function or CDF $Q_N^{\dagger}/P_N$, the question is whether we can say anything about the QMP distribution over those. Fortunately the answer is yes, due to the regularizing effect of the rearrangement operator. To first study the convergence of $Q_N^{\dagger}$, we will need the following well-known proposition on rearrangement:
\begin{proposition}[\cite{Lorentz1953, Chernozhukov2009}]\label{prop:rearr}
    Let $f,g$ be any two functions $[0,1] \to C$ for some bounded subset $C \subset \R$ with increasing rearrangements $f^{\dagger},g^{\dagger}$ respectively. We then have 
$d_2(f^{\dagger},g^{\dagger}) \leq  d_2(f,g)$.\end{proposition}
Consider the case where $g^{\dagger} = Q^*$ is a proper quantile function. The above proposition then states that the rearrangement of $Q_N$ to $Q_N^{\dagger}$ can only improve the estimate \citep{Chernozhukov2009}. Furthermore, the rearrangement procedure does not hurt the smoothness of the function, which implies the following result.
\begin{proposition}\label{prop:rearr_sobolev}
    Under {Assumptions \ref{as:L2}}-{\ref{as:bandwidth}}, there exists a random function $Q^{\dagger}_\infty$ with realizations in $H^1((0,1))$ such that 
$d_2(Q^{\dagger}_N,Q^{\dagger}_\infty) \to 0$ a.s, where realizations of $Q_N^{\dagger}$ and $Q_\infty^{\dagger}$ are proper monotonically increasing quantile functions. 
\end{proposition}
We highlight that $Q_{\infty}^{\dagger}$  is absolutely continuous on $(0,1)$ up to the equivalence class of $H^1((0,1))$ a.s., which follows from the well-known property that rearrangement preserves absolute continuity ({Theorem \ref{app:thm_sobolev}} in the Appendix).
We thus have that the predictive quantile function $Q_N^{\dagger}$ converges to an absolutely continuous random quantile function $Q_{\infty}^{\dagger}$ a.s. In other words, the rearrangement operator does not significantly affect the  predictive asymptotics of the QMP, although we highlight that in general the posterior mean of the QMP is not $Q_n^{\dagger}$.
Note that the convergence is only in the $L^2$ norm, as strengthening the convergence to hold in the Sobolev norm is trickier. A technical result states that a variant of the rearrangement operator is continuous in the Sobolev space $H^1((0,1))$ \citep{Coron1984,Almgren1989}, so it is likely possible for our case. However, this is stronger than what we require, so we leave this for future work.

\paragraph{\textbf{Summary}}
Although the results are technical, the intuition is hopefully clear.
We have utilized the function-valued martingale convergence theorem to show that the quantile estimate $Q_N$ converges a.s. (in the norm of the respective Banach space) to a random $Q_\infty$ under {Algorithm \ref{alg:QMP}}. The regularizing behaviour of the rearrangement operator then assures us that the {implicit} proper quantile functions $Q_N^{\dagger}$ also converge in $L^2$ a.s. to a random proper quantile function $Q_\infty^{\dagger}$. This guarantees the existence of the QMP, which is precisely the distribution of $Q_\infty^{\dagger}$. 
Under additional smoothness assumptions on $Q_n$ and the bandwidth sequence $\rho_N$, we can then leverage the Sobolev space to show that the support of the QMP is on proper quantile functions which are absolutely continuous on $(0,1)$.

\subsubsection{QMP over probability measures and asymptotic exchangeability}\label{sec:prob}
So far, we have only been working in the quantile space, but it is interesting to study the QMP on the more familiar space of probability measures. This will also allow us to make  statements on the convergence of limiting probability distributions as studied in \cite{Berti2004} and \cite{Fong2023a}. Fortunately, we can leverage a simple connection between $L^2$ convergence of quantile functions and weak convergence.
\begin{proposition}
    Under {Assumptions \ref{as:L2}} and {\ref{as:alpha}}, there exists a random probability measure $P_\infty$ on $\R$ such that $P_N{\to} P_\infty$ in Wasserstein-2 distance a.s., which further implies $P_N{\to} P_\infty$ weakly a.s. Under the additional {Assumptions \ref{as:weak_deriv}} and {\ref{as:bandwidth}}, $P_\infty$ corresponds to an absolutely continuous $Q_\infty^{\dagger}$. 
\end{proposition}
\begin{proof}
    As $d_2(Q_N^{\dagger},Q_\infty^{\dagger})$ is exactly the Wasserstein-2 distance between $P_N$ and $P_\infty$, where $P_\infty$ is computed from  $Q_\infty$, we have that $P_N \to P_\infty$ in Wasserstein-2 distance a.s. As the Wasserstein distance metrizes weak convergence in $\R$ (e.g. \cite[Theorem 6.8]{Villani2009}), we have the above result. 
\end{proof}
As we have a.s. weak convergence of $P_N$ to a random probability measure, we can make the usual statements on limiting empirical distributions and asymptotic exchangeability. 
\begin{corollary}\label{cor:exchange}
     Under {Assumptions \ref{as:L2}} and {\ref{as:alpha}}, the sequence $(Y_{n+1},Y_{n+2},\ldots)$ arising from {Algorithm \ref{alg:QMP}} is asymptotically exchangeable. Furthermore, the empirical distribution of $(Y_{n+1},Y_{n+2},\ldots, Y_N)$ converges weakly to $P_\infty$ a.s. as $N \to \infty$. 
\end{corollary}
Both results in the above corollary arise due to the a.s. weak convergence of $P_N$ to some $P_\infty$, which implies that this convergence of the predictive distribution is sufficient for Bayesian inference  \citep{Fong2023a,Cui2023}. The QMP distribution over any functional is then simply the push-forward of $\theta(P_\infty)$. The c.i.d. condition is a very convenient means to attain this convergence, but the above two results highlight that it is by no means necessary. Unlike in \cite{Fong2023a}, which relies on the c.i.d. condition, here we instead rely on a martingale condition on the potentially non-monotonic quantile estimate.
A keen reader may notice that we have not assured absolute continuity on the probability measure $P_\infty$, which would then imply the existence of a probability density function. Unfortunately the absolute continuity and non-strict monotonicity of $Q_\infty^{\dagger}$ is not enough to guarantee this, as any flat regions of $Q_\infty^{\dagger}$ could be mapped to an atom for $P_\infty$. However, absolute continuity of $Q_\infty^{\dagger}$ allows us to guarantee that $P_\infty$ does not have any gaps in its support, and in practice we also see that $P_\infty$ is continuous a.s.

\subsubsection{Gaussian process}
Having established the existence of $Q_\infty$ which is distributed according to the QMP, a natural question is to investigate the properties of $Q_\infty - Q_N$ as we take $N \to \infty$ in {Algorithm \ref{alg:QMP}}. This is closely related to the study carried out in \cite{Fortini2020,Fortini2023,Fortini2024}, but we will require some technical tools from empirical process theory as we would like to study the entire function $Q_\infty$. One surprising consequence of the theory to come is the simplicity of the law of $Q_\infty - Q_N$, which allows us to accelerate  sampling from the QMP even further. We now introduce the results before discussing their implications.

For the rest of this section, we will assume that
$\alpha_N$ takes the form given in {Assumption \ref{as:alpha}}. To begin, we first discuss the object of study. We will focus on quantifying the convergence of $Q_N$ to $Q_\infty$, as this is much more tractable than the rearranged case. Specifically, we are interested in the law of the \textit{random} function $S_N = Q_\infty - Q_{N-1}$ as $N \to \infty$, which we suspect to be Gaussian due to the summative form of (\ref{eq:quantile_copula}). More concretely, let us define the random function 
\begin{align}\label{eq:SN}
    S_N(u) = Q_\infty(u) - Q_{N-1}(u) = \sum_{i = N}^\infty \alpha_i \left(u - H_{\rho_i}(u,V_i)\right)
\end{align} 
 where $V_i \iid \mathcal{U}(0,1)$ for all $i\geq N$. We highlight to the reader again that $S_N$ has an   additive form and in particular consists of  a sum of independent terms. As an aside, one concern may be that the distribution of $S_N$ does not depend on observed data (through $Q_n$). However, we can quell these concerns by drawing a connection to the Bayesian bootstrap, where the random Dirichlet weights $w_{1:n}$ do not depend on the data at all, but the \textit{location} of observations contribute to the posterior. In the QMP case, $S_N$ plays the role of the Dirichlet weights, and the initial function $Q_n$ plays the role of the observations' locations.

This independent form of $S_N$ is in fact a strength of the QMP compared to the traditional MP, as it allows us to much more easily leverage central limit theorems for the sum of independent functions.
Armed with this, we can study the convergence of the whole function, which depends on technical empirical process theory that we defer to Section \ref{sec:emp} in the Appendix. In particular, the independent form of $S_N(u)$ allows us to easily verify an asymptotic tightness condition and marginal convergence to a Gaussian distribution using the Lindeberg-Feller central limit theorem (CLT), which gives the following result.
\begin{theorem}\label{th:gp}
    Under {Assumptions \ref{as:alpha}} and {\ref{as:bandwidth}}, the function $\sqrt{N}S_N$ converges weakly in $\ell^\infty\left((0,1)\right)$  to $\mathbb{G}_a$, where $\mathbb{G}_a$ is a zero-mean GP with covariance function
$\E\left[\mathbb{G}_a(u)\, \mathbb{G}_a(u^{\prime})\right] = a^2 (\min\{u,u^{\prime}\} - uu^{\prime})$. 
\end{theorem}
\begin{proof_outline}
  Asymptotic tightness of $\sqrt{N}S_N$ is shown in {Theorem \ref{th:tight}} in the Appendix.
  We also show in the Appendix that any finite collection of points of $\sqrt{N}S_N(u)$ converges to a Gaussian distribution using the Lindeberg-Feller CLT, which together with asymptotic tightness is sufficient for weak convergence to the GP.
\end{proof_outline}
This covariance function is $a^{2}$ times the Brownian bridge covariance, which is unsurprising as this arises in the asymptotics for traditional quantile estimation as well.
We conclude this section with a brief discussion of the implications of the above, and postpone a detailed demonstration for Section \ref{sec:implications}. Following \cite{Fortini2020}, we note that the above gives us a measure of contraction of $Q_N$ to $Q_\infty$, which is quantified by $\sqrt{N}$ term pre-multiplying $S_N$. More interesting for us however, is the ability to approximate {Algorithm \ref{alg:QMP}} with the above GP, which we dedicate Section \ref{sec:approx} to. A remaining question is whether the rearranged $Q_\infty^{\dagger}$ satisfies a similar result. Our conjecture is that it may hold, but it is challenging to extend the proof due to an issue of the centering function. Nonetheless, as we are primarily interested in posterior sampling, we can still utilize the asymptotic normality to sample $Q_\infty$ which then gives the implied $Q_\infty^{\dagger}$.

\subsection{Frequentist asymptotics}
We now address the frequentist properties of QMP, which requires a different set of technical tools, but relies on similar recursive arguments such as martingale theory. We will shortly see that posterior consistency and contraction rates can be shown for the QMP, where the $L^2((0,1))$ Hilbert space and rearrangement theory aid us greatly. The proofs depend critically on the consistency and the convergence rate of the initial $Q_n^{\dagger}$. However, the latter properties depend on somewhat more technical tools from the stochastic approximation literature. We hope to distinguish this in the discussion below.

To begin, we introduce the setup which differs to the previous subsection. Let $Y_{1:n} \iid P^*$ where $P^*$ has the corresponding quantile function $Q^*$, and we consider the case as $n \to \infty$. 
Following the discussion in Section \ref{sec:initial}, we study the frequentist properties of the QMP obtained through applying {Algorithm \ref{alg:fit}} to the i.i.d. observations $Y_{1:n}$ to obtain the initial $Q_n^{\dagger}$, followed by predictive resampling with {Algorithm \ref{alg:QMP}} in order to obtain $Q^{\dagger}_\infty$. The QMP is then the distribution of $Q^{\dagger}_\infty$ conditional on  $Y_{1:n}$.

\subsubsection{Posterior consistency}\label{sec:consistency}
Posterior consistency is a crucial property of a Bayesian model which in our context states that the posterior distribution concentrates on the true $Q^*$ from which the data is i.i.d. This is much stronger than Doob's consistency theorem, which only holds a.s. with respect to the prior and is closely connected to the previously discussed predictive asymptotics. Posterior consistency usually hinges on the Kullback-Leibler (KL) property of the prior distribution \citep[Chapter 6]{Ghosal2017}, which states that the prior allocates non-zero mass to a KL ball around the truth. Within the martingale posterior context, no such prior distribution exists, so we must develop novel tools for posterior consistency. \cite{Fong2023a} showed consistency of the posterior mean of the MP, but did not make any statements on the entire posterior distribution. We will now show this for the QMP case, which requires the following conditions.
\begin{assumption}[Lipschitz quantile function]\label{as:consistency_truth}
Assume that $P^*$ has a quantile function $Q^*$ which is $M$-Lipschitz continuous on $[0,1]$, where $M$ is a constant. Furthermore, $Q_0$ is chosen to be Lipschitz continuous. 
\end{assumption}
A sufficient condition for this is that $P^*$ has compact support, and $P^*$ is continuously differentiable with strictly positive derivative on its support (e.g. see \cite[Lemma 21.4]{vanderVaart2000}). We now have consistency of the initial $Q_n^{\dagger}$.
\begin{theorem}\label{th:consistency_mean}
    Under {Assumptions \ref{as:alpha}},  {\ref{as:bandwidth}} and  {\ref{as:consistency_truth}}, we have that $d_2( Q_n^{\dagger},Q^*) \to 0$ a.s.$[P^*]$ under {Algorithm \ref{alg:fit}}.
\end{theorem}
\begin{proof_outline}
    The proof has similar components to the proofs of consistency in \cite{Hahn2018} and \cite{Fong2023a}, but require additional tools specialized to quantile functions and rearrangement. We show that $d_2( Q_n^{\dagger},Q^*)$ is an \textit{almost supermartingale} in the sense of \cite{Robbins1971}. The bandwidth condition ensures that (\ref{eq:rearr_quantile_copula}) approaches a variant of the step update (\ref{eq:recursive_quantile_freq}). The condition $\sum \alpha_n^2 < \infty$ prevent the errors from  accumulating so $d_2(Q_n^{\dagger},Q^*)$ converges a.s. 
 The Lipschitz assumption on $Q^*$ and $\sum \alpha_n =  \infty$ guarantee that the distance converges to 0 a.s.  We also highlight that the rearrangement inequality in {Proposition \ref{prop:rearr}} is crucial in handling the rearrangement step after updating with each data point.  
\end{proof_outline}

Let us now write $Q_{n\infty}^{\dagger}$ as the random function obtained from {Algorithm \ref{alg:QMP}} starting at $Q_n^{\dagger}$ for each $n$, where the additional index $n$ on $Q_{n\infty}^{\dagger}$ is to indicate the dependence on the initial $Q_n^{\dagger}$. A novel contribution of our work is that consistency of $Q_n^{\dagger}$ can be used to show consistency of the entire QMP, which follows from an application of Markov's inequality and {Proposition \ref{prop:rearr}}.
\begin{theorem}\label{th:consistency}
    Under  {Assumptions \ref{as:alpha}},  {\ref{as:bandwidth}} and  {\ref{as:consistency_truth}}, for any $\varepsilon > 0$, the QMP from {Algorithms \ref{alg:fit}} and  {\ref{alg:QMP}} satisfies
    \begin{align*}
\Pi_n\left(Q^{\dagger}_{n\infty}: d_2\left(Q^{\dagger}_{n\infty},Q^*\right) \geq \varepsilon \mid Y_{1:n}\right) \to 0 \quad \textnormal{a.s.}[P^*]
    \end{align*}
\end{theorem}
\begin{proof_outline}
    We follow a similar approach to Example 8.5 from \cite{Ghosal2017}.
    As $d_2^2(Q^{\dagger}_{n\infty}, Q^*) \leq d_2^2\left(Q_{n\infty}, Q^*\right)$ from {Proposition \ref{prop:rearr}}, we have from Markov's inequality that
    \begin{align*}
        \Pi_n\left(Q^{\dagger}_{n\infty}: d_2\left(Q^{\dagger}_{n\infty},Q^*\right) \geq \varepsilon \mid Y_{1:n}\right) \leq \frac{1}{\varepsilon^2 }\E\left[d^2_2\left(Q_{n\infty},Q^*\right)\mid Y_{1:n}\right].
    \end{align*}
 We decompose $d^2_2\left(Q_{n\infty},Q^*\right)$ into a posterior variance component $\E[d^2_2(Q_{n\infty}, Q_n^{\dagger}) \mid Y_{1:n}]$, a point estimate component $d_2( Q_n^{\dagger},Q^*)$ and a cross-term. The posterior variance is sent to 0 by the sequence  $\alpha_N$, so posterior consistency depends only on consistency of $Q_n^{\dagger}$, which is guaranteed by 
    {Theorem \ref{th:consistency_mean}}.
\end{proof_outline}
Once again, the connections of the $L^2$ distance between quantile functions and the Wasserstein  metric suggest that the QMP over $P_\infty$ is consistent at $P^*$ in the Wasserstein metric; posterior asymptotics in this metric space has also been studied by \cite{Chae2021}.

\subsubsection{Posterior contraction rate}\label{sec:contraction}
A more challenging but informative result is the posterior contraction rate, which quantifies how quickly the QMP concentrates on the true $Q^*$. Once again, we will rely on the convergence rate of $Q_n^{\dagger}$ to the truth, but we have only managed to show results for quite stringent additional assumptions, given below.

\begin{assumption}[Lipschitz quantile functions, learning rate and bandwidth]\label{as:contraction_init}
   Suppose {Assumption \ref{as:consistency_truth}} holds, and
     additionally that $\alpha_i =  a(i+1)^{-1}$ for  $a > M/2$ and the bandwidth satisfies $\rho_i = \sqrt{1- ci^{-k}}$ 
    for $k > 4$ and $c \in (0,1)$.
\end{assumption}
\begin{theorem}\label{th:contraction_mean}
    Under {Assumption \ref{as:contraction_init}}, we have that for any $0< \delta < 1$, {Algorithm \ref{alg:fit}} satisfies 
    \begin{align*}
n^{\delta}d^2_2\left(Q_n^{\dagger},Q^*\right) \to 0 \quad \text{a.s.}[P^*]
    \end{align*}
\end{theorem}
\begin{proof_outline}
    The proof follows a similar argument to \cite{Aboubacar2014}, where we extend the consistency proof to show that $n^{\delta}d^2_2(Q_n^{\dagger},Q^*) $ is an almost supermartingale. 
\end{proof_outline}
\begin{theorem}\label{th:contraction}
Under {Assumption \ref{as:contraction_init}}, the sequence $\varepsilon_n = n^{-\delta/2}$ for any $0 < \delta < 1$ is a valid posterior contraction rate for the QMP from {Algorithms \ref{alg:fit}} and  {\ref{alg:QMP}}, that is for any finite $K> 0$, we have
\begin{align*}
\Pi_n\left(Q^{\dagger}_{n\infty}: d_2\left(Q^{\dagger}_{n\infty},Q^*\right) \geq K\varepsilon_n  \mid Y_{1:n}\right) \to 0 \quad \textnormal{a.s.}[P^*]  
\end{align*}
\end{theorem}
\begin{proof_outline}
    The proof continues from that of {Theorem \ref{th:consistency}}. The posterior variance is $O(n^{-1})$ due to the sequence $\alpha_n$, so we just require the convergence rate of $Q_n^{\dagger}$ as provided by {Theorem \ref{th:contraction_mean}}. 
\end{proof_outline}
Although it is encouraging that obtaining a posterior contraction rate is possible for the QMP, the assumptions on $a$ and $k$ are not conducive for good performance in practice, as we will see in Section \ref{sec:implications}. 
In particular, we require $k < 1$ for smoothness, and $M$ can be very large if $P^*$ has light tails, greatly inflating posterior variance. As a result, we do not suggest the usage of the above {Theorem \ref{th:contraction}} for hyperparameter setting. We believe it likely that the condition on $k$ is an artefact of the proof, and suspect it may be relaxed. However, it is possible that the assumption on $a > M/2$ is necessary. One potential solution is to only consider the posterior contraction rate of the quantile function on a subset of $(0,1)$, which can decrease the required lower bound on $a$. Another potential remedy for this impractical setting of $a$ is to use a functional learning rate $a(u)$ based on a density estimate, which we discuss in Section \ref{sec:functional}.

\section{Hyperparameters and approximate posterior sampling}\label{sec:implications}
Although the theory just introduced is technical, we now shed light on the practical utility of the above theory and its extensions for practical selection of hyperparameters and approximate posterior sampling. 

\subsection{Learning rate }\label{sec:learning_rate}

 The sequence $\alpha_i$ is extremely important for both consistency of $Q_n^{\dagger}$ and the amount of uncertainty obtained when predictive resampling. As a reminder, we will let $\alpha_i = a(i+1)^{-1}$, as in {Assumption \ref{as:alpha}},
where $a \in (0,\infty)$ which we will refer to as the learning rate.
The above sequence clearly satisfies (\ref{eq:alpha_conditions}) due the rate of $\alpha_i \to 0$. The choice of the learning rate $a$ however requires much care, as it directly controls the magnitude of the posterior uncertainty. Perhaps surprisingly, a default choice for $a$ can be justified by considering the asymptotic posterior variance of a low-dimensional functional of the QMP, which we now discuss. This works well in practice across a general range of settings.

Consider the mean of $P_\infty$, which can be written as $
  \mu_\infty = \int_0^1 \, Q_\infty^{\dagger}(u) \, du = \int_0^1 \, Q_\infty(u) \, du$,
where the last equality can be seen from the integral preserving property of increasing rearrangement ({Lemma \ref{app:lem_equi}} in the Appendix). This allows us to work with $Q_\infty$ directly instead of $Q_\infty^{\dagger}$, where the latter is much more challenging due to its non-linearity.
The following proposition quantifies the posterior mean and asymptotic variance of $\mu_\infty$.
\begin{proposition}\label{prop:mu}
    Let $\mu_n = \int Q^{\dagger}_n(u) \, du$ for $\{Q^{\dagger}_n\}_{n \geq 1}$ from {Algorithm \ref{alg:fit}},  and $\mu_{n\infty} = \int_0^1 Q_{n\infty}^{\dagger}(u)\, du$ where $Q_{n\infty}^{\dagger}$ arises from {Algorithm \ref{alg:QMP}} starting from $Q^{\dagger}_n$. Under {Assumptions \ref{as:alpha}} and {\ref{as:bandwidth}}, we have $E[\mu_{n\infty} \mid Y_{1:n}] = \mu_n$  for each $n\geq 1$, and ${n}\, \E\left[\left(\mu_{n\infty} - \mu_n\right)^2\mid Y_{1:n}\right]\to {a^2}/{12}$ a.s.$[P^*]$.
\end{proposition}
If $\mu_n$ is the sample mean for $Y_{1:n} \iid P^*$ with mean $\mu^*$, then it has the asymptotic variance $\sigma^2/n$, where $\sigma^2$ is the variance of $Y \sim P^*$.
A natural matching of the asymptotic variance of the QMP to the sample mean then involves setting $a = \sqrt{12}\, \sigma$. Although we cannot guarantee that  $Q_n^{\dagger}$ gives an efficient estimate of $\mu_n$,  this serves as a simple default choice which works well in practice, and one can also regard our suggestion as a lower bound on $a$. In this case, we can actually show asymptotic normality by leveraging {Theorem \ref{th:gp}}, which we defer to Section \ref{app:sec_munormal} in the Appendix.

One potential weakness of the QMP is that $a$ is only a scalar, so we have to choose a single low-dimensional functional for which we want to match asymptotic variances. 
Nonetheless, our above suggestion based on the mean functional works well in practice. 
As discussed after {Theorem \ref{th:contraction}}, the above issue can also be potentially alleviated with a functional learning rate $a(u)$ which we discuss in Section \ref{sec:functional}, but requires a separate density estimate. More generally, the learning rate for martingale posteriors remains an important open topic of research.

\subsection{Bandwidth sequence }\label{sec:bandwidth}
The bandwidth sequence $\rho_i$ governs the smoothness of the update, and there are two competing effects. First, we would like $\rho_i \to 1$ so $H_{\rho_i}(u,v)$ approaches an indicator function, which is required for frequentist consistency in the $L^2$ norm ({Theorem \ref{th:consistency}}). This is akin to the condition required for kernel density estimation, and we see in practice that having $\rho_i \to 1$ also improves the initial $Q_n^{\dagger}$. On the other hand, {Theorem \ref{th:sobolev_mart}} assures us that posterior samples of the quantile function from the QMP are only absolutely continuous if $\rho_i$ does not approach 1 too quickly. Furthermore, under the assumption of $Q^*(u)$ being differentiable, if $\rho_i$ approaches 1 too quickly, then the weak derivatives $q_n^{\dagger}$ do not approximate the derivative of $q^*$ well. A slower convergence of $\rho_i \to 1$ also results in fewer violations of monotonicity when applying {Algorithm \ref{alg:fit}}. The importance of {Assumption \ref{as:bandwidth}} for {Corollary \ref{cor:continuous}} is illustrated in {Figure \ref{fig:bandwidth} (left, middle)}, where we see that QMP samples of $Q_N^{\dagger}$ are smooth for $k = 0.5$, but non-smooth for $k = 1.5$.

Our suggestion is thus to set the bandwidth sequence as
$\rho_i = \sqrt{1 - ci^{-k}}$ as in {Assumption \ref{as:bandwidth}},
  where $c \in (0,1)$ and $k \in (0,1)$ are two hyperparameters. 
This form arises naturally from the proofs of {Theorems \ref{th:sobolev_mart}} and {\ref{th:consistency}}.
Although both theorems are satisfied for any $k \in (0,1)$,  we find the choice of $k = 0.5$ to work well in practice which balances between smoothness of the QMP and attaining $L^2$ consistency. We then suggest setting the constant $c \in (0,1)$ in a data-adaptive manner, which allows fine-tuning of the smoothness of the 
initial $Q_n^{\dagger}$ to the specific dataset. As we have $\rho_1 = \sqrt{1-c}$ and $\rho_n = \sqrt{1-cn^{-0.5}}$,  the constant $c$ controls the initial value  $\rho_1$ which increases monotonically to $\sqrt{1-cn^{-0.5}}$ as $i \to n$. 

To choose $c$, we suggest maximizing the prequential log  score due to its connections to the marginal likelihood \citep{Dawid1984,Gneiting2007,Fong2020}. In particular, the prequential log score is easy to compute in our setting, as we have
$$\sum_{i = 1}^n \log \left[p_{i-1}(Y_i)\right] = - \sum_{i = 1}^n \log \left[{q_{i-1}^{\dagger}\left(P_{i-1}\left(Y_i\right)\right)}\right]$$
where $q_{i}^{\dagger}$ is the weak derivative of $Q_i^{\dagger}$. The existence of $q_{i}^{\dagger}$ is guaranteed by the absolute continuity of $Q_i^{\dagger}$ and {Theorem \ref{app:thm_sobolev}}
in the Appendix. The choice of the above is justified as we should rely on $q_i^{\dagger}$ in some way to set $c$, as relying on $Q_i^{\dagger}$ alone (e.g. with the $L^2$ norm) will not guarantee smooth estimates. We can compute $q_i^{\dagger}$ easily with finite differences, and $P_{i-1}(Y_i)$ is already computed for our update.

\begin{figure}[!h]
\begin{center}
\includegraphics[width=\textwidth]{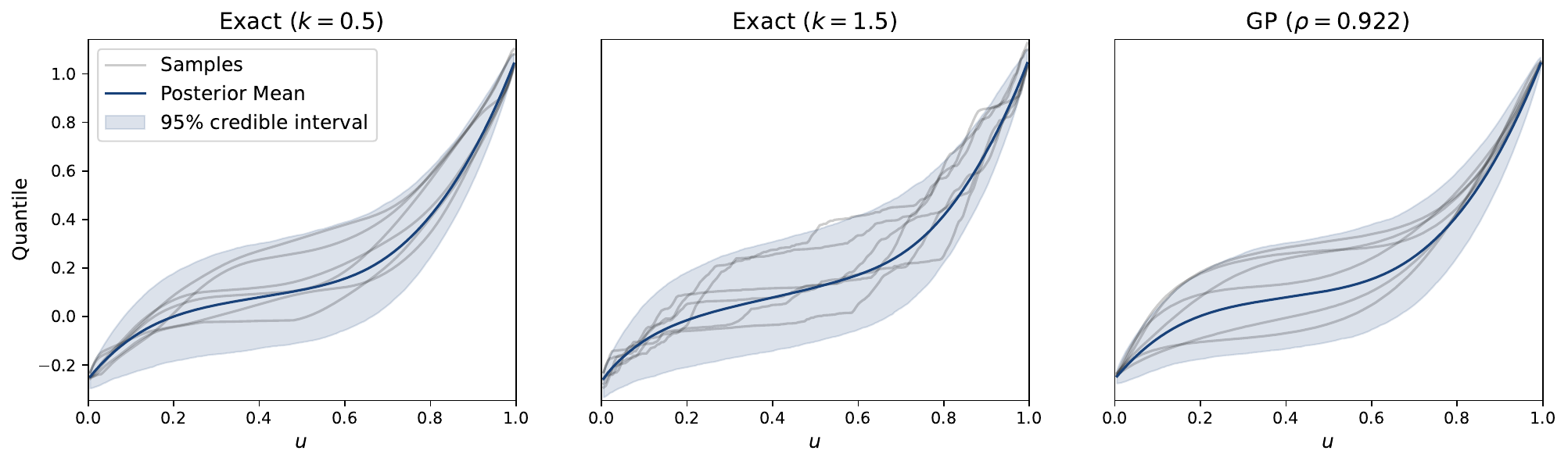}
\end{center}
\caption{Posterior samples, mean and 95\% credible intervals of ${Q}^{\dagger}$ for (Left) $k = 0.5$; (Middle) $k = 1.5$; (Right) GP approximation; all plots are with initial $Q_n(u) = 4(u-0.4)^3 +0.2u$, $n = 10$, $N = n + 5000$, $a \approx 0.95$ and $c = 0.5$; generating $B = 5000$ exact and approximate samples required 15s and 0.2s respectively.}
\label{fig:bandwidth}
\end{figure}

\subsection{Approximate posterior sampling}\label{sec:approx}
This subsection is dedicated to utilizing {Theorem \ref{th:gp}} in order to drastically accelerate quantile predictive resampling. For $S_n(u) = Q_\infty(u) - Q_n(u)$, we essentially have that $a^{-1}\sqrt{n+1}\,S_n \overset{d}{\approx}  \mathbb{G}$ for sufficiently large $n$,
where $\mathbb{G} \sim \mathcal{GP}\left(0, \left(\min\{u,u'\} - uu'\right)\right)$ is the Brownian bridge. Unlike in the case of \cite{Fortini2020, Fortini2023} and the regular MP, the distribution of $\mathbb{G}$ does not depend on any random quantities, which arises from working with the quantile instead of the distribution, and allows easier sampling. Furthermore, we only require realizations of $Q_\infty$ to lie in $L^2((0,1))$ or $H^1((0,1))$, which is much simpler than needing realizations to be valid probability measures as in the regular MP case. It thus seems reasonable to approximate sampling $Q_\infty$ with
$\widetilde{Q}_\infty= Q^{\dagger}_n + {a\, \mathbb{G}}/{\sqrt{n+1}}$.
Algorithmically, this involves drawing a sample from a Brownian bridge, then scaling it  by $a/\sqrt{n}$ and adding it to the initial $Q_n^{\dagger}$. The immediate downside to this approach is that samples of $\widetilde{Q}_\infty$ will not be smooth (i.e. in $H^1((0,1))$) even if $Q_\infty$ is from {Theorem \ref{th:sobolev_mart}}, due to the a.s. nowhere differentiability of paths from a Brownian bridge. 

To remedy this, we propose the following alternative approximation:
\begin{align*}
\widetilde{Q}_\infty= Q^{\dagger}_n + a\,{\mathbb{G}_{\rho_{n+1}}}/{\sqrt{n+1}},
\end{align*}
where $\mathbb{G}_{\rho}$ is a zero-mean GP with covariance function
${k}_{\rho}(u,u') = C_{\rho^2}(u,u') - uu'$ and $C_\rho(u,u')$ is the bivariate normal copula. To justify this above choice, we have the following theorem.
\begin{theorem}\label{th:approx_GP}
    Let $S_n = Q_\infty - Q_n$  and let $\widetilde{S}_n = a\,{\mathbb{G}_{\rho_{n+1}}}/{\sqrt{n+1}}$ be the approximation as defined above, and suppose {Assumptions \ref{as:alpha}} and {\ref{as:bandwidth}} hold true. The covariance function of $S_n$, which we write as $k_n(u,u') := \E\left[S_n(u)\, S_n(u')\right]$,
    satisfies the following for all $u,u' \in (0,1)$ and $n \geq 1$:
\begin{align*}
 k_{\rho_{n+1}}(u,u') \leq  r_n^{-1}{k_n(u,u')} \leq \min\{u,u'\} - uu',
\end{align*}
where $r_n = \sum_{i = n+1}^\infty\alpha_i^2 \approx a^2(n+1)^{-1}$,
and both $k_{\rho_{n+1}}(u,u')$ and $r_n^{-1}k_{n}(u,u')$ converge to $\min\{u,u'\} - uu'$ as $n \to \infty$.
Furthermore, realizations of $\widetilde{S}_n$ lie in $H^1((0,1))$ a.s., and $a^{-1}\sqrt{n}\,\widetilde{S}_n$ converges weakly  in $\ell^\infty((0,1))$ to the Brownian bridge $\mathbb{G}$.
\end{theorem}
From the above, we have that $\widetilde{Q}_\infty$ and $Q_\infty$ have the same distribution asymptotically when suitably normalized, which happens as $\rho_{n} \to 1$. Furthermore, realizations of $\widetilde{Q}_\infty$ (and thus $\widetilde{Q}_\infty^{\dagger}$) lie in the same Sobolev space a.s. This occurs as the true covariance function $k_n$ lies in between $k_{\rho_{n+1}}$ and that of the Brownian motion in terms of smoothness, where we prefer $k_{\rho_{n+1}}$ to $k_n$ as the former is much cheaper to compute. The above theorem thus justifies the choice of $\widetilde{Q}_\infty$ as a suitable approximation to $Q_\infty$.
The above inequality actually suggests that  
 sample paths of $\widetilde{Q}_\infty$ may be slightly smoother than that of $Q_\infty$. In practice, this effect disappears quickly with increasing $n$ as the inequality is very tight even for moderate $n$. 

This approximate sampling scheme is given in {Algorithm \ref{alg:approx_QMP}}, where drawing from the GP is very cheap and detailed in Section \ref{app:sec:alg} of the Appendix. In practice, this approximation works extremely well, as we illustrate in {Figure \ref{fig:bandwidth} (right)}. Both samples and credible intervals of $\widetilde{Q}_\infty^{\dagger}$ are visually very similar to $Q_N^{\dagger}$ even for $n = 10$. Furthermore, generating $B = 5000$ posterior samples required 15s and 0.2s for the exact and approximate case respectively, which indicates a substantial speedup. We will see further demonstration of the computational gains and similar results in later in the illustrations. \vspace{2mm}
\begin{figure}[ht]
\center
\small
 \begin{minipage}{.57\linewidth}
\begin{algorithm}[H]
{Initialize $Q_n$ from {Algorithm \ref{alg:fit}}\\}
{Set $\rho_{n+1} = \sqrt{1- c(n+1)^{-k}}$\\}
\For{$b \gets 1$ \textnormal{\textbf{to}} $B$}{
    {Draw ${S}^{(b)}\sim \mathcal{GP}(0, C_{\rho_{n+1}^2}(u,u')- uu')$}\\
    Compute $\widetilde{Q}_\infty^{(b)} = Q_n^{\dagger} + a\,S^{(b)}/\sqrt{n+1}$
}
 {Return $\left\{{\widetilde{Q}_\infty^{\dagger}}{}^{(1)},\ldots,{{\widetilde{Q}}_\infty^{\dagger}}{}^{(B)}\right\}$}
\caption{Approximate QMP sampling with GPs }\label{alg:approx_QMP}
\end{algorithm}\vspace{-2mm}
\end{minipage}

\end{figure}

\section{Quantile regression}\label{sec:quantreg}

Having established the framework and theory for the QMP, we now introduce the QMP in the quantile regression setting, which is a natural extension. This is in contrast to the usual intricacies involved in specifying nonparametric prior distributions with covariate dependence. We will focus on the linear case, and leave discussion of potential directions for the non-linear case to Section \ref{sec:nonlinear}.

To begin, we assume that $\{Y_i, X_i\}_{i = 1,\dots,n} \iid P^*(y,x)$, where $Y \in \mathbbm{R}$ and $X \in \mathcal{X} \subset\mathbbm{R}^p$. The conditional distribution $P^*(y \mid x)$ is assumed to have a quantile function which varies linearly, that is
$Q^*(u \mid x) =  \beta^*(u)^T x$,
where $\beta^*(u):(0,1) \to \R^d$ is the true unknown coefficients. As we can write
$\beta^*(u) =\argmin_\beta \int \rho_u(y - \beta^T x) \, dP^*(y,x)$,
this immediately suggests a quantile regression version of (\ref{eq:quantile_copula}):
\begin{align}\label{eq:beta_freq}
      \beta_{n+1}(u) &= \beta_{n}(u)  + \alpha_{n+1}\left[u - H_{\rho_{n+1}}\left(u,P_{n}(Y_{n+1} \mid X_{n+1})\right)\right]X_{n+1},
\end{align}
where $P_n(y \mid x) = \int_0^1 \mathbbm{1}\left(Q_n(u \mid x) \leq y\right)\, du$.
 We now utilize the above for the QMP for quantile regression. 

\subsection{Quantile predictive resampling}
The predictive resampling scheme for the quantile regression setting is a straightforward extension of Section \ref{sec:PR}. The key extra ingredient is that we will use the empirical distribution for predictive resampling $X_{n+1:\infty}$, which is equivalent to the Bayesian bootstrap as suggested in \cite{Fong2023a}. This is particularly natural in our setting, where we are mainly interested in $P^*(y \mid x)$ or $Q^*(u \mid x)$. Quantile predictive resampling then consists of first drawing $X_{N+1} \sim \frac{1}{N}\sum_{i = 1}^N \delta_{X_i}$, then simulating $V_{N+1} \sim \mathcal{U}(0,1)$ and computing
    \begin{align}\label{eq:beta_qmp}
      \beta_{N+1}(u) &= \beta_{N}(u)  + \alpha_{N+1}\left[u - H_{\rho_{N+1}}\left(u,V_{N+1}\right)\right]X_{N+1},
\end{align}
which again is a martingale.
The simple uniform r.v. again arises as $P_N(Y_{N+1} \mid X_{N+1}) \sim \mathcal{U}(0,1)$ if $Y_{N+1} \sim P_N(\cdot \mid X_{N+1})$ and $P_N(\cdot \mid X_{N+1})$. For the covariates, it will be simpler computationally to draw $w_{1:n} \sim \text{Dir}(1,\ldots,1)$, followed by $X_{n+1:N}  \iid \sum_{i = 1}^n w_i \delta_{X_i}$. To draw actual samples of the observations $Y_{N+1}$ given $X_{N+1}$, we can analogously compute $Y_{N+1} = \beta_{N+1}(V_{N+1})^T X_{N+1}$ which is straightforward. 
For the initial estimate, we can once again just apply the update $(\ref{eq:beta_freq})$ on the i.i.d. observables $Y_{1:n}$. Our implicit quantile function $Q_i^{\dagger}(u \mid x)$ is then the increasing rearrangement of $\beta_i(u)^Tx$ at each value of $x$.

\subsection{Predictive asymptotics}
\subsubsection{Martingale}
It is not too hard to verify that we once again have a pointwise martingale condition, that is
\begin{align*}
    \E\left[\beta_{N+1}(u) \mid Y_{1:N}, X_{1:N}\right] = \beta_N(u),
\end{align*}
which we can see by first computing the conditional expectation of $Y_{N+1}$ given $X_{N+1}$, which returns $\beta_N(u)$, so the additional expectation over $X_{N+1}$ does not affect this. 
Looking at  ($\ref{eq:beta_qmp}$), we see that each component $\beta_{Nj}(u)$ for $j \in \{1,\ldots,p\}$ is a function-valued martingale as before, with the additional term due to the covariates. It is thus not too difficult to show the following.
\begin{theorem}\label{th:qr_sobolev}
    Under {Assumptions \ref{as:alpha}}, {\ref{as:bandwidth}},  {\ref{as:beta_L2}} and {\ref{as:beta_weak_deriv}}, there exists a random vector function $\beta_\infty(u)$ with realizations  in $H^1((0,1))^p$ such that $\beta_N(u)$ satisfies $d_{1,2}\left(\beta_{Nj},\beta_{\infty j}\right)  \to 0$ for each component $j \in \{1,\ldots,p\}$  a.s. Furthermore, each component of the realizations of $\beta_\infty(u)$ is absolutely continuous on $(0,1)$ a.s.
\end{theorem}
\begin{proof_outline}
For each dimension $j \in \{1,\ldots,p\}$, we can apply the same derivation as in {Theorem \ref{th:sobolev_mart}}, with the key difference that the update term for $\beta_{Nj}(u)$ is scaled by $X_{Nj}$.
\end{proof_outline}
In order to study the result of the rearrangement process, we now study the conditional quantile function estimate directly, which satisfies 
$Q_{N}(u\mid x) = \beta_N(u)^T x$ for each $x \in \mathcal{X}$. The implicit conditional quantile function is then the increasing rearrangement of $Q_N(u \mid x)$ for each $x$, which we write as $Q_N^{\dagger}(u \mid x)$.
\begin{proposition}\label{prop:condit_quant}
Under {Assumptions \ref{as:alpha}}, {\ref{as:bandwidth}},   {\ref{as:beta_L2}} and  {\ref{as:beta_weak_deriv}}, for each $x \in \mathcal{X}$, there exists a random function $Q_\infty^{\dagger}(u \mid x)$ with realizations in $H^1((0,1))$ such that $d_2(Q_N^{\dagger}(\cdot \mid x),Q_\infty^{\dagger}(\cdot \mid x)) \to 0$  a.s., where $Q_N^{\dagger}(u \mid x)$ and $Q_\infty^{\dagger}(u \mid x)$ are proper monotonically increasing quantile functions. Furthermore, $Q_\infty^{\dagger}(u \mid x)$ is the increasing rearrangement of $Q_\infty(u \mid x)=\beta_\infty(u)^T x$ a.s. 
\end{proposition}
\begin{proof_outline}
Since $Q_N(u \mid x)$ is just a weighted sum of $\beta_{Nj}(u)$, which are elements in a Banach space, the continuous mapping theorem can be used to show $Q_N(u \mid x) \to \beta_\infty(u)^T x$ a.s. The rearrangement step is then analogous to {Theorem \ref{th:sobolev_mart}}. 
\end{proof_outline}
We remark that once again, since $Q_\infty^{\dagger}(u \mid x)$ is in $H^1((0,1))$, it can be identified almost everywhere with an absolutely continuous conditional quantile function. One could also make similar statements on the weak convergence of the conditional distributions. 

An interesting phenomenon due to the nonlinearity of rearrangement is that even if $Q_N(u \mid x)$ is linear in $x$, the rearranged $Q_N^{\dagger}(u \mid x)$ may no longer be so. Nonetheless, {Proposition \ref{prop:rearr}} guarantees us that $Q_N^{\dagger}(u \mid x)$ will always be closer to $Q^*(u \mid x)$ in $L^2$ compared to $Q_N(u \mid x)$, so it is not too much of an issue for estimation. Interestingly, we can still say something about the QMP over the regression function $\E[Y \mid X]$. Let us define $\E_\infty[Y \mid x] := \int_0^1 Q_\infty^{\dagger}(u \mid x) \, du$, so realizations of $\E_\infty[Y \mid x]$ are samples of the regression function from the QMP. We then have the below, which follows from the equimeasurable property of rearrangement.
\begin{proposition}\label{prop:linear}
  Under {Assumptions \ref{as:alpha}}, {\ref{as:bandwidth}},  {\ref{as:beta_L2}} and  {\ref{as:beta_weak_deriv}}, the QMP has support over linear regression functions, that is realizations of $\E_\infty[Y \mid x]$ are linear functions of $x$ a.s.
\end{proposition}

\subsubsection{Gaussian process}
We can again study the asymptotic normality, this time focusing on the vector $\beta_{n}(u)$. Consider the difference
\begin{align*}
    S_{N}(u,j) = \sum_{i=N}^\infty\alpha_{i}\left[u -H_{\rho_{i}}\left(u,V_{i}\right)\right]X_{ij}\quad 
\end{align*}
for $j \in \{1,\ldots,p\}$ and $u \in (0,1)$,
where $X_{ij}$ is the $j$-th entry of $X_{i}$. All of the results in this subsection will be conditional on the Bayesian bootstrap weights $w_{1:n}$ and $X_{1:n}$. Similar to the non-regression case, we can use the Cram\'{e}r-Wold device to help us study the joint convergence of $S_{N}$ for an arbitrary finite collection of points.
Combining the above with asymptotic tightness, we can again extend the finite-dimensional joint convergence to uniform convergence with respect to $\mathcal{F} = (0,1) \times \{1,\ldots, p\}$.

\begin{theorem}\label{th:reg_gp}
   Under {Assumptions \ref{as:alpha}}, {\ref{as:bandwidth}},  {\ref{as:beta_L2}} and {\ref{as:beta_weak_deriv}}, conditional on $w_{1:n}$, $\sqrt{N}S_{N}$ converges weakly in $\ell^{\infty}(\mathcal{F})$ to $\mathbb{G}_{a}$ almost surely, where $\mathbb{G}_{a}$ is a zero-mean GP with covariance function
    $\E[\mathbb{G}_{a}(u,j), \mathbb{G}_{a}(u^{\prime}, j^{\prime})] = a^{2}\left[\sum_{k=1}^{n}w_{k}X_{kj}X_{kj^{\prime}}\right](\min\{u,u^{\prime}\} - uu^{\prime})$.
\end{theorem}
We are then free to replace the covariance function of the limiting GP with $C_{\rho_{n+1}^2}(u,u^{\prime}) - uu^{\prime}$ for approximate sampling as before, giving the covariance function 
\begin{align}\label{eq:reg_kernel}
k_{\rho_{n+1}}(\{u,j\}, \{u',j'\}; w_{1:n}) = \left[\sum_{k=1}^{n}w_{k}X_{kj}X_{kj^{\prime}}\right]\left(C_{\rho_{n+1}^2}(u,u^{\prime}) - uu^{\prime}\right).
\end{align}

\subsection{Frequentist asymptotics}\label{sec:reg_consistency}
In the quantile regression setting, the frequentist asymptotics of the QMP is unfortunately more challenging. The main challenge is that the rearrangement $Q_n^{\dagger}(u \mid x)$ does not preserve linearity of the rearranged conditional quantile, so we do not necessarily have a corresponding vector $\beta_n^{\dagger}(u)$. As a result, we cannot use an analogous rearranged update like in Section \ref{sec:initial}. We are however able to show an analogous posterior consistency result in the case where $\rho = 1$, which we detail in Section \ref{app:sec:QR_consistency} of the Appendix, as this special case lends itself more easily to a consistent estimate. However, this does not extend easily to the $\rho \neq 1$ case. Nonetheless, (\ref{eq:beta_qmp}) works well in practice, and for sufficiently slow rate of $\rho_i \to 1$, we find that $Q_n^{\dagger}(u \mid x) = Q_n(u \mid x)$ anyways. We thus conjecture that it will also satisfy posterior consistency, and we leave this for future work. 

\subsection{Practical considerations}
In the quantile regression case, the same considerations as Section \ref{sec:implications} can be made, where the added complications are that we also need to handle the random covariates.
 \subsubsection{Approximate posterior sampling}
As outlined in Section \ref{sec:reg_consistency}, a rearranged version of the update is not obvious, so we opt for {Algorithm \ref{alg:reg_fit}} to estimate the initial $\beta_n$. In the interest of space, we jump straight to the approximate sampling procedure in {Algorithm \ref{alg:reg_approx_QMP}}, with the exact case in {Algorithm \ref{alg:QMP_reg}} of the Appendix. Once again, the GP approximation is extremely expedient, and drawing from a GP with kernel (\ref{eq:reg_kernel}) is covered in Section \ref{app:sec:alg} of the Appendix. 
\begin{figure}[ht]
\center
\small
 \begin{minipage}{.46\linewidth}
\begin{algorithm}[H]\label{alg:reg_fit}
{Initialize $\beta_0$}\\
Data is $(Y_1,X_1),\ldots, (Y_n,X_n)$\\
\For{$i \gets 1$ \textnormal{\textbf{to}} $n$}{
Compute $V_i = P_{i-1}(Y_{i}\mid X_i)$\\
$\beta_{i}(u) = \beta_{i-1}(u)  + \alpha_{i}\left[u - H_{\rho_{i}}\left(u,V_i\right)\right]X_{i}$}
{Return $\beta_n$}
\caption{Estimation of quantile regression coefficients}
\end{algorithm}
\end{minipage}\hspace{2mm}
 \begin{minipage}{.52\linewidth}
\begin{algorithm}[H]
{Initialize $\beta_n$ from {Algorithm \ref{alg:reg_fit}}}\\
{Compute $\rho_{n+1} = \sqrt{1- c(n+1)^{-k}}$\\}
\For{$b \gets 1$ \textnormal{\textbf{to}} $B$}{
Draw $w_{1:n}^{(b)} \sim \text{Dirichlet}(1,\ldots,1)$\vspace{1mm}\\
Draw $S^{(b)}_{1:p} \sim \mathcal{GP}(0,k_{\rho_{n+1}}(\{u,j\},\{u',j'\};w_{1:n}^{(b)}))$\vspace{1mm}\\
Compute $\widetilde{\beta}^{(b)}_\infty = \beta_n + a\,S^{(b)}_{1:p}/\sqrt{n}$
}
 {Return $\left\{{\widetilde{\beta}_\infty}^{(1)},\ldots,{\widetilde{\beta}_\infty}^{(B)}\right\}$}
\caption{Approximate QMP Sampling for Quantile Regression with GPs}\label{alg:reg_approx_QMP}
\end{algorithm}
\end{minipage}
\end{figure}\vspace{-5mm}

\subsubsection{Hyperparameters}
The quantile regression case has the same hyperparameters, i.e. the learning rate $a$ and the bandwidth sequence $\rho_i$. Fortunately, the bandwidth sequence works exactly as before, where we set the value of $c$ according to $\sum_{i = 1}^n p_{i-1}(Y_i \mid X_i)$ which can be computed analogously. We thus turn our focus on the learning rate. Once again, we can consider the asymptotic posterior variance of a low-dimensional functional. In this case,  we can look at the marginal posterior mean and asymptotic covariance matrix on the linear regression coefficients, $\bar{\beta}_\infty = \int \beta_\infty(u) \, du$.
\begin{proposition}\label{prop:QR_coverage}
For $n \geq 1$, let $\bar{\beta}_n:= \int \beta_n(u) \, du$ for $\{\beta_n\}_{n \geq 1}$ arising from {Algorithm \ref{alg:reg_fit}}, and suppose that $X_{1:n} \iid P^*(x)$ with $\Sigma_x = \E[X_i X_i^T]$. Let $\bar{\beta}_{n\infty} = \int_0^1 \beta_{n\infty}(u)\, du$
 where $\beta_{n\infty}$ arises from {Algorithm \ref{alg:QMP_reg}} starting from $\beta_n$.
Under {Assumptions \ref{as:alpha}}, {\ref{as:bandwidth}} and {\ref{app:as_covariates_cov}}, we then have
$\E\left[\bar{\beta}_{n\infty} \mid Y_{1:n}\right] = \bar{\beta}_n$ for each $n \geq 1$, and 
$${n}\, \E\left[(\bar{\beta}_{n\infty} - \bar{\beta}_n)(\bar{\beta}_{n\infty} - \bar{\beta}_n)^T \mid Y_{1:n},X_{1:n}\right] \to (a^2/12)\, \Sigma_x \quad \textnormal{a.s.}[P^*].$$
\end{proposition}\vspace{-1mm}
We assume the covariates and response are standardized, so the intercept is 0 for simplicity, and {Assumption \ref{app:as_covariates_cov}} ensures $\Sigma_x$ is non-singular. The asymptotic covariance matrix of the least squares estimate of $\hat{\beta}_n$ in linear regression is $\sigma^2 \Sigma_x^{-1}/n$, where $\sigma^2$ is the variance of the residuals from the linear model. We can once again attempt a matching of asymptotic covariances, but matching the entire covariance matrix is not possible with a scalar $a$. Instead, we can match the determinant of the covariance matrices, which can be interpreted as matching the generalized variance \citep{Wilks1932}. This then gives the setting 
$a = {\sqrt{12} \sigma}/{\det \Sigma_x}$,
where we can estimate $\sigma$ and $\Sigma_x$ from the data. 
This default choice appropriately inflates the posterior variance in the presence of highly correlated covariates and as the dimension of $x$ increases, and works well in practice.
Analogous to the unconditional case, we can also adopt a $u$-specific and dimension-specific learning rate, $a_j(u) =a_j \, a(u) $, at the cost of having to depend on a separate density estimate of the residuals. We provide a brief discussion in Section \ref{app:sec_func} of the Appendix, but leave a detailed investigation for future work.

\section{Illustrations}\label{sec:illustrations}
We now illustrate the QMP on a simulation and real dataset respectively. All methods are implemented in \texttt{JAX} \citep{Jax2018} in Python, and executed on an Apple M2 Pro CPU. Due to the parallel nature of the QMP, significant acceleration is possible on a GPU \citep{Fong2023a}, but we use a CPU to illustrate the speed-up attained by the GP approximation.

\subsection{Simulations}\label{sec:simulation}
In this section, we demonstrate the method and practical performance for unconditional quantile estimation under different sample sizes, as well as comparing the computation time of exact and approximate sampling schemes. Let $Y_{1:n} \iid P^*$, where $P^*$ has the associated quantile function $Q^*(u) = 4(u-0.4)^3 + 0.2u$. We consider two sample sizes, $n = 50$ and $n = 500$, and compare the QMP distributions. For estimation, we initialize with $Q_0(u) = y_{\text{min}} + (y_{\text{max}}- y_{\text{min}})\, u$, which implies a uniform distribution over the range of the observations, and is appropriate here as we know the range of $y$ is bounded. 
We average over $10$ permutations of the data to compute $Q_n^{\dagger}$. We follow the guidance of Section \ref{sec:learning_rate} and \ref{sec:bandwidth}, and set $c$ by maximizing the prequential log score (also averaged over 10 permutations) on a  grid of $c\in (0,1)$ values of size $20$. For exact predictive resampling, we let $N = n + 5000$, and sample $B = 5000$ independent posterior samples.
 For all examples, we compute the quantile function estimates on a uniform grid on $[0,1]$ of size 200. 

\begin{figure}[!h]
\begin{center}
\includegraphics[width=\textwidth]{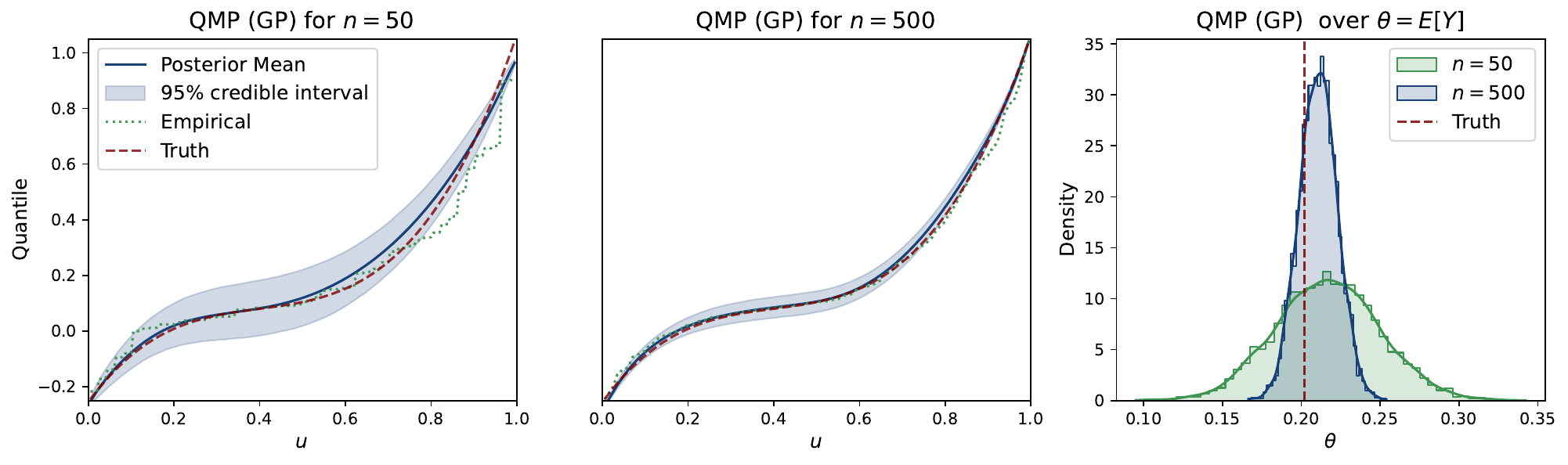}
\end{center}
\caption{QMP over $Q_\infty^{\dagger}$  for (Left) $n = 50$; (Middle) $n = 500$; (Right) QMP over  $\theta = \E[Y]$;  we only show the GP approximation as it is visually indistinguishable from exact sampling. } 
\label{fig:simulation}
\end{figure}

The selection of $c$ and estimation of $Q_n^{\dagger}$ for $n = 50$ and $n = 500$ required 0.7s and 1.4s respectively, where $c$ is chosen to be 0.6 and 0.75 respectively. We highlight that tuning $c$ can be easily parallelized if desired. In both sample sizes, exact predictive resampling required 15 seconds, whereas approximate predictive resampling with the GP only required 0.15s, which is a very significant speed-up. In {Figure \ref{fig:simulation}}, we plot the QMP mean and 95\% credible intervals for $Q_\infty^{\dagger}(u)$ and $\theta = \E[Y]$ for the two simulated sample sizes, with the empirical quantile estimate and true $Q^*$ for reference. As the exact and approximate QMP are visually indistinguishable, we only plot the latter in the interest of space in the main paper, with the exact QMP in Section \ref{app:sec_sim} of the Appendix.  
We can see that the posterior mean is monotonic and smooth, and is  regularized towards the initial linear $Q_0$ compared to the empirical quantile estimate.  As $n$ increases, the posterior mean approaches the truth, and the credible intervals shrink and capture the truth for central values of $u$ but 
seem to be anticonservative for values of $u$ close to $0$ or $1$.  As addressed by {Proposition \ref{prop:mu}}, the learning rate $a$ is chosen based on the asymptotic variance for the mean functional, which manifests as conservative and anticonservative credible intervals for the central and tail quantiles respectively. 
This is an inherent limitation of the scalar learning rate, and we discuss a potential extension on the QMP to address this in Section \ref{sec:functional}. 
We see in the {Figure \ref{fig:simulation} (right)} that the posterior distribution for $\theta$ concentrates at $n$ increases.\vspace{5mm}

\subsection{Cyclone dataset}
Following \cite{Tokdar2012} and \cite{An2024}, we now demonstrate the QMP for quantile regression in a real dataset based on a tropical cyclone intensity dataset from \cite{Elsner2008}. The dataset\footnote{\href{https://myweb.fsu.edu/jelsner/temp/Data.html}{https://myweb.fsu.edu/jelsner/temp/Data.html}} consists of $n = 2097$ tropical cyclones and their respective lifetime maximum wind speeds from the years 1981-2006. Covariates include the year, basin, latitude, and age of the cyclone; see the Supplementary Information of \cite{Elsner2008} for more details. Both \cite{Tokdar2012} and \cite{An2024} studied a subset of tropical cyclones in the North Atlantic (NA) basin ($n = 291$) with the year as the single covariate, and identified an increasing trend.

 For the QMP, we initialize $Q_0$ by setting $\beta_{0j}(u) = 0$ for $j \in \{1,\ldots,d\}$ and only set the intercept term $\beta_{00}(u)$ to be non-zero, which corresponds to initializing $Q_0(u \mid x) = Q_0(u)$. We set $\beta_{00}(u)$ to be the line interpolating the lower and upper quartile of $y$, which will reduce the impact of outliers on $Q_0$ compared to using the whole range of $y$. For both data sizes, we  average over $10$ permutations, but this could be reduced for large $n$ as there is less sensitivity to data ordering. Once again, we choose $c\in (0,1)$ by maximizing the prequential log score on a grid of size $20$, and estimate $\beta(u)$ on a grid on $[0,1]$ of size 200. We standardize all covariates and the response, and rescale after estimation. For the results, we again only present the GP approximation, as the posterior samples are visually indistinguishable from the exact sampler; this comparison is provided in Section \ref{app:sec_cyc} of the Appendix.  
As benchmarks, we compare to quantile regression with the \texttt{quantreg} package \citep{Koenker2018} for each $u$ independently followed by increasing rearrangement. We also compare to the dependent quantile pyramids (DQP) method of \cite{An2024}, and utilize the author's MCMC implementation in \texttt{C++}.

We first analyze the subset of tropical cyclones within the NA basin ($n = 291$) with a single covariate and the lifetime maximum wind speed as the response. MCMC for the DQP required 26 minutes to generate  $B = 10000$ posterior samples after a burn-in period of $10000$ samples. This long burn-in is necessary due to slow mixing, and note that the posterior samples may still be highly correlated without thinning.
The QMP required 1s for tuning the hyperparameter ($c = 0.95$) and estimating $Q_n^{\dagger}$. A further 33 seconds for exact predictive resampling or 0.4 seconds for approximate predictive resampling was needed to generate $B = 10000$ independent QMP samples. Not only is this orders of magnitude faster than MCMC,  both exact and approximate predictive resampling are inherently parallelizable, and can be efficiently accelerated using GPUs if desired \citep{Fong2023a}. Furthermore, the samples produced are independent and convergence concerns are minor, unlike MCMC where mixing is always a concern. The effect of truncating at a final $N$ is relatively harmless for exact predictive resampling, and the approximate sampler with the GP does not even require  truncation. This highlights the scalability of the QMP due to not relying on MCMC.

\begin{figure}[!h]
\begin{center}
\includegraphics[width=\textwidth]{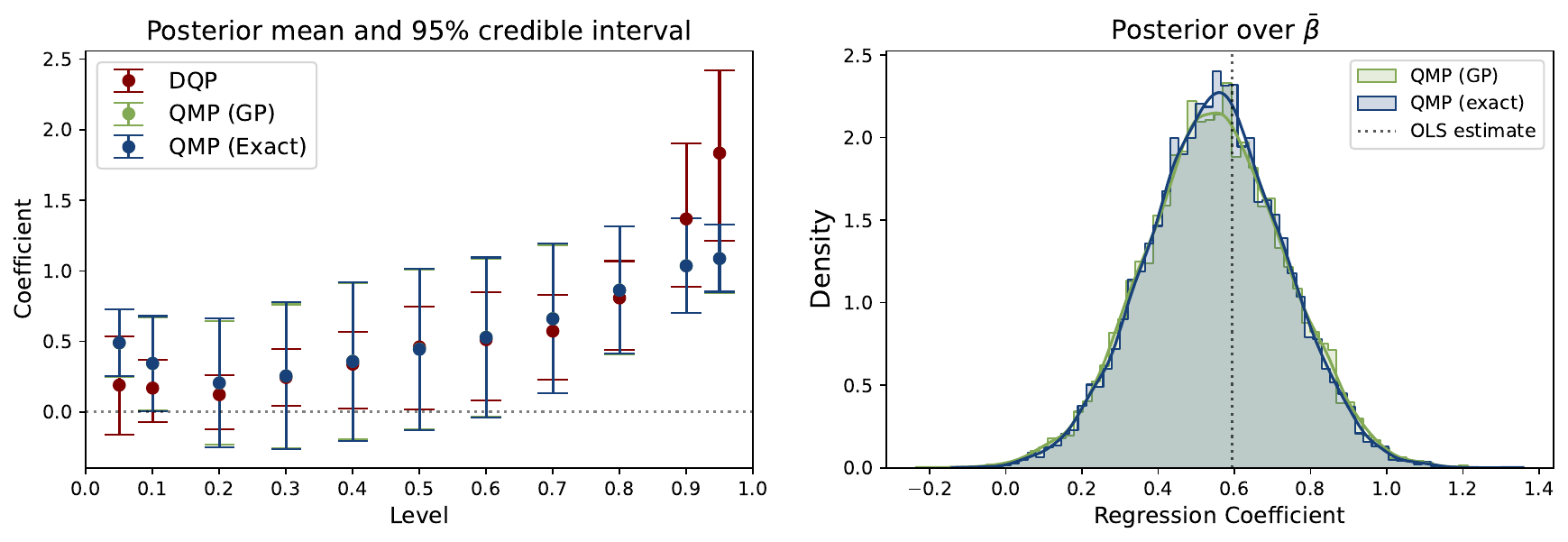}
\end{center}
\caption{Tropical cyclone maximum speeds in the NA basin ($n = 291$): (Left) Posterior mean and 95\% credible intervals for $\beta_{1\infty}(u)$ from the exact and approximate QMP and DQP; (Right) Posterior distribution of $\bar{\beta}_\infty$ for the exact
and approximate QMP } 
\label{fig:reg_small_cred}
\end{figure}
\begin{figure}[!h]
\begin{center}
\includegraphics[width=\textwidth]{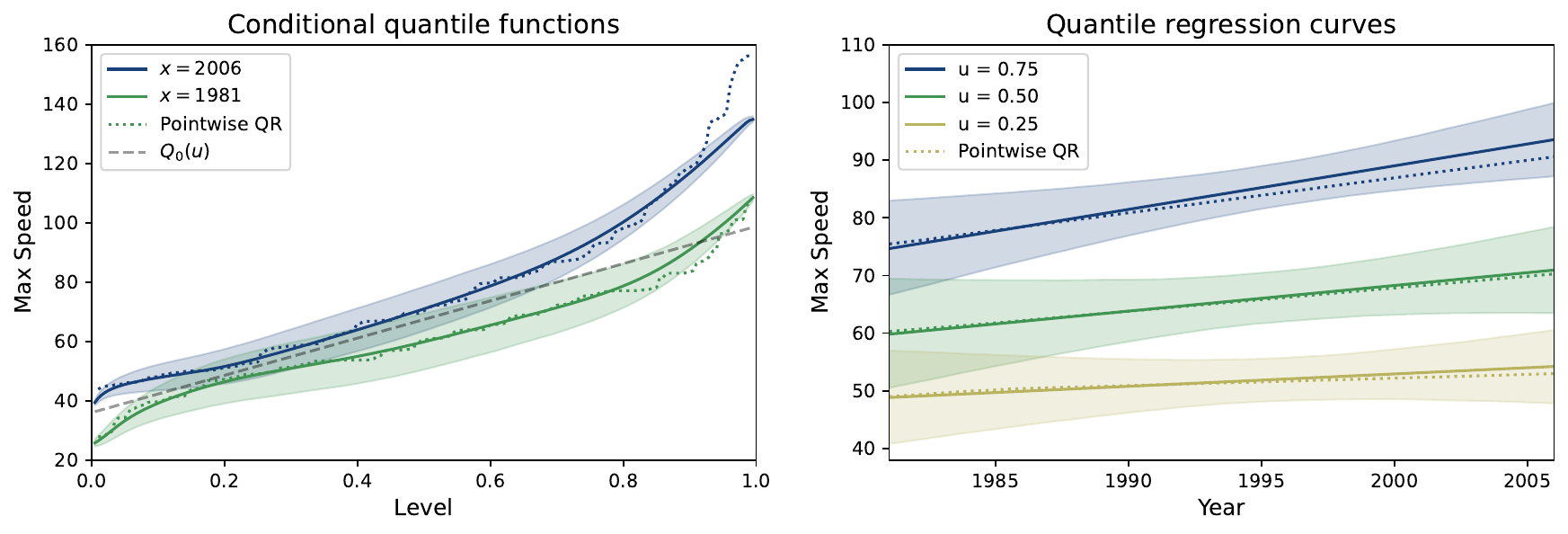}
\end{center}
\caption{Tropical cyclone maximum speeds in the NA basin ($n = 291$): (Left) Posterior mean and 95\% credible intervals for $Q(u \mid x = 1981)$ and  $Q(u \mid x = 2006)$  from the approximate QMP; (Right)  Posterior mean and 95\% credible intervals for $Q^{\dagger}_\infty(u = u^* \mid x)$ for $u^* \in \{0.25,0.50,0.75\}$ from the approximate QMP } 
\label{fig:reg_small}
\end{figure}
\begin{figure}[!h]
\begin{center}
\includegraphics[width=\textwidth]{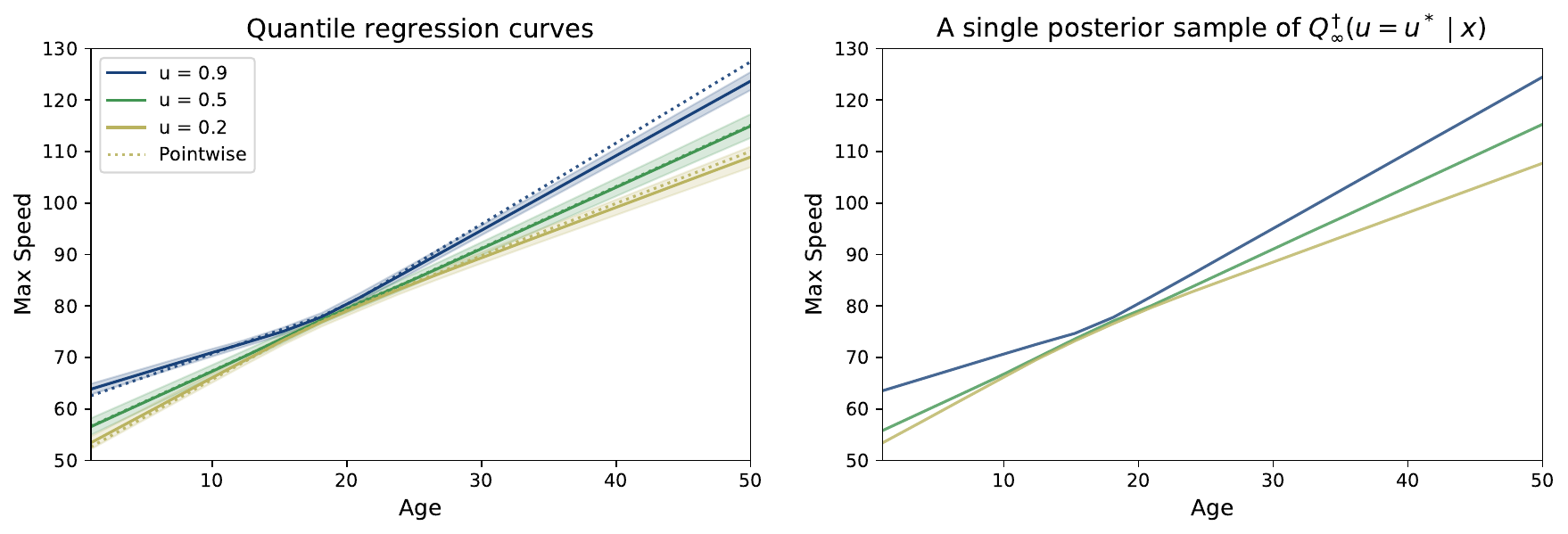}
\end{center}
\caption{Tropical cyclone maximum speeds globally ($n = 2097$): (Left) Posterior mean and 95\% credible intervals for $Q^{\dagger}_\infty(u = u^* \mid x)$ for $u^* \in \{0.25,0.50,0.75\}$; (Right) A single posterior sample of $Q^{\dagger}_\infty(u = u^* \mid x)$ for $u^* \in \{0.25,0.50,0.75\}$ from the approximate QMP } 
\label{fig:reg_big}
\end{figure}

{Figure \ref{fig:reg_small_cred} (left)} illustrates the posterior means and 95\% credible intervals of $\beta_{\infty,1}(u)$ corresponding to the year, in comparison to that from the DQP. In general, we caution against interpreting $\beta_{1\infty}(u)$ directly, as we generally do not have $Q_\infty^{\dagger}(u \mid x) \neq \beta_\infty(u)^{T} x$. However, in this specific case, the smoothness of the update resulted in no rearrangement being required for both $Q_n(u \mid x)$ and $Q_\infty(u \mid x)$. We see here that the exact and approximate QMP are numerically indistinguishable, so again it seems that $n$ is already sufficiently large for the GP approximation to hold. In comparison to the DQP, we see that the QMP has wider credible intervals within the centre but narrower in the tails, and the QMP posterior mean is also more regularized towards the initial $\beta_{01}(u) = 0$ than the DQP. {Figure \ref{fig:reg_small_cred} (right)} illustrates the exact and approximate QMP over the linear regression coefficient $\bar{\beta}_\infty$, where again the exact and approximate QMP are very similar.

{Figure \ref{fig:reg_small} (left)} illustrates the posterior mean and 95\% credible intervals for the conditional quantile functions for the earliest and latest year, i.e. $Q_\infty^{\dagger}(u \mid x = 1981)$ and $Q_\infty^{\dagger}(u \mid x = 2016)$. As mentioned earlier, in this specific case, no rearrangement was necessary as the updates are sufficiently smooth. The conditional quantiles are smooth and monotonic, and again deviate from the quantile regression estimate for values of $u$ near $1$, as it is regularized more towards the linear $Q_0(u)$. We see that there is an increasing trend in maximum speed with year, with a more significant difference for small and large values of $u$. {Figure \ref{fig:reg_small} (right)} illustrates the posterior means and 95\% credible intervals of quantile regression curves $Q_\infty^{\dagger}(u  = u^*\mid x)$ at $u^* \in \{0.25,0.50,0.75\}$, which in this case are linear and non-crossing. Again, we see that there is an increasing trend which is larger for values of $u$ near $1$.

We now study the full data set $(n = 2097, d = 3)$ with the year, latitude and cyclone age as covariates, where we exclude the basin indicator due to strong collinearity with latitude.  We do not compute the DQP posterior due to the computational expense. The QMP required 5.2s for tuning the hyperparameter ($c = 0.95$) and estimating $Q_n^{\dagger}$, which can be accelerated if fewer data permutations are used. Exact and approximate predictive resampling then required a further 42s and 0.7s respectively, where again we only display results for the approximate QMP as they are visually indistinguishable from the exact sampler. In this case, the cyclone age is the most significant predictor of maximum speed. 
{Figure \ref{fig:reg_big} (left)} shows the posterior mean and 95\% credible intervals for $Q(u = u^* \mid x)$ for $u^* \in \{0.2,0.5,0.9\}$, where we fix the year and latitude at the respective sample means and only vary age. We see that the credible intervals are tighter, and again the QMP agrees with the pointwise QR for $u = 0.5$ and $u = 0.2$  but is regularized towards $Q_0$ for $u = 0.9$. In this setting, the effect of increasing rearrangement is clear: the posterior mean of the quantile regression curves are non-crossing but are no longer linear, and {Figure \ref{fig:reg_big} (right)} shows a single posterior sample of $Q_\infty^{\dagger}(u= u^* \mid x)$ for different values of $u$, which also do not cross.

\section{Discussion and extensions} \label{sec:discussion}
In this paper, we introduce the quantile martingale posterior (QMP), which is a method for nonparametric Bayesian quantile estimation/regression based on a solely predictive framework, where we focus on the smooth case. Model specification only requires an estimate of the (conditional) quantile function, which does not need to be monotonic, as we rely on increasing rearrangement which naturally arises from predictive resampling. One main advantage of the QMP is that we no longer need to specify a likelihood or a prior distribution, which is complex in the quantile estimation/regression case. Another key advantage is computational cost - we can carry out exact posterior sampling without MCMC, where we are orders of magnitude faster and free of convergence challenges. By relying on an asymptotic Gaussian process approximation of the QMP, we can accelerate posterior sampling even further. Compared to the original martingale posterior, the space of quantile function estimates is also easier to work with for the theory.
However, this gain in flexibility of model specification and computational speed comes at a cost of being less `automatic' than traditional Bayesian inference. Significant effort is needed to show the existence, support and consistency/contraction rate of the QMP, and there are still some gaps in the theory for the regression case. Furthermore, careful specification of the learning rate and bandwidth sequence are needed to achieve good results, which is a limitation of the recursive approach. We now discuss some potential future directions to alleviate some of these limitations.

\subsection{Functional learning rates}\label{sec:functional}

Throughout the paper, we hinted at the inherent limitation of a scalar learning rate $a$, resulting in sub-optimal estimation of the quantile function near $u = 0$ and $u = 1$, as well as the need to inflate posterior uncertainty for central values of $u$ to compensate for anticonservative uncertainty in the tails. A potential extension of the QMP to tackle this limitation is to introduce a functional learning rate $a(u)$ which depends on $u$, allowing for a slower and faster learning rate in the center and tails respectively. In Section \ref{app:sec_func} of the Appendix, we show that under some assumptions on $a(u)$, this does not affect posterior consistency. We also conjecture that attaining a posterior contraction rate of $n^{-1}$ can be attained under more reasonable hyperparameter settings, but leave this for future work. To guide the setting of $a(u)$, we note that the asymptotic variance of the empirical quantile estimate is equal to $u(1-u)\,  q^*(u)^2$ \citep{vanderVaart2000}, where $q^*(u) = 1/p^*\left(Q^*(u)\right)$ is the quantile density function. This hints at an appropriate choice of $a(u) = q^*(u)$, which is also suggested in \cite{Aboubacar2014}. One downside of this approach is the need to separately estimate a density function, which is somewhat unsatisfying from a coherence point of view. Furthermore, the posterior uncertainty of the QMP will be very sensitive to the tails of the estimated density, as posterior variance will be proportional to the reciprocal of the density, and the tails are difficult to estimate. In the Appendix, we also explore an example where we estimate $p^*$ using a kernel density estimate, but leave a proper investigation for future work.

\subsection{Multivariate data and non-linear quantiles}\label{sec:nonlinear}
In this paper, we focused on the case where $y$ is univariate and the conditional quantiles are linear in $x$. However, the predictive asymptotics extends to the case where $Q_N(u)$ is multivariate, which is also hinted at when we studied the vector of quantile regression coefficients. As a result, an extension to the multivariate case, where $Q_N(u)$ is a generative predictive, may be of interest. The challenge here is then to design a recursive update, where we may want to leverage machine learning due to connections with generative adversarial networks \citep{Goodfellow2020}. We believe this to be a fruitful line of research where deep generative models may be used for Bayesian inference.
Extensions to increasing rearrangement within the multivariate case may also be of interest, e.g. as studied in \cite{Carlier2016,Rosenberg2022}. Another obvious extension is to nonlinear quantile regression, which in theory involves replacing $X_{n+1}$ in (\ref{eq:beta_freq}) with the gradient of a nonlinear function estimator, again overlapping with machine learning.

\newpage
\section*{Acknowledgments}
\textbf{AY} receives funding from Novo Nordisk. We thank Hyoin An for providing the code for the DQP method which we used for our experiments. 
\vspace{-4mm}

\section*{Code}
Code for reproducing the results in the paper can be found at \url{https://github.com/edfong/qmp}. \vspace{-4mm}

\bibliographystyle{abbrvnat}
\bibliography{paper-ref}

\newpage
\begin{appendices}
\setcounter{equation}{0}
\renewcommand\theequation{A\arabic{equation}}
\setcounter{algocf}{0}
\renewcommand\thealgocf{A\arabic{algocf}}
\setcounter{figure}{0}
\renewcommand\thefigure{A\arabic{figure}}

\section{Prerequisite theory}\label{app:sec_prereq}
In this section, we provide overviews of a few key topics along with key results that are necessary for the proof of the main results. 

\subsection{Banach space valued martingales}\label{app:sec_ban_mart}

In this section, we introduce Banach space valued martingales, and provide the key theorem on martingale convergence with reference to the seminal book of \cite{Pisier2016} on Banach-valued martingales. 
We begin with a summary of expectations in Banach spaces, but omit details on Bochner integrals which can be found in references such as \cite{Yosida2012} and \cite{Hytonen2016}.

Let $B$ be a Banach space of real-valued functions $g:(0,1) \to \R$ with  norm $\|\cdot \|_B$. Let $(\Omega,\mathcal{F},\mathbb{P})$ denote the probability space. A random variable (r.v.) in this case is a function $f: \Omega \to B$ which is Bochner (or strongly) measurable. As we will only be working with separable spaces, Pettis' theorem implies that weak and strong separability are equivalent, so we only need to check scalar measurability of $Tf: \Omega \to \R$ for every continuous linear operator $T: B \to \R$.  
The Bochner integral generalizes the Lebesgue integral to Banach spaces by constructing a sequence of simple functions which converge to $f$ pointwise.
If $f$ is Bochner integrable, then we write the expectation of $f$ as the Bochner integral relative to $\mathbb{P}$, that is
$
\E[f] = \int f d\mathbb{P},
$
where $\E[f]$ is an element of $B$.
Note that $f$ is Bochner integrable if and only if $\int \|f\|_B \, d\mathbb{P} < \infty$, which involves checking that the real-valued function $\|f\|_B$ is integrable in the traditional sense. For every continuous linear operator $T: B \to \R$, the expectation satisfies
$
 T \int f d\mathbb{P}= \int Tf d\mathbb{P}.
$

We now introduce the Banach space valued version of $L^p$ spaces for r.v.s.
We write $L^p\left(\Omega, \mathcal{F}, \mathbb{P}; B\right)$ or $L^p(B)$ as the space of (equivalence classes of) Bochner measurable functions with $\int \|f\|_B^p d\mathbb{P} < \infty$ for some $1 \leq p < \infty$; we will mostly be using $p = 2$. The norm of this $L^p$-space, which is also a Banach space, is then defined as $\|f\|_{L^p(B)} = \left(\int \|f\|_B^p d\mathbb{P}\right)^{1/p}$. A realization of $f$, i.e. $f(\omega)$ for some $\omega \in \Omega$, can be interpreted as a random function in $B$.

The conditional expectation can be analogously defined, e.g. \citet[Chapter 1.2]{Pisier2016} or \citet[Chapter 5]{Diestel1977}. Let $\mathcal{A}\subseteq \mathcal{F}$ denote a sub-$\sigma$-algebra and $f$ a Bochner integrable r.v. as before. The conditional expectation of $f$ given $\mathcal{A}$
is then the $B$-valued $\mathcal{A}$-measurable r.v. $\E^{\mathcal{A}}\left[f\right]$ which satisfies $\int_A \E^{\mathcal{A}}\left[f\right]\, d\mathbb{P} = \int_A f \, d \mathbb{P}$
for all $A \in \mathcal{A}$, which exists and is unique up to the null set of $\mathbb{P}$. The conditional expectation also satisfies $\E^{\mathcal{A}}\left[Tf\right]= T\,\E^{\mathcal{A}}\left[f\right]$ for any continuous linear operator $T: B \to \R$, e.g. \cite[Remark 1.11]{Pisier2016}.

A martingale in $B$ is then the extension of the regular martingale as follows. Let $\{\mathcal{F}_i\}_{i \geq 0}$ denote a filtration, and define $\mathcal{F}_\infty := \sigma\left(\cup_{i \geq n}\mathcal{F}_i\right)$. A sequence of random functions $\{f_i\}_{i \geq 1}$ in $L_1(B)$ is then a Banach space valued martingale if $f_i$ is $\mathcal{F}_i$-measurable and $\E^{\mathcal{F}_i}\left[f_{i+1}\right] = f_i$ a.s. for each $i \geq 0$.  A detailed overview can be found in  \citet[Chapter 1.3]{Pisier2016}. 

We now require a technical result on Hilbert-valued martingales, which is a direct specialization of \citet[Theorem 2.9]{Pisier2016} to the case where $B$ is a Hilbert space. Hilbert spaces automatically have the Radon-Nikodym property, so we have martingale convergence given a boundedness condition. We will be leveraging the below result to show existence of the QMP.
\begin{thm}[{\cite[Theorem 2.9]{Pisier2016}}]\label{app:thm_hilbert_mart}
    Let $\{f_i\}_{i \geq 0}$ be a Banach space valued martingale as defined above. Further assume that $B$ is in fact a Hilbert space. For some $p \geq 1$, if $\sup_i \|f_i\|_{L^p(B)} < \infty$, 
    then there exists an $\mathcal{F}_\infty$-measurable $f_\infty \in L^p(B)$ such that 
    $\| f_i - f_\infty \|_{B} \to 0$ a.s. and $\|f_i - f_\infty\|_{L^p(B)} \to 0$.  
\end{thm}

\subsubsection{$L^2((0,1))$ spaces}
Suppose again that $f: \Omega \to B$ is Bochner integrable. Bochner integrals are not usually computed explicitly in practice, but we will require this later to check for martingale conditions. We thus outline how evaluating expectations pointwise on random functions suffices in Hilbert spaces.
One approach is to utilize the fact that for two elements in $B$, $x = y$ if and only if $Tx = Ty$ for all continuous linear operators $T: B \to \R$, i.e. $T \in B^*$ where $B^*$ is the dual space of $B$. 
In the specific case where $B$ is a Hilbert space, we have from the Riesz representation theorem that for each $T$, there exists an $h_T\in B$ such that $Tf = \langle h_T, f \rangle_{L^2}$, where $\langle \cdot, \cdot \rangle_{L^2}$ is the inner product of the Hilbert space. As a result, the Bochner integral satisfies
\begin{align*}
    \langle h_T, \E\left[f \right]  \rangle_{L^2} =\E\left[ \langle h_T, f \rangle_{L^2}\right] 
\end{align*}
for each $T \in B^*$. Consider the case where $B = L^2((0,1))$, which is a separable Hilbert space consisting of functions $g:(0,1) \to \R$ which are bounded in $L^2$. The above can then be written as
\begin{align*}
   \langle h_T, \E\left[f \right]  \rangle_{L^2}  &= \E\left[\int_0^1 h_T(u)\, f(u)\, du\right]
   = \int_0^1 h_T(u)\,\E\left[ f(u)\right]\, du,
\end{align*}
where we have used the linearity of the Bochner integral in the first equality and Fubini's theorem in the second. 
As a result, for $B = L^2((0,1))$, it suffices to compute expectations pointwise. We verify this formally for the martingale condition in {Lemma \ref{app:lem_bochnerl2}}.

\subsubsection{Sobolev spaces}\label{app:sec_sobolev}
The second separable Hilbert space that we will consider is the Sobolev space; see \cite{Leoni2017} for a thorough exposition. To begin, consider a function $g \in L_1((0,1))$. The function $g' \in L_1((0,1))$ is 
is a first-order weak derivative of $g$ if it satisfies
$$\int_0^1 \, g(u) \, \psi'(u) \, du = - \int_0^1 g'(u) \psi(u) du$$
for all $\psi$ which are infinitely differentiable with $\psi(0) = \psi(1) = 0$. 

In particular, we will consider $W^{1,2}((0,1)) = H^1((0,1))$, which consists of the subset of functions $g\in L^2((0,1))$ which have first-order weak derivatives $g' \in L^2((0,1))$. The inner product is
\begin{align*}
    \langle g,h \rangle_{H^1} = \langle g,h\rangle_{L^2} + \langle g',h'\rangle_{L^2},
\end{align*}
and the norm is then simply 
$$\|g\|_{1,2} = \sqrt{\int_0^1 g(u)^2\, du + \int_0^1 {g'(u)}^2\, du.}$$

Sobolev spaces are intimately related to absolutely continuous functions. An absolutely continuous function $\bar{g}:(0,1)\to \R$ is differentiable almost everywhere, where its derivative satisfies $\bar{g}' \in L_1((0,1))$ and 
\begin{align*}
    \bar{g}(u) = \bar{g}(a) + \int_a^u \bar{g}'(t)\, dt
\end{align*}
for any $a,u \in (0,1)$. A very useful property of the space $H^1((0,1))$ in the univariate case is the following. 
\begin{prop}[{\cite[Theorem 7.16]{Leoni2017}}]\label{app:prop_abcont}
    Suppose $g:(0,1) \to \R$. If $g\in H^1((0,1)) = W^{1,2}((0,1))$, then there exists an absolutely continuous function $\bar{g}:(0,1) \to \R$ where $g = \bar{g}$ almost everywhere.
    Furthermore, both $\bar{g}$ and its regular derivative $\bar{g}'$ lie in $L^2((0,1))$, and $\bar{g}$ is H{\"o}lder continuous with exponent $\alpha = 1/2$. 
\end{prop}
As a result, $H^1((0,1))$ is a very appropriate choice for the space of quantile estimates, as it contains absolutely continuous functions (or at least with an absolutely continuous representative). Furthermore, it is a Hilbert space so we can apply the martingale limit theorem with ease. We will shortly see that Sobolev spaces play nicely with montone rearrangement as well. 

Another useful property is the following.
\begin{prop}\label{app:prop_bounded}
    Let $f \in H^1((0,1))$. Then $f$ is essentially bounded, that is $\|f\|_{\infty} < \infty$ where
    \begin{align*}
      \|f\|_{\infty}  = \inf\{M: f(u) \leq M \textnormal{ for Lebesgue-almost all } u \in (0,1) \}.
    \end{align*}
\end{prop}
\begin{proof}
This follows directly from \citet[Theorem 7.34]{Leoni2017} with $I = (0,1)$, as $f \in H^1(I)$ implies $f \in W^{1,1}_{\textnormal{loc}}(I)$, which is the space of locally integrable functions with locally integrable weak derivatives. Choosing $p = q = 2$, $r = \infty$, $\ell = 1/4$, the theorem gives 
\begin{align*}
    \|f\|_{\infty} \leq 2\|f\|_{2} + \frac{1}{2}\|f'\|_{2} \leq 2\|f\|_{1,2}
\end{align*}
which gives the result.
\end{proof}

We now consider a Bochner integrable r.v. $f: \Omega \to H^1$, where we omit the domain $(0,1)$ for brevity. Again, as $H^1$ is a Hilbert space, for each continuous linear operator $T \in (H^1)^*$, we have an element $h_T \in H^1$ with weak derivative $h_T' \in L^2((0,1))$ which satisfies
\begin{align*}
    \langle h_T, \E[f]\rangle_{H^1} &= \E\left[\int_0^1 h_T(u)\, f(u)\,du\right] +  \E\left[\int_0^1 h'_T(u)\, f'(u)\, du\right] \\
    &=\int_0^1 h_T(u)\, \E[f(u)] \,du +  \int_0^1 h'_T(u)\, \E[f'(u)]\, du 
\end{align*}
which follows from linearity and Fubini's theorem again. Once again, we can just compute the pointwise expectations of $f$ and its (weak) derivative $f'$. We also verify this formally for the martingale condition in {Lemma \ref{app:lem_bochnerh1}}. 

\subsection{Rearrangement}
\subsubsection{Decreasing rearrangement}
In this subsection, we state and show some useful properties of increasing rearrangement. Most of the literature concerns the \textit{decreasing} rearrangement of functions, so we will make explicit the connection to increasing rearrangement. We first introduce decreasing rearrangement, and recommend  \cite{Kesavan2006} and \citet[Chapter 4]{Leoni2017} for more details. Let $f: [0,1] \to [0,K]$ be a Lebesgue measurable function, where $0<K< \infty$. 
The distribution function $S: [0,\infty) \to [0,1]$ of $f$ is defined as
\begin{align*}
    S(y) = \int_0^1 \mathbbm{1}\left(f(u) > y\right)du
\end{align*}
The bounded range of $f$ can be relaxed to $K = \infty$ as long as $f$ vanishes at infinity, which means that $S(y) < \infty$ for every $y > 0$ (and $f$ is Lebesgue measurable), but we will not need that here. 
From \citet[Proposition 4.1]{Leoni2017}, the distribution function is decreasing and right continuous, and clearly we have $0\leq S(y) \leq 1$ with $S(y) = 0$ for all $y \geq K$. The \textit{decreasing} rearrangement of $f$, which we write as $f_{\dagger}:[0,1] \to [0,K]$, is the left inverse of the distribution function, that is
\begin{align*}
    f_{\dagger}(u):= \inf \{y \in [0,K]: S(y) \leq u\}.
\end{align*}
The existence of $f_{\dagger}$ follows as $S$ is decreasing and bounded from below. From \citet[Proposition 4.3]{Leoni2017}, we have that $f_{\dagger}$ is also decreasing and right continuous. Another useful property is the equimeasurable property, that is for all $y \geq 0$, we have
\begin{align*}
    \int_0^1 \mathbbm{1}\left(f_{\dagger}(u) > y\right)\, du =  \int_0^1 \mathbbm{1}\left(f(u) > y\right)\, du.
\end{align*}
In fact, this equimeasurability holds more generally, which will be useful later on.
\begin{lem}[{\cite[Theorem 4.16]{Leoni2017}}]\label{app:lem_equi}
    Let $f: [0,1] \to [0,K]$ and let $h:[0,\infty) \to [0,\infty)$ be a Borel measurable function. We then have
    \begin{align*}
        \int_0^1 h(f(u))\, du = \int_0^1 h(f_{\dagger}(u))\, du.
    \end{align*}
\end{lem}
\begin{proof}
    We have the result of  \citet[Theorem 4.16]{Leoni2017} with equality as $(0,1)$ has finite Lebesgue measure.
\end{proof}

Perhaps the most useful property of decreasing rearrangement for estimation is the following inequality, of which there are many generalizations.
\begin{prop}[{\cite[Theorem 4.19]{Leoni2017}}]\label{app:prop_L2}
Let $f,g: [0,1] \to [0,K]$ with respective decreasing rearrangements $f^{\dagger},g^{\dagger}$. We then have 
\begin{align*}
  d_2(f_{\dagger}, g_{\dagger}) \leq d_2(f,g)
\end{align*}
where $d_2(f,g) = \sqrt{\int_0^1 (f(u) - g(u))^2\, du}$ is the $L^2$ norm.
\end{prop}
The above will help us later when considering the convergence of rearranged quantile estimates, and was used extensively by \citet{Chernozhukov2010}. Essentially, the above states that the decreasing rearrangement is continuous from $L^2$ to itself. Actually, the above proposition can be weakened to the case where $f,g \in L^2((0,1))$ \cite[Theorem 1.2.3]{Kesavan2006}, but we will not need that here.

As mentioned earlier, rearrangement works nicely with Sobolev spaces, as rearrangement has a regularization effect on the function. The well-known result below formalizes this.
\begin{thm}[{\cite[Theorem 4.22]{Leoni2017}}]\label{app:thm_sobolev}
Let $f:[0,1] \to [0,K]$. If $f$ is absolutely continuous on $[0,1]$ with weak derivative $f'$, then $f_{\dagger}$ is also absolutely continuous on $[0,1]$ with weak derivative $f'_{\dagger}$. Furthermore, we have
\begin{align*}
    \|f_{\dagger}'\|_{2} \leq \|f'\|_{2}.
\end{align*}
\end{thm}
We can apply {Lemma \ref{app:lem_equi}} with $h(x) = x^2$, which gives $\|f_{\dagger}\|_{2}^2= \|f\|_{2}^2$, and together with the above gives $\|f_{\dagger}\|_{1,2} \leq \|f\|_{1,2}$. In other words, 
decreasing rearrangement decreases the Sobolev norm, so it has a regularizing effect. In the univariate case, \cite{Coron1984} showed the stronger result that the \textit{symmetric} decreasing rearrangement is also continuous from $W^{1,p}(\R)$ to itself, which hints at an extension of {Proposition \ref{app:prop_L2}} to the Sobolev norm (with the \textit{nonsymmetric} decreasing rearrangement), but we leave that for future work.

\subsubsection{Increasing rearrangement}
Our interest is actually on \textit{increasing} rearrangement, and on functions with both positive and negative support. Consider then a function $Q:[0,1] \to C$, where $C \subset \R$ is a bounded subset of the real line. This is also assumed in \cite{Chernozhukov2009,Chernozhukov2010}.  In the main paper, we introduced the increasing distribution function $P: C \to [0,1]$ as the familiar cumulative distribution function, 
\begin{align*}
    P(y) = \int_0^1 \mathbbm{1}\left(Q(u) \leq y\right)\,du,
\end{align*}
with the increasing rearrangement $Q^{\dagger}: [0,1] \to C$ as 
\begin{align*}
    Q^{\dagger}(u) = \inf\{y \in C: P(y)\geq u\}.
\end{align*}
Here, we have that $P$ is increasing and right continuous, as $P(y) = 1 - S(y)$ where $S(y)$ is decreasing and right continuous. This then suggests that $Q^{\dagger}$ is increasing and left continuous, as expected.

Suppose $C$ is an interval, which is bounded so we can write $C = [-a,b]$ for positive and finite constants $a,b$. Most results for rearrangement are stated for non-negative $f$, so it is helpful to carry out a translation.
\begin{lem}
For $Q: [0,1] \to [-a,b]$, where $a,b$ are finite and positive constants, let $Q_+ := Q+a$ be the translated non-negative function.  We then have $Q^{\dagger}(u) = Q_+^{\dagger}+a$, where $Q^{\dagger}$ and $Q_+^{\dagger}$ are the increasing rearrangements of $Q$ and $Q_+$ respectively.
\end{lem}
\begin{proof}
   For $y \in [0,a+b]$, we have
\begin{align*}
    P_+(y) &= \int_0^1\mathbbm{1}\left(Q(u) + a \leq y\right)\,du=P(y - a)
\end{align*}
Similarly, for  $y' = y- a$, we have
\begin{align*}
    Q_+^{\dagger}(u) &=\inf\{y \in [0,a+b]: P(y-a) \geq u\}\\
    &= \inf\{y' \in C: P(y') \geq u\}  + a = Q^{\dagger}(u) + a.
\end{align*}
\end{proof}
  As a result, we can just assume that $Q: [0,1] \to [0,K]$ without loss of generality for the remainder of this section. To leverage the results on decreasing rearrangement, we fortunately have a simple relationship between the increasing and decreasing rearrangement.
\begin{lem}{\label{app:lem_inc_dec}}
Let $Q: [0,1] \to [0,K]$ for some finite and positive $K$, and let $Q_{\dagger}$ and $Q^{\dagger}$ denote its decreasing and increasing rearrangement respectively. Then we have
\begin{align*}
    Q^{\dagger}(u) = Q_{\dagger}(1-u)
\end{align*}
for all $u \in [0,1]$.
\end{lem}
\begin{proof}
  Again, we have $S(y) = 1 - P(y)$,
which for each $u \in [0,1]$ gives
\begin{align*}
    Q^{\dagger}(u) &= \inf\{y \in [0,K]: 1- S(y)\geq u\}\\
    &= \inf\{y \in [0,K]: S(y) \leq 1- u\}\\
    &= Q_{\dagger}(1-u).
\end{align*}  
\end{proof}
This connection is also commented on \citet[Section 1]{Korenovskii2007} and \citet[Exercise 1.4.1]{Kesavan2006}. This allows us then to directly apply all the results of the previous subsection, which we state formally for completion. 
\begin{cor}
    For any $f,g:[0,1] \to [0,K]$, {Lemma \ref{app:lem_equi}}, {Proposition \ref{app:prop_L2}} and {Theorem \ref{app:thm_sobolev}} all apply if all instances of $f_{\dagger}$, $f'_{\dagger}$ and $g_{\dagger}$ are replaced with $f^{\dagger}$, ${f'}^{\dagger}$ and $g^{\dagger}$ respectively.
\end{cor}
\begin{proof}
    For all appropriate integrals involving decreasing rearrangements, substitute $f_{\dagger}(u)$ with $f^{\dagger}(1-u)$ (and likewise for $g_{\dagger}, f'_{\dagger}$) and carry out a change of variables to $u' = 1-u$, which has Jacobian determinant $1$ and integration limits $u' \in (0,1)$. 
\end{proof}
It is perhaps not too surprising as the increasing rearrangement is also equimeasurable, so it will very similar properties to the decreasing rearrangement.

\subsection{Empirical process theory}\label{sec:emp}
In this section, we show an auxiliary empirical process result that we require for showing asymptotic tightness later.
The weak $L^{2}$-pseudonorm of a variable $X \sim \Pr$ is defined as
\begin{equation*}
    \Vert X \Vert_{\Pr,2,\infty} = \sup_{x > 0} x \Pr(|X| > x)^{1/2}.
\end{equation*}
Note that it is upper-bounded by the $L^{2}(\Pr)$ norm:
\begin{equation*}
    \Vert X \Vert_{\Pr,2,\infty} \leq \Vert X \Vert_{\Pr,2} = (\E_{\Pr}[X^{2}])^{1/2}.
\end{equation*}
This is because for any value of $x > 0$, we have
\begin{equation*}
    x^{2}\E_{\Pr}[\mathbbm{1}(|X| > x)] \leq x^{2}\E_{\Pr}\left[\frac{X^{2}}{x^{2}}\mathbbm{1}(|X| > x)\right] \leq \Vert X \Vert_{\Pr,2}^{2}.
\end{equation*}

\begin{lem} \label{lem::prob2.5.5}
    For any positive r.v. $X \sim \Pr$, we have the inequality
    \begin{equation*}
        \sup_{x>0} x\E_{\Pr}[X\mathbbm{1}(X > x)] \leq 2 \Vert X \Vert_{\Pr,2,\infty}^{2}.
    \end{equation*}
    This is the second inequality of Problem 2.5.5 in \cite{vanderVaart2023}.
\end{lem}

\begin{proof}
    For any value of $x$, the left-hand side of the inequality (without the supremum) can be written as
    \begin{equation} \label{eqn::weakl2}
        x\int_{t=0}^{\infty} \Pr(X\mathbbm{1}(X > x) > t)\,dt = x\int_{t=0}^{x} \Pr(X\mathbbm{1}(X > x) > t)\,dt + x\int_{t=x}^{\infty} \Pr(X\mathbbm{1}(X > x) > t)\,dt.
    \end{equation}
    For the integrand in the first term on the right-hand side, we have
    \begin{align*}
        \Pr(X\mathbbm{1}(X > x) > t) &= \E_{\Pr}[\mathbbm{1}(X\mathbbm{1}\{X > x\} > t)]\\
        &= \E_{\Pr}[\mathbbm{1}\{X > x\}] \\
        &= \Pr(X > x),
    \end{align*}
    where the second inequality follows from $t$ being less than or equal to $x$. Thus, the first term on the right-hand side of (\ref{eqn::weakl2}) is bounded above by $\Vert X \Vert_{\Pr,2,\infty}^{2}$.
    
    The integrand in the second term can be written as
    \begin{equation*}
        \Pr(X\mathbbm{1}(X > x) > t) = \frac{1}{t^{2}}t^{2}\Pr(X\mathbbm{1}(X > x) > t) \leq \frac{1}{t^{2}}\Vert X \Vert_{\Pr,2,\infty}^{2}.
    \end{equation*}
    The integral of $1/t^{2}$ from $t=x$ to $\infty$ is $1/x$. Putting the two terms together gives the result.
\end{proof}

For each $N \in \mathbb{N}$, let $\{Z_{Ni}:\,i \geq N\}$ be a sequence of independent stochastic processes indexed by a common semimetric space $(\mathcal{F}, d)$. For every $N$, define the bracketing number $\mathfrak{N}_{[]}(\varepsilon, \mathcal{F},L^{2,N})$ to be the minimal number of sets $\mathfrak{N}_{\varepsilon}^{N}$ in a partition $\mathcal{F} = \cup_{j=1}^{\mathfrak{N}_{\varepsilon}^{N}} \mathcal{F}_{\varepsilon j}^{N}$ of the index set into sets $\mathcal{F}_{\varepsilon j}^{N}$ such that, for every partitioning set $\mathcal{F}_{\varepsilon j}^{N}$, we have
\begin{equation*}
    \sum_{i=N}^{\infty} \E^{*} \sup_{f,g \in \mathcal{F}_{\varepsilon j}^{N}} |Z_{Ni}(f) - Z_{Ni}(g)|^{2} \leq \varepsilon^{2}.
\end{equation*}

We will ultimately set $Z_{Ni} =  \sqrt{N}\alpha_{i}H_{\rho_{i}}(u,V_{i})$. Then
    \begin{equation*}
        \sqrt{N}S_{N} = -\sum_{i=N}^{\infty} (Z_{Ni} - \E Z_{Ni}).
    \end{equation*}
    The space $\ell^{\infty}(\mathcal{F})$ is the set of functions $z: \mathcal{F} \rightarrow \R$ with
$\Vert z\Vert_{\mathcal{F}} = \sup_{t \in \mathcal{F}} |z(t)| < \infty$.
This is a metric space with respect to $d(z_{1},z_{2}) = \Vert z_{1}-z_{2}\Vert_{\mathcal{F}}$. 
We wish to show that $\sqrt{N}S_{N}$ is \textit{asymptotically tight}, which means that for every $\varepsilon > 0$ there exists a compact set $K$ such that $\liminf_{N \rightarrow \infty} \mathbb{P}(\sqrt{N}S_{N} \in K^{\delta}) \geq 1-\varepsilon$,
where $K^{\delta} = \{y \in \ell^{\infty}((0,1)):\,d(y,K) < \delta\}$. This can be achieved by verifying the conditions in the following general result.

\begin{thm}[Bracketing CLT with infinite sums] \label{th::bracketing_clt}
    Suppose that $(\mathcal{F},d)$ is totally bounded and each $Z_{Ni}$ has a finite second moment. Suppose also that
    \begin{align*}
        \sum_{i=N}^{\infty} \E^{*}\Vert Z_{Ni}\Vert_{\mathcal{F}}\mathbbm{1}\{\Vert Z_{Ni}\Vert_{\mathcal{F}} > \eta \} &\rightarrow 0 \quad \text{for every $\eta > 0$,} \\
        \sup_{d(f,g) < \delta_{N}} \sum_{i=N}^{\infty} \E[(Z_{Ni}(f)-Z_{Ni}(g)^{2}] &\rightarrow 0\quad \text{for every $\delta_{N} \downarrow 0$,} \\ 
        \int_{0}^{\delta_{N}} \sqrt{\log \mathfrak{N}_{[]}(\varepsilon,\mathcal{F}, L^{2,N}})\,d\varepsilon &\rightarrow 0\quad \text{for every $\delta_{N} \downarrow 0$.} 
    \end{align*}
    Then the sequence $\sum_{i=N}^{\infty}(Z_{Ni}-\E Z_{Ni})$ is asymptotically tight in $\ell^{\infty}(\mathcal{F})$ and converges in distribution provided it converges marginally. If the partitions can be chosen independent of $N$, then the middle of the displayed conditions is unnecessary.
\end{thm}
\begin{proof}
    The proof of this result mostly follows that of Theorem 2.11.9 in \cite{vanderVaart2023}. The crucial difference lies with the application of Bernstein's inequality, which is restricted to finite sums of variables.

    Under the conditions of the theorem, there exists for every $N$ a sequence of nested partitions $\mathcal{F} = \cup_{j=1}^{\mathfrak{N}_{2^{-q}}^{N}}\mathcal{F}_{qj}^{N}$ such that for every $j$ and $N$,
    \begin{align} \label{eqn::entropy_sum}
        \lim_{q_{0} \rightarrow \infty} \limsup_{N \rightarrow \infty} \sum_{q > q_{0}} 2^{-q-1} \sqrt{\log \mathfrak{N}_{2^{-q}}^{N}} &= 0, \\
        \label{eqn::covering} \sup_{f,g \in \mathcal{F}_{q j}^{N}} \sum_{i=N}^{\infty} \E\{Z_{Ni}(f) - Z_{Ni}(g)\}^{2} &\leq 2^{-2q},\\
        \label{eqn::weak_l2norm}\sum_{i=N}^{\infty} \sup_{t_{i}} t_{i}^{2} \Pr^{*}\left( \sup_{f,g \in \mathcal{F}_{q j}^{N}} |Z_{Ni}(f) - Z_{Ni}(g)| > t_{i}\right) &\leq 2^{-2q}.
    \end{align}
    Equation (\ref{eqn::entropy_sum}) above could be viewed as a lower-bound histogram approximation to the entropy integral
    $$\int_{\varepsilon = 0}^{2^{-q_{0}}}\sqrt{\log \mathfrak{N}_{\varepsilon}^{N}}\,d\varepsilon.$$ Equations (\ref{eqn::covering}) and (\ref{eqn::weak_l2norm}) follow from the same counting argument as the proof of Theorem 2.5.8 in \cite{vanderVaart2023}.

    Choose an element $f_{qj}$ from each partitioning set $\mathcal{F}^{N}_{qj}$ and define
    \begin{align*}
        \pi_{q}f &= f_{qj}, \\
        (\Delta_{q}f)_{Ni} &= \sup_{g,h \in \mathcal{F}_{q j}^{N}} |Z_{Ni}(g) - Z_{Ni}(h)|, \quad \text{if }f \in \mathcal{F}^{N}_{qj}\\
        a_{q} &= 2^{-q} \bigg/ \sqrt{\log \mathfrak{N}_{2^{-(q+1)}}^{N}}.
    \end{align*}
    We interpret $f_{qj}$ as the ``representative'' of the partitioning set $\mathcal{F}^{N}_{qj}$, and $\pi_{q}$ projects $f$ onto the representative that shares its partitioning set. Also, $(\Delta_{q}f)_{Ni}$ is the maximum distance between two points on $Z_{Ni}$ evaluated within the partitioning set that contains $f$. 
    For $q > q_{0}$, define indicator functions
    \begin{align*}
        (A_{q-1}f)_{Ni} &= \mathbbm{1}\{(\Delta_{q_0}f)_{Ni} \leq a_{q_0},\ldots, (\Delta_{q-1}f)_{Ni} \leq a_{q-1}\} \\
        (B_{q-1}f)_{Ni} &= \mathbbm{1}\{(\Delta_{q_0}f)_{Ni} \leq a_{q_0},\ldots, (\Delta_{q-1}f)_{Ni} \leq a_{q-1}, (\Delta_{q}f)_{Ni} > a_{q}\} \\
        (B_{q_0}f)_{Ni} &= \mathbbm{1}\{(\Delta_{q_0}f)_{Ni} > a_{q_0}\}.
    \end{align*}
    We wish to show that
    \begin{equation} \label{eqn::asymp_tight}
        \lim_{q_0 \rightarrow \infty} \limsup_{n \rightarrow \infty} \E^{*} \left\Vert \sum_{i=N}^{\infty} (Z_{Ni}^{\circ}(f) - Z_{Ni}^{\circ}(\pi_{q_0}f))\right\Vert_{\mathcal{F}} = 0
    \end{equation}
    for the centred processes $Z_{Ni}^{\circ}$, such that Theorem 1.5.6 of \cite{vanderVaart2023} implies asymptotic tightness in the case where the partitions do not depend on $N$. To achieve this, we consider the decomposition
    \begin{align*}
        Z_{Ni}(f) - Z_{Ni}(\pi_{q_0}f) = (Z_{Ni}(f) - Z_{Ni}(\pi_{q_0}f))(B_{q_0}f)_{Ni} &+ \sum_{q > q_{0}}(Z_{Ni}(f) - Z_{Ni}(\pi_{q}f)(B_{q}f)_{Ni}\\
        &+ \sum_{q > q_{0}}(Z_{Ni}(\pi_{q}f) - Z_{Ni}(\pi_{q-1}f)(A_{q-1}f)_{Ni}.
    \end{align*}
    For each of the three terms on the right-hand side separately, we centre at zero expectation, sum from $i=N$ to $\infty$ and take the supremum over $\mathcal{F}$. It is sufficient to then show that each of the resulting three expressions converge to zero in mean as $n \rightarrow \infty$ followed by $q_{0} \rightarrow \infty$.

    As argued by \cite{vanderVaart2023}, the Lindeberg condition implies that there is no loss in generality in assuming that $\Vert Z_{Ni} \Vert_{\mathcal{F}} \leq \eta{N}$ for all $i \geq N$ for some sequence of numbers $\eta_{N} \downarrow 0$. This implies that $(\Delta_{q}f)_{Ni} \leq 2\eta_{N}$ for all $i \geq N$, and the first expression is zero as soon as $2\eta_{N} \leq a_{q_0}$. Condition (\ref{eqn::entropy_sum}) implies that $a_{q_{0}}$ is bounded away from 0 for any fixed $q_{0}$ as $N \rightarrow \infty$. If this were not the case, then we must have
    \begin{equation*}
        \limsup_{N \rightarrow \infty} \log \mathfrak{N}_{2^{-(q_{0} + 1)}}^{N} = \infty.
    \end{equation*}
    And we have
    \begin{equation*}
        \log \mathfrak{N}_{2^{-q}}^{N} \geq \log \mathfrak{N}_{2^{-(q_{0} + 1)}}^{N}
    \end{equation*}
    for all $q \geq q_{0} + 1$. Thus, this would imply that
    \begin{equation*}
        \limsup_{N \rightarrow \infty} \sum_{q > q_0} 2^{-q-1} \sqrt{\log \mathfrak{N}_{2^{-q}}^{N}} = \infty,
    \end{equation*}
    which contradicts the condition. We deduce that $2\eta_{N} \leq a_{q_0}$ for all sufficiently large $n$.

    For the second expression, we start by noting that $(\Delta_{q}f)_{Ni}(B_{q}f)_{Ni} \leq (\Delta_{q-1}f)_{Ni}(B_{q}f)_{Ni} \leq a_{q-1}$ by the nesting of the partitions and the definition of $B_{q}f$. It follows that
    \begin{align*}
        |Z_{Ni}(f) - Z_{Ni}(\pi_{q}f)|(B_{q}f)_{Ni} &\leq (\Delta_{q}f)_{Ni}(B_{q}f)_{Ni}  \leq a_{q-1}\\
        \text{Var}\left[\sum_{i=N}^{\infty}(\Delta_{q}f)_{Ni}(B_{q}f)_{Ni}\right] &\leq \sum_{i=N}^{\infty}\E\left[(\Delta_{q}f)^{2}_{Ni}(B_{q}f)^{2}_{Ni}\right]\\
        &\leq a_{q-1}\sum_{i=N}^{\infty}\E\left[(\Delta_{q}f)_{Ni}\mathbbm{1}\{(\Delta_{q}f)_{Ni} > a_{q}\}\right] \\
        &\leq 2 \frac{a_{q-1}}{a_{q}} 2^{-2q},
    \end{align*}
    where the last inequality above follows from {Lemma \ref{lem::prob2.5.5}} and condition (\ref{eqn::weak_l2norm}). Since each summand $(\Delta_{q}f)_{Ni}(B_{q}f)_{Ni}$ is independent, the variance of the infinite sum above is greater than the variance of the partial sum that replaces $\infty$ with any finite $m > N$. For such a fixed $m$, we can apply Bernstein's inequality (e.g. Lemma 2.2.10 of \cite{vanderVaart2023}) to deduce that for every $x > 0$, we have
    \begin{equation*}
        \Pr\left(\left|\sum_{i=N}^{m}(\Delta_{q}f)_{Ni}(B_{q}f)_{Ni} -\E[(\Delta_{q}f)_{Ni}(B_{q}f)_{Ni}]\right| > x\right) \leq 2\exp\left(-\frac{\frac{1}{2}x^{2}}{2 \frac{a_{q-1}}{a_{q}} 2^{-2q}+\frac{2}{3}a_{q-1}x}\right).
    \end{equation*}
    Note that the right-hand side does not depend on $m$. We also have that the left-hand side converges to
    \begin{equation} \label{eqn::lim_prob}
        \Pr\left(\left|\sum_{i=N}^{\infty}(\Delta_{q}f)_{Ni}(B_{q}f)_{Ni} -\E[(\Delta_{q}f)_{Ni}(B_{q}f)_{Ni}]\right| > x\right)
    \end{equation}
    as $m \rightarrow \infty$ (for fixed $N$ and $q$) at all continuity points $x$, so the probability (\ref{eqn::lim_prob}) must therefore share the same exponential upper bound. In fact, the bound also holds at any discontinuity point. To see this, note that
    \begin{equation*}
        \Pr\left(\left|\sum_{i=N}^{\infty}(\Delta_{q}f)_{Ni}(B_{q}f)_{Ni} -\E[(\Delta_{q}f)_{Ni}(B_{q}f)_{Ni}]\right| > x\right) - 2\exp\left(-\frac{\frac{1}{2}x^{2}}{2 \frac{a_{q-1}}{a_{q}} 2^{-2q}+\frac{2}{3}a_{q-1}x}\right)
    \end{equation*}
    is right-continuous. If the above display is greater than 0 at some discontinuity point $x^{*}$, then it must be greater than zero for all $x$ on some interval $[x^{*},x^{*} + \delta]$ for $\delta >0$, which leads to a contradiction because the number of discontinuity points is countable. Thus, we can now apply Lemma 2.11.17 from \cite{vanderVaart2023}, and the remaining steps for handling the second expression follow the proof of Theorem 2.11.9 from \cite{vanderVaart2023}.

    The analysis of the third expression proceeds similarly. We have the following bounds:
    \begin{align*}
        |Z_{Ni}(\pi_{q}f) - Z_{Ni}(\pi_{q-1}f)|(A_{q-1}f)_{Ni} &\leq (\Delta_{q-1}f)_{Ni}(A_{q-1}f)_{Ni} \\
        & \leq a_{q-1}\\
        \text{Var}\left[\sum_{i=N}^{\infty} \{Z_{Ni}(\pi_{q}f)-Z_{Ni}(\pi_{q-1}f)\}(A_{q-1}f)_{Ni}\right] & \leq \sum_{i=N}^{\infty}\E[\{Z_{Ni}(\pi_{q}f)-Z_{Ni}(\pi_{q-1}f)\}^{2}(A_{q-1}f)_{Ni}]\\
        &\leq 2^{-2(q-1)},
    \end{align*}
    where the final inequality follows from the nesting of the partitions and condition (\ref{eqn::covering}). By applying a similar argument to before based on Bernstein's inequality, we derive the upper-bound
    \begin{align*}
        &\Pr\left(\left|\sum_{i=N}^{\infty} \{Z_{Ni}(\pi_{q}f)-Z_{Ni}(\pi_{q-1}f)\}(A_{q-1}f)_{Ni} -\E[\{Z_{Ni}(\pi_{q}f)-Z_{Ni}(\pi_{q-1}f)\}(A_{q-1}f)_{Ni}]\right| > x\right) \\
        &\leq 2\exp\left(-\frac{\frac{1}{2}x^{2}}{2^{-2(q-1)}+\frac{2}{3}a_{q-1}x}\right).
    \end{align*}
    Now we can again apply Lemma 2.11.17 to finish handling the third expression.

    This concludes our proof of (\ref{eqn::asymp_tight}). If the partitions depend on $N$, we require an additional step. Let $\delta_{N}$ be a sequence tending to zero as $N \rightarrow \infty$. First we have
    \begin{align*}
        \E \sup_{d(f,g) < \delta_{N}}\left|\sum_{i=N}^{\infty} (Z_{Ni}^{\circ}(f) - Z_{Ni}^{\circ}(g))\right| &\leq 2\E^{*} \left\Vert \sum_{i=N}^{\infty} (Z_{Ni}^{\circ}(f) - Z_{Ni}^{\circ}(\pi_{q_0}f))\right\Vert_{\mathcal{F}}\\
        &+\E\left[ \sup_{d(f,g) < \delta_{N}} \sum_{i=N}^{\infty} \left|Z_{Ni}^{\circ}(\pi_{q_0}f) - Z_{Ni}^{\circ}(\pi_{q_0}g)\right|\right].
    \end{align*}
    We have already dealt with the first term on the right-hand side. For the second term, consider the set
    \begin{equation*}
        \mathcal{H}^{N}_{q_0} = \{(\Tilde{f},\Tilde{g}):\, \text{there exists }f,g\in\mathcal{F}\text{ with }d(f,g) <\delta_{N}, \pi_{q_0}f = \Tilde{f}, \pi_{q_{0}}g = \Tilde{g} \}.
    \end{equation*}
    The size of $\mathcal{H}^{N}_{q_0}$ is at most $(\mathfrak{N}_{q_0}^{N})^{2}$. Define $\zeta_{N} \geq 0$ as
    \begin{equation*}
        \zeta^{2}_{n} = \sup_{d(f,g)< \delta_{N}}\sum_{i=N}^{\infty}\E[(Z_{Ni}^{\circ}(f)-Z_{Ni}^{\circ}(g))^{2}] \lesssim \sup_{d(f,g)< \delta_{N}}\sum_{i=N}^{\infty}\E[(Z_{Ni}(f)-Z_{Ni}(g))^{2}],
    \end{equation*}
    which tends to zero as $\delta_{N} \downarrow 0$ by assumption. For $(\Tilde{f},\Tilde{g}) \in \mathcal{H}_{q_0}^{N}$, we have the following bounds:
    \begin{align*}
        \left|Z_{Ni}^{\circ}(\Tilde{f}) - Z_{Ni}^{\circ}(\Tilde{g})\right| &\leq 4\eta{N}\\
        \text{Var}\left[\sum_{i=N}^{\infty} |Z_{Ni}^{\circ}(\Tilde{f})-Z_{Ni}^{\circ}(\Tilde{g})|\right] &\leq \sum_{i=N}^{\infty}\E[(Z_{Ni}^{\circ}(\Tilde{f}) - Z_{Ni}^{\circ}(\Tilde{g}))^{2}] \\ 
        &\leq 3\sum_{i=1}^{\infty}\E[(Z_{Ni}^{\circ}(f) - Z_{Ni}^{\circ}(\Tilde{f}))^{2}] +\E[(Z_{Ni}^{\circ}(g) - Z_{Ni}^{\circ}(\Tilde{g}))^{2}]\\&+\E[(Z_{Ni}^{\circ}(f) - Z_{Ni}^{\circ}(g))^{2}]\\
        &\lesssim 2^{-2q_{0}} + \zeta_{N}^{2},
    \end{align*}
    where $(f,g)$ above are any elements of $\mathcal{F}$ satisfying $d(f,g) < \delta_{N}$ with $\pi_{q_0}f = \Tilde{f}$ and $\pi_{q_0}g = \Tilde{g}$. Thus, by using a similar Bernstein inequality argument to before and applying Lemma 2.2.13 of \cite{vanderVaart2023}, we yield
    \begin{align*}
       \E\left[ \sup_{d(f,g) < \delta_{N}} \sum_{i=N}^{\infty} \left|Z_{Ni}^{\circ}(\pi_{q_0}f) - Z_{Ni}^{\circ}(\pi_{q_0}g)\right|\right] &=\E\left[ \max_{(\Tilde{f},\Tilde{g}) \in \mathcal{H}_{q_0}^{N}} \sum_{i=N}^{\infty} \left|Z_{Ni}^{\circ}(\Tilde{f}) - Z_{Ni}^{\circ}(\Tilde{g})\right|\right]\\
        &\lesssim \log N_{q_0}^{N}\eta{N} + \sqrt{\log N_{q_0}^{N}}(2^{-q_{0}} + \zeta_{N}).
    \end{align*}
    We showed earlier that for fixed $q_{0}$, the limit superior of $\log \mathfrak{N}_{q_0}^{N}$ as $N \rightarrow \infty$ is finite, so the $\eta_{N}$ and $\zeta_{N}$ terms above go to zero as $N \rightarrow 0$. This leaves the $\sqrt{\log \mathfrak{N}_{q_0}^{N}}2^{-q_{0}}$ term, which goes to zero as $N \rightarrow \infty$ and then $q_{0} \rightarrow \infty$ by condition (\ref{eqn::entropy_sum}). Finally, we can apply Theorem 1.5.7 of \cite{vanderVaart2023} to obtain asymptotic tightness.
\end{proof}

\subsection{Almost supermartingales}
As our recursive estimates are closely related to stochastic approximation, it is not surprising that we will borrow some tools from that literature to study the frequentist asymptotic properties of the QMP. In particular, a very useful theorem is given by \cite{Robbins1971}, which has been used for proving consistency for other closely connected recursive Bayesian methods like in \citet{Martin2009,Hahn2018} and \citet{Fong2023a}. We restate the almost supermartingale convergence theorem below.
\begin{thm}[{\cite[]{Robbins1971}}]\label{app:thm_asmart}
Consider a probability space $(\Omega,\mathcal{F},\mathbb{P})$, and let $\{\mathcal{F}_i\}_{i \geq 0}$ denote a filtration. Let $\{L_i\}_{i \geq 0}$ be a sequence of non-negative r.v.s $L_i:\Omega \to \R$ adapted to the filtration (i.e. $L_i$ is $\mathcal{F}_i$-measurable for $i \geq 0$). Suppose that $\{L_i\}_{i \geq 0}$ is an {almost supermartingale}, that is it satisfies for $i \geq 0$
\begin{align*}
E[L_{i+1} \mid \mathcal{F}_{i}]  \leq (1+B_i)\, L_{i} + C_i - D_i
\end{align*}
where $(B_i, C_i, D_i)$ are non-negative adapted r.v.s. 
If $\{\sum_{i=1}^\infty B_i < \infty,\quad  \sum_{i=1}^\infty C_i< \infty\}$ hold a.s., then the limit $L_\infty:=\lim_{i\to \infty} L_i $ exists and is finite a.s. and $\sum_{n=1}^\infty D_i < \infty $ a.s.
\end{thm}

\subsection{Bivariate normal copula}\label{app:sec_copula}
Although most of the theory can be extended for general copulas, we specialize most of the proofs for the case with the bivariate normal copula. As such, we provide some useful properties here, and refer to \cite{Meyer2013} for more details.

The bivariate normal copula distribution is the bivariate cumulative distribution function (CDF) $C_\rho(u,v)$ which takes the form
\begin{align*}
C_\rho(u,v) = \Phi_2\left(z_u, z_v; \rho\right)
\end{align*}
where $u,v \in (0,1)$, $z_u = \Phi^{-1}(u)$ is the normal quantile function at $u$, and similarly for $v$. The correlation parameter is $\rho \in [-1,1]$ in general, but we will only consider $\rho \in [0,1]$ for our purposes.  
Here, $\Phi_2(\mu_1, \mu_2; \rho)$ is the standard bivariate normal CDF evaluated at $\mu_1,\mu_2$ with correlation $\rho$.  

The conditional distribution of the bivariate normal copula (conditional on $v$) takes the form
\begin{align*}
    H_\rho(u,v) &= \frac{\partial}{\partial v} C_\rho(u,v) =\Phi\left\{\frac{\Phi^{-1}\left(u\right) - \rho \Phi^{-1}(v)}{\sqrt{1-\rho^2}}\right\},
\end{align*}
and the density of the bivariate normal copula is
\begin{align*}
    c_\rho(u,v) &= \frac{\partial^2}{\partial u \partial v}C_\rho(u,v) =  \frac{1}{\sqrt{1-\rho^2}}\exp\left( \frac{2\rho z_u z_v - \rho^2(z_u^2 + z_v^2)}{2(1-\rho^2)}\right).
\end{align*}
Note that $C_\rho(u,v)$ and $c_\rho(u,v)$ are symmetric in its inputs, but $H_\rho(u,v)$ is not.

As the marginal distribution of $C_\rho(u,v)$ is uniform, we have that
\begin{align*}
    \int_0^1 c_\rho(u,v) \, du =  \int_0^1 c_\rho(u,v) \, dv = 1.
\end{align*}
This in turn implies
\begin{align*}
    \int_0^1 H_\rho(u,v) \, dv =\int_0^1 \int_0^u c_\rho(u',v) \, du'\, dv = u,
\end{align*}
which is crucial for the martingale property.

From \cite{Meyer2013}, the bivariate normal copula $C_\rho(u,v)$ satisfies the following ordering property. For any $\rho' \leq \rho$, we have
\begin{align*}
    \max(u+v-1,0) \leq C_{\rho'}(u,v) \leq C_{\rho}(u,v) \leq \min(u,v),
\end{align*}
for all $u,v \in [0,1]$. Furthermore, the lower and upper bounds are attained with $\rho   \to 0 $ and $\rho \to 1$ respectively, i.e. $\lim_{\rho \to 1}C_\rho(u,v) = \min(u,v)$ and $\lim_{\rho \to 0}C_\rho(u,v) = \max(u+v-1,0)$.

\subsection{Useful identities}
A useful integral we will need for the proofs is the following for the bivariate copula density.
\begin{lem}\label{app:lem_cop_squared}
    Let $c_\rho(u,v) = \frac{\partial}{\partial u} H_\rho(u,v)$ denote the bivariate normal copula density. We have that
    \begin{align*}
        \int_0^1\int_0^1 c_\rho(u,v)^2 \, du \, dv = \frac{1}{1-\rho^2}\cdot
    \end{align*}
\end{lem}
\begin{proof}
     A change of variables from $v \to z_v$ gives
    \begin{align*}
        \int_0^1 c_\rho(u,v)^2\, dv &= \frac{1}{{1-\rho^2}}\int_{-\infty}^\infty \exp\left(\frac{2\rho z_u z_v - \rho^2(z_u^2 + z_v^2)}{1-\rho^2}\right) \, \phi(z_v)\, dz_v\\
        &=\frac{1}{\sqrt{2\pi}{(1-\rho^2)}}\int_{-\infty}^\infty \exp\left(\frac{4\rho z_u z_v-(1+\rho^2)z_v^2 - 2\rho^2z_u^2 }{2(1-\rho^2)}\right) \, dz_v
    \end{align*}
    where $\phi$ is the normal density. Completing the square then gives us
    \begin{align*}
\int_0^1 c_\rho(u,v)^2\, dv &=\frac{1}{\sqrt{2\pi}(1-\rho^2)}\exp\left(\frac{\rho^2}{1+\rho^2}z_u^2\right)\int_{-\infty}^\infty \exp\left(-\frac{(1+\rho^2)}{2(1-\rho^2)}\left[\left(z_v - \frac{2\rho z_u}{1+\rho^2}\right)^2\right]\right) \, dz_v \\
&=\frac{1}{\sqrt{1-\rho^4}}\exp\left(\frac{\rho^2}{1+\rho^2}z_u^2\right).
    \end{align*}
    Carrying out another change of variables from $u \to z_u$ then gives
    \begin{align*}
        \int_0^1\int_0^1 c_\rho(u,v)^2\, du\, dv &= \frac{1}{\sqrt{2\pi(1-\rho^4)}}\int_{-\infty}^{-\infty}\exp\left(-\frac{(1-\rho^2)}{2(1+\rho^2)}z_u^2\right)\, dz_u\\
        &=\sqrt{\frac{1+\rho^2}{1-\rho^2}}\frac{1}{\sqrt{1-\rho^4}} = \frac{1}{1-\rho^2}\cdot
    \end{align*}
\end{proof}
We also have the following useful upper bound on $c_\rho(u,v)$.
\begin{lem}\label{app:lem_cop_bounded}
    The bivariate copula density satisfies
    \begin{align*}
        c_\rho(u,v) \leq \frac{1}{\sqrt{1-\rho^2}}\exp\left(\frac{z_v^2}{2}\right)
    \end{align*}
   The inequality holds if we replace $z_v$ with $z_u$.
    \begin{proof}
      For a given $v$, standard calculations give that  $z_u^* = z_v/\rho$ maximizes $c_\rho(u,v)$, which returns the above expression.
    \end{proof}
\end{lem}

Another very useful lemma which we will use for the covariance function of the QMP is the following.
\begin{lem}\label{app:lem_covariance}
For $u,u' \in (0,1)$, the copula update function satisfies
\begin{align*}
    \int_0^1\, \left[u - H_\rho(u,v)\right]\left[u' - H_\rho(u',v)\right]\, dv = C_{\rho^2}(u,u') - uu'.
\end{align*}
\end{lem}
\begin{proof}
    First, we can easily see that
    \begin{align*}
    \int_0^1\, \left[u - H_\rho(u,v)\right]\left[u' - H_\rho(u',v)\right]\, dv =  \int_0^1\, H_\rho(u,v)\,H_\rho(u',v)\, dv - uu'.
\end{align*}
To compute the integral, we write $H_\rho$ in terms of $c_\rho$:
\begin{align*}
    \int_0^1\, H_\rho(u,v)\,H_\rho(u',v)\, dv  =\int_{-\infty}^u \int_{-\infty}^{u'} \int_0^1 c_\rho(w,v)\, c_\rho(w',v)\, dv\, dw \, dw'
\end{align*}
The inner integral can be computed as
\begin{align*}
&\int_0^1 c_\rho(w,v)\, c_\rho(w',v)\, dv \\&= \frac{1}{1-\rho^2}\int_{-\infty}^\infty \exp\left(\frac{-\rho^2(z_w^2 + z_{w'}^2) + 2\rho (z_w + z_{w'})\, z_v - 2\rho^2 z_v^2}{2(1-\rho^2)}\right)\, \phi(z_v)\, dz_v\\
 &= \frac{1}{\sqrt{2\pi}(1-\rho^2)}\exp\left(\frac{-\rho^2(z_w^2 + z_{w'}^2)}{2(1-\rho^2)}\right)\int_{-\infty}^\infty \exp\left(\frac{2\rho (z_w + z_{w'})\, z_v - (1+\rho^2) z_v^2}{2(1-\rho^2)}\right)\, dz_v
\end{align*}
where $\phi$ is the standard normal density function. Completing the square gives
\begin{align*}
\int_{-\infty}^\infty \exp\left(\frac{2\rho (z_w + z_{w'})\, z_v - (1+\rho^2) z_v^2}{2(1-\rho^2)}\right)\, dz_v = \sqrt{2\pi}\sqrt{\frac{1-\rho^2}{1+\rho^2}}\exp\left(\frac{\rho^2\left(z_w + z_{w'}\right)^2}{2(1+\rho^2)(1-\rho^2)}\right).
\end{align*}
Combining the above gives
\begin{align*}
    \int_0^1 c_\rho(w,v)\, c_\rho(w',v)\, dv &= \frac{1}{\sqrt{1-\rho^4}}\exp\left(\frac{-\rho^4\left(z_{w}^2 + z_{w'}^2\right)+ 2\rho^2 z_{w}z_{w'}}{2(1-\rho^4)}\right) = c_{\rho^2}(w,w').
\end{align*}
Finally, this gives
\begin{align*}
     \int_0^1\, H_\rho(u,v)\,H_\rho(u',v)\, dv  =C_{\rho^2}(u,u').
\end{align*}
\end{proof}

\section{Proofs of main results}
We now include full proofs of the main results from the paper, leveraging the prerequisite results.

\subsection{Proposition \ref{prop:L2_mart}}
For predictive asymptotics, we will treat the first $n$ data points $Y_{1:n}$ as fixed. Note that we start indexing at $N = n$ is so that the extension to the case where $Y_{1:n}$ is i.i.d. from $P^*$ is straightforward.
We will apply {Theorem \ref{app:thm_hilbert_mart}} for the space $L^2((0,1))$ under {Assumptions \ref{as:L2}} and {\ref{as:alpha}}. First, we verify that $\{Q_N\}_{N\geq n}$ is a Banach space valued martingale under {Algorithm \ref{alg:QMP}}. 

 We begin with the simplified case where $n = 0$. Let $(\Omega, \mathcal{F},\Pr)$ denote the probability space and $L^2(B)$ the space of Bochner measurable functions with $\|f\|_{L^2(B)} < \infty$ as in Section \ref{app:sec_ban_mart}, with $B = L^2((0,1))$. Let $V_i \iid \mathcal{U}(0,1)$ for $i \geq 1$, i.e. each $V_i:\Omega \to (0,1)$ is an independent uniform r.v. Define the filtration $\{\mathcal{F}\}_{N \geq 1}$ where $\mathcal{F}_N = \sigma(V_1,\ldots,V_N)$ and $\mathcal{F}_0 = \{\emptyset, \Omega\}$. For $N \geq 1$, define the mapping $S_N: (0,1)^N  \to B$ where
\begin{align}\label{app:eq_Snu}
    S_N(v_{1:N})(u) = \sum_{i = 1}^N \alpha_i (u - H_{\rho_i}(u,v_i))
\end{align}
for each $u \in (0,1)$. For each $v_{1:N} \in (0,1)^N$, we clearly have $S_N(v_{1:N}) \in B$ as it is bounded by $|S_N(v_{1:N})(u)| \leq \sum_{i = 1}^N \alpha_i$.
 Here $\rho_i \in (0,1)$ and $\alpha_i\in \R^+$ are arbitrary sequences where $\alpha_i <\infty$. To begin, we require the following lemma, which is a technical exercise, but we can fortunately repeat a similar argument for later proofs. The key is that as we are working in a separable Hilbert space, we can revert back to checking scalar conditions using inner products (which is termed `scalarization' for general Banach spaces in \cite{Pisier2016}). This also formally verifies our intuition that a pointwise martingale condition is sufficient.
 
\begin{lem}\label{app:lem_bochnerl2}
For each $N$, the random variable  $S_{N}\circ V_{1:N}:\Omega \to B$ is Bochner $\mathcal{F}_{N}$-measurable, lies in $L^2(B)$ and satisfies
$\E^{\mathcal{F}_{N}}[S_{N+1}(V_{1:N+1})] = S_{N}(V_{1:N})$ a.s.
\end{lem}
\begin{proof}
First, we highlight that $B$ is a separable Hilbert space, so we can appeal to Pettis' measurability theorem (e.g. \citet[Theorem 1.1.6]{Hytonen2016}) to show Bochner measurability by verifying weak measurability. For each continuous linear functional $T \in B^*$, let $h_T$ be the Riesz representation of $T$. We thus just need to verify Borel measurability of the scalar function $g_T: (0,1)^N \to \R$
 where
 \begin{align*}
    g(v_{1:N})  &:= \langle S_N(v_{1:N}), h_T \rangle_{L^2} = \int_0^1 S_N(v_{1:N})(u) \, h_T(u) \, du.
\end{align*}
For each $u \in (0,1)$, the function $S_N(v_{1:N})(u)\, h_T(u)$ is continuous in $v_{1:N}$, and $|S_N(v_{1:N})(u)\, h_T(u)|$ is bounded by $\|h_T\|_{\infty}\sum_{i = 1}^N \alpha_i$. Consider an arbitrary vector sequence $v^j_{1:N} \to v^*_{1:N}$.
Dominated convergence gives $g(v^j_{1:N}) \to g(v^*_{1:N})$, which
 implies $g(v_{1:N})$ is a continuous function on $(0,1)^N$,  and thus by composition $g \circ V_{1:N}$ is $\mathcal{F}_{N}$-measurable. We thus have that $S_N \circ V_{1:N}$ is Bochner $\mathcal{F}_N$-measurable. To show it lies in $L^2(B)$, we have $\|S_N(v_{1:n})\|_2\leq \sum_{i = 1}^N \alpha_i$ for all  $v_{1:n}$, so $ \|S_N(V_{1:N})\|_{L^2(B)}\leq \sum_{i = 1}^N \alpha_i$. 

For the final part, we leverage the discussion in Section \ref{app:sec_ban_mart}. Consider again a continuous linear functional $T \in B^*$. The conditional expectation operator  satisfies the following (e.g. \citet[Remark 1.11]{Pisier2016}):
\begin{align*}
    T\E^{\mathcal{F}_N}[S_{N+1}(V_{1:N+1})] = \E[TS_{N+1}(V_{1:N+1})\mid \mathcal{F}_N] \quad \text{a.s.}
\end{align*}
where we revert to the standard notation $E[\cdot \mid \mathcal{F}]$ when working with scalar conditional expectations for clarity. Using the Riesz representation again, we have
\begin{align*}
    T\E^{\mathcal{F}_N}[S_{N+1}(V_{1:N+1})] &= \E\left[\int_0^1 \, h_T(u) \, S_{N+1}(V_{1:N+1})(u) \, du \mid \mathcal{F}_N\right]\\
    &= \E\left[\int_0^1 \, h_T(u) \, S_{N}(V_{1:N})(u) \, du \mid \mathcal{F}_N\right] \\
    &+ \alpha_{N+1}\E\left[\int_0^1 \, h_T(u) \,(u - H_{\rho_{N+1}}(u,V_{N+1})) \, du \mid \mathcal{F}_N\right].
\end{align*}
We then have
\begin{align*}
    \E\left[\int_0^1 \, h_T(u) \, S_{N}(V_{1:N})(u) \, du \mid \mathcal{F}_N\right] = \int_0^1 \, h_T(u) \, S_{N}(V_{1:N})(u) \, du = TS_N(V_{1:N}) \quad \text{a.s.}
\end{align*}
since $\int_0^1 \, h_T(u) \, S_{N}(V_{1:N})(u) \, du$ is $\mathcal{F}_N$-measurable. For the second term, we have
\begin{align*}
    \E\left[\int_0^1 \, h_T(u) \,(u - H_{\rho_{N+1}}(u,V_{N+1})) \, du \mid \mathcal{F}_N\right] &=  \E\left[\int_0^1 \, h_T(u) \,(u - H_{\rho_{N+1}}(u,V_{N+1})) \, du \right] \quad \text{a.s.}\\
    &= \int_0^1 \int_0^1 \, h_T(u) \,(u - H_{\rho_{N+1}}(u,v)) \, du\, dv   \quad \text{a.s.}\\
    &= 0 \quad \text{a.s.}
\end{align*}
where we have used the independence of $V_{N+1}$ from $\mathcal{F}_N$ in the first line, and the last line follows from Fubini's theorem and the pointwise martingale property of the bivariate copula update (see Section \ref{app:sec_copula}). As a result, we have
\begin{align*}
    T\E^{\mathcal{F}_N}[S_{N+1}(V_{1:N+1})] = TS_{N}(V_{1:N}) \quad \text{a.s.}
\end{align*}
for each $T \in B^*$. From the Bochner measurability of $S_N(V_{1:N})$, we have from \citet[Corollary 1.1.25]{Hytonen2016} that the above is sufficient for $\E^{\mathcal{F}_N}[S_{N+1}(V_{1:N+1})] = S_N(V_{1:N})$ a.s. This follows as testing a.s. equality for all continuous linear functionals is sufficient under strong measurability.
\end{proof}

Although the above is more of a technical exercise, it verifies our intuition that having a martingale pointwise for a function is sufficient for it to be a $B$-valued martingale. For the QMP, we can then construct the r.v. $Q_N: \Omega \to B$, where 
\begin{align}\label{app:eq_relab}
    Q_N(u) = Q_n(u) + \sum_{i = n+1}^N \alpha_i (u - H_{\rho_i}(u,V_i))
\end{align}
for $u \in (0,1)$, and $\alpha_i$ is now specified as in {Assumption \ref{as:alpha}}. With a relabelling of indices so $n > 0$, the term on the right is equivalent to $S_N(V_{1:N})(u)$ as defined above. Under {Assumption \ref{as:L2}}, it is clear from {Lemma \ref{app:lem_bochnerl2}} that $Q_N$ is Bochner $\mathcal{F}_N$-measurable, $Q_N \in L^2(B)$, and
\begin{align*}
    \E^{\mathcal{F}_N}[Q_{N+1}] = Q_N \quad \text{a.s.}
\end{align*}
 $\{Q_N\}_{N \geq n}$  is thus a $B$-valued martingale. 
 
 We will now show that $\sup_{N \geq n} \|Q_N\|_{L^2(B)} < \infty$ in order to apply {Theorem \ref{app:thm_hilbert_mart}}. We begin with
\begin{align*}
    \|Q_N\|^2_{L^2(B)}  &= \E\left[\int_0^1 Q_N(u)^2\, du\right]=\int_0^1\E\left[Q_N(u)^2\right]\, du
\end{align*}
which follows from Tonelli's theorem. As $\{Q_N(u)\}_{N \geq n}$ is a martingale for each $u\in (0,1)$, we have
\begin{align*}
   \E[Q_N(u)^2 \mid \mathcal{F}_{N-1}] &= Q_{N-1}^2(u) + \alpha_N^2\E\left[ \left(u - H_{\rho_N}\left(u,V_N\right)\right)^2\right]\leq Q_{N-1}^2(u) + \alpha_N^2
\end{align*}
which follows as $H_\rho(u,v) \leq 1$. Iterated expectation gives us 
\begin{align*}
   \E[Q_N(u)^2 ] \leq Q_n^2(u) + \sum_{i = n+1}^N \alpha_i^2.
\end{align*}
By {Assumption \ref{as:alpha}}, we have that $\sup_{N\geq n} \sum_{i = n+1}^N \alpha_i^2 \leq C < \infty$. As a result, we have
\begin{align*}
    \sup_{N\geq n}  \|Q_N\|^2_{L^2(B)}  \leq \int_0^1 \sup_{N\geq n} \E\left[Q_N(u)^2\right]\,du \leq \int_0^1 Q_n^2(u) \,du + C.
\end{align*}
By {Assumption \ref{as:L2}}, we then have $\sup_{N \geq n}  \|Q_N\|_{L^2(B)} < \infty$. We can thus apply {Theorem \ref{app:thm_hilbert_mart}}. 

\subsection{{Theorem \ref{th:sobolev_mart}} and {Corollary \ref{cor:continuous}}}
The proof of {Theorem \ref{th:sobolev_mart}} is similar to the proof of {Proposition \ref{prop:L2_mart}}, but we will be working in the Sobolev space $B' = H^1((0,1))$. Following the discussion
in Section \ref{app:sec_ban_mart}, the martingale condition can be checked pointwise again, as long as $Q_N$ takes values in $H^1((0,1))$. We extend {Lemma \ref{app:lem_bochnerl2}} below.
Consider the same setup as in the proof of {Proposition \ref{prop:L2_mart}}, but replace all mentions of $B$ with $B'$. We first verify that $S_N(v_{1:N})$ as defined in (\ref{app:eq_Snu}) is in $H^1((0,1))$ for all $v_{1:N} \in (0,1)^N$. First, we note that
the update function has partial derivative
\begin{align}\label{app:eq_partial}
    \frac{\partial}{\partial u} (u - H_\rho(u,v)) = 1 - c_\rho(u,v).
\end{align}
We thus have the following for each $v_{1:N} \in (0,1)^N$:
\begin{align*}
  s_N(v_{1:N})(u) := \frac{\partial}{\partial u} S_N(v_{1:N})(u) = \sum_{i = 1}^n \alpha_i (1 - c_{\rho_i}(u,v_i)). 
\end{align*}
{Lemma \ref{app:lem_cop_bounded}} then gives that $s_N(v_{1:n})$ is bounded for each $v_{1:N} \in (0,1)^N$ and each $N$, so $\|s_N(v_{1:N})\|_2 < \infty$ and $S_N: (0,1)^N \to H^1((0,1))$ for each $N$. We then have the following lemma.
\begin{lem}\label{app:lem_bochnerh1}
For each $N$, the random variable $S_N \circ V_{1:N}: \Omega \to B'$ is Bochner $\mathcal{F}_N$-measurable, lies in $L^2(B')$ and satisfies
$\E^{\mathcal{F}_N}[S_{N+1}(V_{1:N+1})] = S_N(V_{1:N})$ a.s.  
\end{lem}
\begin{proof}
    As $H^1$ is also a separable Hilbert space, the proof follows very much in the same way as {Lemma \ref{app:lem_bochnerl2}} with $B$ replaced with $B'$, with the only difference arising from the Sobolev inner product. To check Bochner measurability, we show Borel measurability of
    \begin{align*}
        g'(v_{1:N}):= \langle S_N(v_{1:N}), h_T\rangle_{H^1} =  g(v_{1:N})+  \int_0^1 s_N(v_{1:N})(u) \, h'_T(u) \, du
    \end{align*}
    where $h'_T\in B$ is the weak derivative of $h_T$.
    The first term is already shown to be continuous. For the second term, define $t(v) = \int_0^1 c_\rho(u,v)\, h_T'(u)\, du$, and consider a sequence $v_i \to v^*$ for $v^* \in (0,1)$. We then have
    \begin{align*}
        |t(v) - t(v^*)| &\leq \int_0^1 |c_\rho(u,v_i)-c_\rho(u,v^*)|\, |h_T'(u)|\, du\\
        &\leq \|c_\rho(\cdot,v_i)-c_\rho(\cdot,v^*)\|_2 \, \|h_T'\|_2\\
        &\leq K\|c_\rho(\cdot,v_i)-c_\rho(\cdot,v^*)\|_2
    \end{align*}
    As $|v_i - v^*| \leq  \varepsilon$ for sufficiently large $i$, $v_i$ is eventually bounded away from 0 and 1, so $c_\rho(u,v_i)$ is eventually bounded uniformly over $u$ by {Lemma \ref{app:lem_cop_bounded}}. Dominated convergence implies $t(v)$ is continuous on $(0,1)$, so $g'(v_{1:N})$ is continuous on $(0,1)^N$ and thus  $g'\circ V_{1:N}$ is $\mathcal{F}_N$-measurable.
    
 Showing that $S_N(V_{1:N}) \in L^2(B')$ requires some more work. As we already showed that $S_N(V_{1:N})$ is in $L^2(B)$, we just need to verify that $\E\left[\|s_N(V_{1:N})\|^2_2\right] < \infty$.
The linearity of expectation gives
\begin{align*}
    \E\left[\|s_N(V_{1:N})\|^2_2\right] = \int_0^1 \sum_{i = 1}^N \sum_{j = 1}^N \alpha_i \alpha_j \E\left[(1-c_{\rho_i}(u,V_i))(1-c_{\rho_j}(u,V_j))\right]\, du.
\end{align*}
As $V_i$ is independent of $V_j$ for $i \neq j$, and $\E[(1-c_{\rho}(u,V_i))] = 0$ from Section \ref{app:sec_copula}, the cross-terms disappear and we have
     \begin{align*}
    \E\left[\|s_N(V_{1:N})\|^2_2\right] &=  \sum_{i = 1}^N \alpha_i^2 \int_0^1\E\left[(1-c_{\rho_i}(u,V_i))^2\right]\, du\\
    &=  \sum_{i = 1}^N \alpha_i^2 \left[\int_0^1 \int_0^1 c_{\rho_i}(u,v)^2\, du \, dv - 1\right]\\
    &=  \sum_{i = 1}^N \alpha_i^2 \frac{\rho_i^2}{1-\rho_i^2}
\end{align*}
where the last line follows from {Lemma \ref{app:lem_cop_squared}}. As $\rho_i \neq 1$, the above is bounded for each $N$, so $S_N(V_{1:N}) \in L^2(B')$.

For the final part, we again just need to verify a.s. equality for continuous linear functionals $T \in {B'}^*$ with Riesz representation $h_T \in B'$. This time, we have
\begin{align*}
    T\E^{\mathcal{F}_N}[S_{N+1}(V_{1:N+1})] &= \E\left[\int_0^1 \, h_T(u) \, S_{N+1}(V_{1:N+1})(u) \, du \mid \mathcal{F}_N\right]\\
    &+ \E\left[\int_0^1 \, h'_T(u) \, s_{N+1}(V_{1:N+1})(u) \, du \mid \mathcal{F}_N\right].
\end{align*}
 From the proof of {Lemma \ref{app:lem_bochnerl2}}, we have that the first term is equal to $\langle h_T, S_N(V_{1:N}) \rangle_{L^2}$ a.s. For the second term, we can carry out a similar argument which gives
\begin{align*}
    \E\left[\int_0^1 \, h'_T(u) \, s_{N+1}(V_{1:N+1})(u) \, du \mid \mathcal{F}_N\right] &=  \E\left[\int_0^1 \, h'_T(u) \, s_{N}(V_{1:N})(u) \, du \mid \mathcal{F}_N\right]\\
    &+ \alpha_{N+1}  \E\left[\int_0^1 \, h'_T(u) \, (1 - c_{\rho_{N+1}}(u,V_{N+1})) \, du \mid \mathcal{F}_N\right]
\end{align*}
Again, we have the first term as equal to $\langle h'_T(u), \, s_{N}(V_{1:N})\rangle_{L^2}$ a.s. from $\mathcal{F}_N$-measurability, and
\begin{align*}
     \E\left[\int_0^1 \, h'_T(u) \, (1 - c_{\rho_{N+1}}(u,V_{N+1})) \, du \mid \mathcal{F}_N\right] &=\int_0^1 \int_0^1 \, h'_T(u) \, (1 - c_{\rho_{N+1}}(u,v)) \, du\, dv \quad \text{a.s.}\\
     &= 0 \quad \text{a.s.}
\end{align*}
where we can apply Fubini's theorem as Cauchy-Schwarz gives
\begin{align*}
    \int_0^1 \int_0^1 \,\left| h'_T(u) \, (1 - c_{\rho_{N+1}}(u,v))\right| \, du\, dv  \leq \|h_T'\|_2 \, \sqrt{\int_0^1 \int_0^1\left(1-c_{\rho_{N+1}}(u,v)\right)^2\, du\,dv} < \infty.
\end{align*}
As a result, we have
\begin{align*}
     T\E^{\mathcal{F}_N}[S_{N+1}(V_{1:N+1})] &= \langle h_T(u), \, S_{N}(V_{1:N})\rangle_{L^2}+ \langle h'_T(u), \, s_{N}(V_{1:N})\rangle_{L^2} \quad \text{a.s.} \\
     &=\langle h_T(u), \, S_{N}(V_{1:N})\rangle_{H^1} \quad \text{a.s.}\\
     &=TS_N(V_{1:N}) \quad \text{a.s.}
\end{align*}
for each $T \in {B'}^*$. Again from Bochner measurability and \citet[Corollary 1.1.25]{Hytonen2016}, we have $\E^{\mathcal{F}_N}[S_{N+1}(V_{1:N+1})] = S_N(V_{1:N})$ a.s.
\end{proof}
Once again, pointwise martingales are sufficient. We then define $Q_N:\Omega \to B'$ again as (\ref{app:eq_relab}), with $\alpha_i$ and $\rho_i$ from {Assumptions \ref{as:alpha}} and  {\ref{as:bandwidth}} respectively. As $Q_n \in B'$ by {Assumptions \ref{as:L2}} and {\ref{as:weak_deriv}}, $\{Q_{N}\}_{N \geq n+1}$ is a $B'$-valued martingale by {Lemma \ref{app:lem_bochnerh1}}.

We now verify that $\sup_{N \geq n} \|Q_N\|_{L^2(B')}< \infty$, where $Q_N$ is weakly differentiable a.s. with a.e. unique weak derivative
\begin{align*}
    q_N(u) = q_n(u) + \sum_{i = n+1}^N \alpha_i \left(1- c_{\rho_i}(u,V_i)\right)
\end{align*}
for $u \in (0,1)$. We begin with
\begin{align*}
    \|Q_N \|^2_{L^2(B')}&= \E\left[\int_0^1 Q_N(u)^2\,du + \int_0^1 q_N(u)^2\,du\right]=\int_0^1\left(\E\left[Q_N(u)^2\right] +\E\left[q_N(u)^2\right] \right)\,du,
\end{align*}
where we have used Tonelli's theorem in the second equality. We have already  bounded $\E[Q_N(u)^2]$ in the proof of {Theorem \ref{prop:L2_mart}}, so we focus on the second term. We first note that $\{q_N(u)\}_{N \geq n}$ is a martingale for each $u \in (0,1)$, as $\int_0^1 c_\rho(u,v) \, dv = 1$. We thus have
\begin{align*}
    \E\left[q_N(u)^2 \mid \mathcal{F}_{N-1}\right] &= q_{N-1}(u)^2 + \alpha_N^2 \E\left[(1-c_{\rho_N}(u,V_N))^2\mid \mathcal{F}_{N-1}\right]\\
    &=q_{N-1}(u)^2 + \alpha_N^2 \left(\E\left[c_{\rho_N}(u,V_N)^2\mid \mathcal{F}_{N-1}\right]-1\right)\\
    &=q_{N-1}(u)^2 + \alpha_N^2 \frac{\rho_N^2}{1-\rho_N^2}
\end{align*}
where we have applied {Lemma \ref{app:lem_cop_squared}} in the last line. Iterated expectation again gives us
\begin{align*}
   \E\left[q_N(u)^2\right]  =q_n(u)^2 + \sum_{i = n+1}^N \alpha_i^2 \frac{\rho_i^2}{1-\rho_i^2}\cdot
\end{align*}
Combining with $\E[Q_N(u)^2]$ calculated earlier, this then gives us
\begin{align*}
    \|Q_N\|^2_{L^2(B')} &\leq \|Q_n\|^2_{1,2} + \sum_{i = n+1}^N \alpha_i^2 \left(1 + \frac{\rho_i^2}{1-\rho_i^2}\right)\\
    &=\|Q_n\|^2_{1,2} + \sum_{i = n+1}^N \frac{\alpha_i^2}{1-\rho_i^2}\cdot
\end{align*}
Now consider the form $\alpha_N = a(N+1)^{-1}$ and $\rho_N = \sqrt{1-cN^{-k}}$ from {Assumptions \ref{as:alpha}} and {\ref{as:bandwidth}}, where $a,c,k$ are all positive constants (and $c \in (0,1)$ so $\rho_N \in (0,1)$). The term to bound is
\begin{align*}
     \sup_{N \geq n}\sum_{i = n+1}^N \frac{\alpha_i^2}{1-\rho_i^2} \leq   ac^{-1}\sum_{i = n+1}^\infty i^{-2+k}. 
\end{align*}
The sum on the right is only bounded if $0\leq k < 1$, as satisfied by {Assumption \ref{as:bandwidth}}, and it is clear that $k$ controls the Sobolev norm and hence the smoothness of the QMP sample paths. We can thus apply {Theorem \ref{app:thm_hilbert_mart}} with $B' = H^1((0,1))$. {Corollary \ref{cor:continuous}} follows directly from {Proposition \ref{app:prop_abcont}}, as $Q_\infty$ takes values in $H^1((0,1))$ a.s.

\subsection{Proposition \ref{prop:rearr_sobolev}}
To begin, we note that the $L^2$ norm is weaker than the Sobolev norm, that is $d_{2}(f, g) \leq  d_{1,2}(f, g)$. As a result, we have from {Theorem \ref{th:sobolev_mart}} that $d_2(Q_N, Q_\infty) \to 0$ a.s. As $Q_N$ and $Q_\infty$ take values in $H^1((0,1))$, {Proposition \ref{app:prop_bounded}} implies that $Q_N$ and $Q_\infty$ are a.e. equal to continuous bounded functions on $(0,1)$ a.s., which we can extend to $[0,1]$. We can then apply {Proposition \ref{prop:rearr}}, which gives 
    $d_2(Q_N^{\dagger}, Q_\infty^{\dagger}) \leq d_2(Q_N, Q_\infty)$, and thus
   $d_2(Q_N^{\dagger}, Q_\infty^{\dagger}) \to 0$ a.s.
Furthermore, the increasing rearrangement variant of {Theorem \ref{app:thm_sobolev}} implies that $Q_N^{\dagger}$ and $Q_\infty^{\dagger}$ are both in $L^2(H^1(0,1))$, and thus take values in $H^1((0,1))$ a.s.

\subsection{Corollary \ref{cor:exchange}}
As $P_N \to P_\infty$ weakly a.s., we can directly apply \citet[Lemma 8.2(b)]{Aldous1985} to show the asymptotic exchangeability of the sequence. For the second part, we assume $n = 0$ so the notation is simpler, but this is just a matter of relabelling the indices. We begin by writing
\begin{align*}
    \mathbb{P}_N(y) := \frac{1}{N}\sum_{i=1}^{N}\mathbbm{1}\left(Y_i \leq y\right)= \frac{1}{N}\sum_{i = 1}^{N} \mathbbm{1}\left(Q_{i-1}(V_i) \leq y\right)
\end{align*}
where $V_i \iid \mathcal{U}(0,1)$. We then have that
\begin{align*}
    \E[\mathbb{P}_N(y)\mid \mathcal{F}_{N-1}] = \frac{N-1}{N}\mathbb{P}_{N-1} + \frac{1}{N}P_{N-1}(y)
\end{align*}
which follows as $\int_0^1 \mathbbm{1}\left(Q_N(v) \leq y\right)\, dv$ is precisely $P_N(y)$. 

Following the proof of \cite[Theorem 2.2]{Berti2004}, for any continuous and bounded $f \in C_b(\mathbb{R})$, we define $U_N = f(Y_{N+1})$.
We then define the martingale 
$$Z_N = \sum_{i=0}^{N-1}\frac{U_i - \E[U_i \mid \mathcal{F}_{i-1}]}{i+1},$$
which is bounded in $L^2$, so $(Z_N)_{N \geq 1}$ converges a.s. Kronecker's lemma then gives
\begin{align*}
     \frac{1}{N}\sum_{i=0}^{N-1}\left(U_i - \E[U_i \mid \mathcal{F}_{i-1}]\right) \to 0 \quad \text{a.s.}
\end{align*}
As a result, we have
\begin{align*}
    \lim_{N \to \infty}\frac{1}{N}\sum_{i=0}^{N-1}U_i =   \lim_{N \to \infty}\frac{1}{N}\sum_{i=0}^{N-1}\E[U_i \mid \mathcal{F}_{i-1}] \quad \text{a.s.}
\end{align*}
 Note that the left term is exactly $\mathbb{P}_N[f]$, where $P[f] = \int f dP$. We also have
\begin{align*}
    \E[U_N \mid \mathcal{F}_{N-1}] = P_N[f] \to  P_\infty [f] \quad \text{a.s.}
\end{align*}
as $P_N$ converges weakly to $P_\infty$ a.s., so $\mathbb{P}_N[f] \to P_\infty[f]$ a.s. through a C\'{e}saro means argument. The above holds a.s. for the function $f = e^{ity}$ for each $t \in \mathbb{R}$, so \citet[Theorem 2.6]{Berti2006} gives us $\mathbb{P}_N \to P_\infty$ weakly a.s.

\subsection{Theorem \ref{th:gp}}
The proof of this theorem consists of two parts. First we verify the conditions of {Theorem \ref{th::bracketing_clt}} to obtain asymptotic tightness of $\sqrt{N}S_{N}$. Then it suffices to check marginal convergence on $\mathcal{F}$ using the Lindeberg-Feller CLT; that is, we show that there is weak convergence to the requisite multivariate normal distribution for any finite collection of points. 

\begin{thm}\label{th:tight}
The sequence of functions $\sqrt{N}S_N$ is  asymptotically tight in $\ell^{\infty}((0,1))$.
\end{thm}
\begin{proof}
We will verify the conditions of {Theorem \ref{th::bracketing_clt}}. Our semimetric space is $(\mathcal{F} = (0,1), d)$, where $d(u_{1},u_{2}) = |u_{1}-u_{2}|^{1/2}$. Clearly, this semimetric space is totally bounded.

First define $Z_{Ni}(u) = \sqrt{N}\alpha_{i}H_{\rho_{i}}(u,V_{i})$ for all $N$ and $i \geq N$. A trivial envelope function for $Z_{Ni}$ is $F_{Ni} = \sqrt{N}\alpha_{i}$. We need to verify the Lindeberg condition
\begin{equation*}
    \sum_{i=N}^{\infty} F_{Ni}\mathbbm{1}\{F_{Ni} > \eta\} \rightarrow 0
\end{equation*}
for every $\eta > 0$. Since $F_{Ni} < c/\sqrt{N}$, we will have $F_{Ni} < \eta$ for all sufficiently large $n$. Thus, the Lindeberg condition holds.

Next we need 
\begin{equation*}
    \sup_{d(u_1,u_2) < \delta_{N}} \sum_{i=N}^{\infty} \E\left\{Z_{Ni}(u_{1})-Z_{Ni}(u_{2})\right\}^{2} \rightarrow 0
\end{equation*}
for every $\delta_{N} \downarrow 0$. Let $u_{1} > u_{2}$ with $d(u_1,u_2) < \delta_{N}$. Note that $Z_{Ni}(u)$ is non-decreasing with $\lim_{u \downarrow 0}Z_{Ni}(u) = 0$ and $\lim_{u \uparrow 1}Z_{Ni}(u) = \sqrt{N}\alpha_{i}$. So
\begin{align*}
    \E\left\{Z_{Ni}(u_{1})-Z_{Ni}(u_{2})\right\}^{2} &\leq \sqrt{N}\alpha_{i}\E\left\{Z_{Ni}(u_{1})-Z_{Ni}(u_{2})\right\} \\
    &\leq N\alpha_{i}^{2}(u_{1}-u_{2})\\
    &< N\alpha_{i}^{2}\delta_{N}^{2}.
\end{align*}
So
\begin{equation*}
    \sup_{d(u_1,u_2) < \delta_{N}} \sum_{i=N}^{\infty} \E\left\{Z_{Ni}(u_{1})-Z_{Ni}(u_{2})\right\}^{2} < N\delta_{N}^{2} \sum_{i=N}^{\infty} \alpha_{i}^{2}.
\end{equation*}
Since $\limsup_{N \rightarrow \infty} N \sum_{i=N}^{\infty} \alpha_{i}^{2} \leq c^{2}$, the right-hand side of the above display tends to zero for any $\delta_{N} \downarrow 0$.

Finally, we need to verify the bracketing entropy integral condition. For all sufficiently large $N$, we will have $N\sum_{i=N}^{\infty} \alpha_{i}^{2} \leq 2c^{2}$. Given $\varepsilon > 0$, choose a partition $0 = u_{0} < u_{1} < \ldots   < u_{M} = 1$ such that $u_{j} - u_{j-1} < \varepsilon^{2}/(2c^{2})$ for every $j$. The number of points in the partition can be chosen to be smaller than a constant times $1/\varepsilon^{2}$. Then
\begin{equation*}
    \E \sup_{u_{j-1} \leq s,t < u_{j}} |Z_{Ni}(s) - Z_{Ni}(t)|^{2} < \sqrt{N}\alpha_{i}\E[Z_{Ni}(u_{j}) - Z_{Ni}(u_{j-1})] < \frac{N \alpha_{i}^{2}\varepsilon^{2}}{2c^{2}},
\end{equation*}
where we have taken $Z_{Ni}(0) = 0$ and $Z_{Ni}(1) = \sqrt{N}\alpha_{i}$ for notational convenience. We deduce that for all sufficiently large $N$,
\begin{equation*}
    \sum_{i=N}^{\infty}\E \sup_{u_{j-1} \leq s,t < u_{j}} |Z_{Ni}(s) - Z_{Ni}(t)|^{2} < \varepsilon^{2}.
\end{equation*}
In other words, $\mathfrak{N}_{[]}(\varepsilon, \mathcal{F}, L^{2,N}) \lesssim {\varepsilon^{-2}}$, which verifies the entropy condition. Thus, {Theorem \ref{th::bracketing_clt}} implies that
\begin{equation*}
    -\sum_{i=N}^{\infty} (Z_{Ni}(u) -\E Z_{Ni}(u)) = \sqrt{N}\sum_{i=N}^{\infty}\alpha_{i}(u - H_{\rho_{i}}(u,V_{i}))
\end{equation*}
is asymptotically tight in $\ell^{\infty}(\mathcal{F})$.
\end{proof}

\begin{prop}\label{app:prop_marginal}
    Consider a vector $\mathbf{S}_N(\mathbf{u}) = [S_N(u_1),\ldots,S_N(u_d)]^T$ where each $u_j \in (0,1)$. We then have
\begin{align*}
    \sqrt{N} \mathbf{S}_N(\mathbf{u}) \overset{d}{\to} \mathcal{N}(\mathbf{0}, a^2 \Sigma)
\end{align*}
where $   \Sigma_{j,k} = \min\{u_j,u_k\} - u_j u_k$.
\end{prop}
\begin{proof}
We will use the Lindeberg-Feller CLT \cite[Proposition 2.27]{vanderVaart2000}, which we state below for convenience. For each $N$, let $\{Z_{N,N},\ldots,Z_{N,m_{N}}\}$ be independent r.v.s with finite variances that satisfy
\begin{align*}
    \sum_{i=N}^{m_{N}} \left(\E[Z_{Ni}^2]-\E[Z_{Ni}]^2\right) \to \sigma^2
\end{align*}
and the Lindeberg condition
\begin{align*}
    \sum_{i=N}^{m_{N}}\E[Z_{Ni}^2 \mathbbm{1}\left(|Z_{Ni} | > \varepsilon\right)] \to 0
\end{align*}
for every $\varepsilon > 0$. Then we have
\begin{align*}
    \sum_{i=N}^{m_{N}}(Z_{Ni} -\E[Z_{Ni}]) \overset{d}{\to} \mathcal{N}(0,\sigma^2).
\end{align*}
We take $m_{N}$ to be any sequence such that $N/m_{N} = o(1)$. Consider an arbitrary vector $\mathbf{t} \in \mathbb{R}^{d}$, and we study the convergence of
\begin{align*}
    \sqrt{N}\, \mathbf{t}^T \mathbf{S}_N(\mathbf{u})  = \sqrt{N}\sum_{j=1}^d \, t_j \, S_N(u_j) = \sum_{i=N}^{\infty}\sqrt{N}\sum_{j=1}^d \, t_j \, \alpha_i(u_{j} - H_{\rho_i}(u_{j},V_i)).
\end{align*}
For $i \geq N$, set
\begin{equation*}
    Z_{Ni} = \sqrt{N}\sum_{j=1}^d \, t_j \, \alpha_i H_{\rho_i}(u_{j},V_i),
\end{equation*}
which has expectation $\sqrt{N}\alpha_{i}\sum_{j=1}^{d}t_{j}u_{j}$. We have
\begin{equation*}
     |Z_{Ni}| \leq a\frac{\sqrt{N}}{i}\sum_{j=1}^d  |t_j | \leq \frac{a}{\sqrt{N}}\sum_{j=1}^d  |t_j |.
\end{equation*}
So for every $\varepsilon > 0$, $|Z_{Ni}| <\varepsilon$ eventually, which establishes the Lindeberg condition. It remains to study the limiting variance. Using the integral test, we obtain
\begin{align*}
   \E[Z_{Ni}^2] -\E[Z_{Ni}]^{2} &= a^2 \frac{N}{i^{2}} \sum_{j=1}^d \sum_{k = 1}^d t_j t_k \{\E\left[H_{\rho_{i}}(u_j,V_{i})\, H_{\rho_{i}}(u_k,V_{i})\right]-u_{j}u_{k}\}\\
    &= a^2 \frac{N}{i^{2}} \sum_{j=1}^d \sum_{k = 1}^d t_j t_k \{C_{\rho_{i}^{2}}(u_{j},u_{k})-u_{j}u_{k}\}.
\end{align*}
Taking $\rho^2 \to 1$ gives $C_{\rho^{2}}(u_{j},u_{k}) \to \min\{u_j,u_k\}$ as discussed in Section \ref{app:sec_copula}. Then
\begin{align*}
    &\left|\sum_{i=N}^{m_{N}}\left(\E[Z^2_{Ni}] -\E[Z_{Ni}]^2\right)  - a^{2}\sum_{j=1}^d \sum_{k = 1}^d t_j t_k \{\min\{u_j,u_k\}-u_{j}u_{k}\} \right|\\ &\leq   \sum_{j=1}^d \sum_{k = 1}^d\left|\sum_{i=N}^{m_{N}} \frac{a^{2}N}{i^{2}}\left[C_{\rho_{i}^{2}}(u_j,u_k) - \min\{u_j,u_k\}\right]\right|+ a^{2}\sum_{j=1}^d \sum_{k = 1}^d t_j t_k \{\min\{u_j,u_k\}-u_{j}u_{k}\} \left|\sum_{i=N}^{m_{N}}\frac{N}{i^{2}}-1\right|\\ 
    &\leq \sum_{j=1}^d \sum_{k = 1}^d\left|C_{\rho_{N}^{2}}(u_{j},u_{k}) - \min\{u_j,u_k\}\right|\left|\sum_{i=N}^{m_{N}} \frac{a^{2}N}{i^{2}}\right| + a^{2}\sum_{j=1}^d \sum_{k = 1}^d t_j t_k \{\min\{u_j,u_k\}-u_{j}u_{k}\} \left|\sum_{i=N}^{\infty}\frac{N}{i^{2}}-1\right|\\
    &+ a^{2}\sum_{j=1}^d \sum_{k = 1}^d t_j t_k \{\min\{u_j,u_k\}-u_{j}u_{k}\} \left|\sum_{i=m_{N}+1}^{\infty}\frac{N}{i^{2}}\right|.
\end{align*}
The first inequality uses the triangle inequality. The first and second terms in the last expression are $o(1)$. The third term is $O(N/m_{N})$, which is $o(1)$ by construction. Thus, we have
\begin{equation*}
     \sum_{i=N}^{m_{N}}\left(\E[Z^2_{Ni}] -\E[Z_{Ni}]^2\right)  \to \sigma^{2} = a^{2}\sum_{j=1}^d \sum_{k = 1}^d t_j t_k \{\min\{u_j,u_k\}-u_{j}u_{k}\} 
\end{equation*}
and we can apply the Lindeberg-Feller CLT to obtain 
\begin{equation*}
    \sum_{i=N}^{m_{N}}\sqrt{N}\sum_{j=1}^d \, t_j \, \alpha_i(u_{j} - H_{\rho_i}(u_{j},V_i)) \overset{d}{\to} \mathcal{N}(0,\sigma^2).
\end{equation*}
The final step is to show that the tail sums from $m_{N}+1$ onwards become asymptotically negligible. We do this by checking that $\sum_{i = m_{N}+1}^{\infty} \text{Var}[Z_{Ni}] \to 0$:
\begin{align*}
    \sum_{i = m_{N}+1}^{\infty} \text{Var}[Z_{Ni}] &= \sum_{i = m_{N}+1}^{\infty}a^2 \frac{N}{i^{2}} \sum_{j=1}^d \sum_{k = 1}^d t_j t_k \{C_{\rho_{i}^{2}}(u_{j},u_{k})-u_{j}u_{k}\}\\
    &\leq \sum_{i = m_{N}+1}^{\infty}a^2 \frac{N}{i^{2}} \left|\sum_{j=1}^d \sum_{k = 1}^d t_j t_k\right|\\
    &\to 0.
\end{align*}
Now the Cram\'er-Wold device completes the proof.
\end{proof}

\subsection{Theorem \ref{th:consistency_mean}}
To begin, we assume that $Y_{1:n}\iid P^*$ which has quantile function $Q^*$ satisfying {Assumption \ref{as:consistency_truth}}. Let $\{\mathcal{F}_{i}\}_{i \geq 1}$ denote the filtration where $\mathcal{F}_i := \sigma\left(Y_1,\ldots,Y_i\right)$. For now, we do not need to consider $Y_{n+1:\infty}$ arising from predictive resampling. 
 To start, we construct an almost supermartingale  as in {Theorem \ref{app:thm_asmart}} arising from {Algorithm \ref{alg:fit}}:
\begin{align*}
     \left(Q^*(u) - Q_{n}(u) \right)^2 &= (Q^*(u) - Q^\dagger_{n-1}(u))^2 + 2\alpha_{n}(Q^*(u) - Q^\dagger_{n-1}(u))\left( H_{\rho_n}\left(u,V_n\right) - u\right) \\
     &+ \alpha_{n}^2\left(H_{\rho_n}\left(u,V_n)\right) - u\right)^2.
\end{align*}
where for shorthand we write $V_n = P_{n-1}(Y_{n})$. We note that $|H_\rho(u,v) - u| \leq 1$.
Integrating with respect to $u$, we have
\begin{align*}
     d_2^2\left(Q^*, Q_{n} \right) &\leq d_2^2(Q^{\dagger}_{n-1}, Q^* )+ 2\alpha_{n}\int \, (Q^*(u) - Q^\dagger_{n-1}(u))\left( H_{\rho_n}\left(u,V_n\right) - u\right)du + \alpha_{n }^2.
\end{align*}
Now comes the key step due to the rearrangement. By {Assumption \ref{as:consistency_truth}}, both $Q^*$ and $Q_0$ are Lipschitz continuous and thus bounded on $[0,1]$. Since $H_\rho(u,v) - u$ is bounded on $u \in [0,1]$, $Q_n^{\dagger}$ is bounded on $[0,1]$. We can thus apply {Proposition \ref{app:prop_L2}}, which gives
$$
d_2^2(Q_{n}^{\dagger},Q^*) \leq d_2^2(Q_{n},Q^*),
$$
so we have 
\begin{align*}
     d_2^2\left( Q^{\dagger}_{n},Q^* \right) &\leq d_2^2(Q^{\dagger}_{n-1},Q^* )+ 2\alpha_{n}\int \, (Q^*(u) - Q^\dagger_{n-1}(u))\left( H_{\rho_n}\left(u,V_n\right) - u\right)du + \alpha_{n}^2.
\end{align*}
Taking the conditional expectation gives us
\begin{align*}
   \E\left[L_{n} \mid \mathcal{F}_{n-1}\right]   \leq L_{n-1} + 2\alpha_{n}\int \, (Q^*(u) - Q^\dagger_{n-1}(u))\left(  K_{n}\left(u\right)-u\right)du + \alpha_{n}^2, 
\end{align*}
where we write $L_n = d_2^2(Q^{\dagger}_n, Q^*)$ and 
$$
K_n(u) = \int  H_{\rho_n}(u,P_{n-1}(y))\, p^*(y) \, dy.
$$
We now subtract and add terms to get
\begin{align*}
   \E\left[L_{n} \mid \mathcal{F}_{n-1}\right]   \leq L_{n-1} -2\alpha_{n}\int \, (Q^*(u) - Q^\dagger_{n-1}(u))(  u - P^*( Q^\dagger_{n-1}(u)))\, du + \alpha_{n}^2 + 2\alpha_{n}\zeta_n
\end{align*}
where 
\begin{align}
    \zeta_n = \int(Q^*(u) - Q^\dagger_{n-1}(u))(  K_{n}\left(u\right)-P^*( Q^\dagger_{n-1}(u)))\, du.
\end{align}
It is not too hard to see that
\begin{align*}
    T(Q_{n-1}^{\dagger}):= \int \, (Q^*(u) - Q^\dagger_{n-1}(u))(u- P^*( Q^\dagger_{n-1}(u)))du \geq 0,
\end{align*}
where the positivity can be seen by applying $Q^*$, which is monotonically increasing, to each term in $u - P^*( Q^\dagger_{n-1}(u))$, giving
\begin{align*}
    u - P^*( Q^\dagger_{n-1}(u)) \geq 0 \implies Q^*(u) -Q^\dagger_{n-1}(u) \geq 0.
\end{align*}
To get a handle on $\zeta_n$, Cauchy-Schwartz gives us
\begin{equation*}
     |\zeta_n| \leq \sqrt{L_{n-1}}\sqrt{\int (  K_{n}\left(u\right)-P^*( Q^\dagger_{n-1}(u)))^2\, du}
\end{equation*}
and applying $\sqrt{x} \leq x + 1$ gives
\begin{equation*}
     |\zeta_n| \leq (L_{n-1} + 1)\sqrt{\kappa_n}
\end{equation*}
where we write 
\begin{align}\label{app:eq_kappan}
    \kappa_n = {\int(  K_{n}\left(u\right)-P^*( Q^\dagger_{n-1}(u)))^2\, du}
\end{align} 
This gives the original inequality as 
\begin{align*}
   \E[L_{n} \mid \mathcal{F}_{n-1}] \leq \left(1 + 2\alpha_{n}\sqrt{\kappa_n}\right)L_n-2\alpha_{n}T(Q_{n-1}^{\dagger}) + \alpha_{n}^2 + 2\alpha_{n}\sqrt{\kappa_n}.
\end{align*}
We now seek to apply {Theorem \ref{app:thm_asmart}} with $B_i = 2\alpha_i \sqrt{\kappa_i}$, $C_i =\alpha_{i}^2 + 2\alpha_{i}\sqrt{\kappa_i}$ and $D_i =2\alpha_{i}T(Q_{i-1}^{\dagger})$, where all terms are positive.
A sufficient condition for $\sum_{i = 1}^\infty B_i < \infty$ and $\sum_{i = 1}^\infty C_i < \infty$ a.s. is if 
$$\sum_{i=1}^\infty \alpha_{i}\sqrt{\kappa_i} < \infty \quad \text{a.s. }$$
which we now prove.

\subsubsection{Controlling $\kappa_n$}
We begin by upper bounding $\kappa_i$, which intuitively measures how far the copula kernel $H_\rho(u,v)$ is from the indicator function $\mathbbm{1}\left(v \leq u\right)$ (averaged over $u,v$). The result is simple to state, but the proof is surprisingly quite involved, and requires specific properties of the bivariate Gaussian copula. The difficulty of the proof arises from the copula's dependence on $P_n$, as we require bounds independent of $P_n$ and $P^*$. This
highlights that while the bivariate copula is very useful for attaining the martingale and coherence required for Bayesian inference, its adaptivity makes it harder to study its properties for estimation. 
\begin{lem}\label{app:lem_kappan}
    Let $\kappa_n$ be defined as in (\ref{app:eq_kappan}), and suppose $\rho_n \to 1$. For sufficiently large $n$, we have that $\kappa_n$ satisfies
    \begin{align*}
\sqrt{\kappa_n} \leq K (1-\rho_n^2)^{1/4}
    \end{align*}
    for some positive finite constant $K$.
\end{lem}
\begin{proof}
    To start, we can write
\begin{align*}
    P^*(Q_{n-1}^{\dagger}(u)) & = \int \mathbbm{1}(y \leq Q_{n-1}^{\dagger}(u)  )\, dP^*(y) = \int \mathbbm{1}(P_{n-1}(y)\leq u )\, dP^*(y).
\end{align*}
This gives
\begin{align*}
    \kappa_n &= \int \left[\int(H_{\rho_n}(u,P_{n-1}(y)) - \mathbbm{1}(P_{n-1}(y) \leq u)) \,dP^*(y)\right]^2\,  du\\
    &\leq \int\int(H_{\rho_n}(u,P_{n-1}(y)) - \mathbbm{1}(P_{n-1}(y) \leq u))^2 \, du \,  dP^*(y).
\end{align*}
Let us write the inner integral as a function of a  general $\rho,v \in (0,1)$:
\begin{align*}
    U(v;\rho) := \int\int(H_{\rho}(u,v) - \mathbbm{1}(v \leq u))^2 \, du.
\end{align*}
Fortunately, we can control  $U(v;\rho)$ by taking $\rho \to 1$. To see this,  note that
\begin{align}\label{app:eq_intHuv}
    \int(H_{\rho}(u,v) - \mathbbm{1}(v \leq u))^2 \, du = \int H^2_{\rho}(u,v) \, du  - 2\int_{v}^1 H_\rho(u,v) \, du + \int_v^1 \,du.
\end{align}
A change of variables with $z_u = \Phi^{-1}(u)$ gives 
\begin{equation*}
    H_\rho(u,v) = \Phi(a + bz_u)
\end{equation*}
where 
\begin{align*}
a = -\frac{\rho}{\sqrt{1-\rho^2}}z_v, \quad b = \frac{1}{\sqrt{1-\rho^2}}
\end{align*}
where $z_v = \Phi^{-1}(v)$. We now require some integrals of Gaussian CDF and density terms, which can be found in \cite{Owen1980}.
For the first term of (\ref{app:eq_intHuv}), we have
\begin{align*}
\int_{-\infty}^\infty \Phi(a + bz_u)^2\phi(z_u) \, dz_u &= \Phi_2\left(\frac{a}{\sqrt{1+b^2}},\frac{a}{\sqrt{1+b^2}}; \rho = \frac{b^2}{1+b^2}\right)
\end{align*}
where 
$\Phi_2(\mu_1, \mu_2; \rho)$ is the standard bivariate normal CDF evaluated at $\mu_1,\mu_2$ with correlation $\rho$. This follows from codes (20,010.4) and (3.5) from \citet{Owen1980}.

For the second term of  (\ref{app:eq_intHuv}), we have 
\begin{align*}
    \int_{z_v}^\infty \Phi(a + bz_u) \,\phi(z_u)\, dz_u &=   \int_{-\infty}^\infty \Phi(a + bz_u) \,\phi(z_u)\, dz_u -    \int_{-\infty}^{z_v} \Phi(a + bz_u) \,\phi(z_u)\, dz_u\\
    &=\Phi\left(\frac{a}{\sqrt{1+b^2}}\right) - \Phi_2\left(\frac{a}{\sqrt{1+b^2}}, z_v; \rho = \frac{-b}{\sqrt{1+b^2}}\right)\\
    &=\Phi_2\left(\frac{a}{\sqrt{1+b^2}}, -z_v; \rho = \frac{b}{\sqrt{1+b^2}}\right)
\end{align*}
where the second line comes from codes (10,010.1) and (10,010.1) from  \citet{Owen1980}, and the third line comes from the identity $\Phi(a) - \Phi_2(a,b; \rho) = \Phi_2(a,-b; -\rho)$.  The final term of  (\ref{app:eq_intHuv}) is simply $1-v$. Putting this together, we have that
\begin{align*}
  U(v; \rho) &\leq  \Phi_2\left(\frac{a}{\sqrt{1+b^2}},\frac{a}{\sqrt{1+b^2}}; \rho = \frac{b^2}{1+b^2}\right) -2\Phi_2\left(\frac{a}{\sqrt{1+b^2}}, -z_v; \rho = \frac{b}{\sqrt{1+b^2}}\right)+ (1-v)\\
  &= \Phi_2\left(\frac{-\rho}{\sqrt{2-\rho^2}}z_v,\frac{-\rho}{\sqrt{2-\rho^2}}z_v; \rho = \frac{1}{2-\rho^2}\right) -2\Phi_2\left(\frac{-\rho}{\sqrt{2-\rho^2}}z_v, -z_v; \rho = \frac{1}{\sqrt{2-\rho^2}}\right)+ (1-v)
\end{align*}
We will first write the above in terms of the tail probability of a bivariate normal distribution, $L(h,k;\rho) = P(x > h, y > k; \rho)$, which is easier to bound and satisfies
\begin{align*}
L(h,k; \rho) &= \Phi_2(-h,-k;\rho)
\end{align*}
This gives us 
\begin{align*}
 U(v; \rho) &\leq  L\left(\frac{\rho}{\sqrt{2-\rho^2}}z_v,\frac{\rho}{\sqrt{2-\rho^2}}z_v; \frac{1}{2-\rho^2}\right) - 2L\left(\frac{\rho}{\sqrt{2-\rho^2}}z_v, z_v; \frac{1}{\sqrt{2-\rho^2}}\right) + (1-v)\\
 &= L(h,h; \bar{\rho}^2) - 2L(h, z_v; \bar{\rho}) + Q(z_v)
\end{align*}
where we define $h = \frac{\rho}{\sqrt{2-\rho^2}}z_v$, $\bar{\rho} = \frac{1}{\sqrt{2-\rho^2}}$, and  $Q(z) = 1-\Phi(z)$ is the Gaussian tail probability.
 We first show that this upper bound is symmetric around $z_v = 0$. Using the identity 
 $$L(h,k;\rho) = 1-\Phi(h) - \Phi(k) + L(-h,-k;\rho)$$ 
 from the above, we can see that $U(v; \rho) = U(1-v; \rho)$ noting that $z_{1-v} = -z_v$.
As a result, we just need to bound  $U(v;\rho)$ for $z_v \geq 0$ (i.e. $v \geq 0.5$). 

We note that $L(h,h; \bar{\rho}^2) \leq L(h,h;\bar{\rho})$ as $L$ is increasing with $\bar{\rho}$, and furthermore $L(h,z_v;\bar{\rho}) \geq  L(z_v,z_v;\bar{\rho})$ as $z_v \geq h$. This then gives us
\begin{align*}
  U(v;\rho)  &\leq  L(h,h; \bar{\rho}) - 2L(z_v, z_v; \bar{\rho}) + Q(z_v).
\end{align*}
One can show that 
$$
\frac{\partial L(h,h; \bar{\rho})}{\partial h} = -2\Phi(-\theta h) \phi(h)
$$
where $\theta = \sqrt{\frac{1-\bar{\rho}}{1+\bar{\rho}}}$. A Taylor expansion gives
\begin{align*}
    L(z_v,z_v;\bar{\rho} ) &= L(h,h;\bar{\rho})  - 2\Phi(-\theta \tilde{h})\,\phi(\tilde{h})\,(z_v- h)\\
    &\geq L(h,h;\bar{\rho})  - 2\Phi(-\theta {h})\,\phi({h})\,(z_v- h),
\end{align*}
where $h\leq \tilde{h} \leq z_v$ . The second line follows as the function $\Phi(-\theta h)\,\phi(h)$ is monotonically decreasing with $h$ for $h\geq 0$.  This then implies
\begin{align*}
   U(v;\rho)  \leq  2\Phi(-\theta {h})\,\phi({h})(z_v- h) +Q(z_v) -  L(z_v,z_v;\bar{\rho}).
\end{align*}
We can upper bound the first term as
\begin{align*}
    \phi(h)(z_v - h) &= \left(\frac{\sqrt{2-\rho^2}}{\rho} - 1\right)\,\phi(h)\,h 
    \leq \phi(1)\left(\frac{\sqrt{2-\rho^2}}{\rho} - 1\right)
\end{align*}
and $\Phi(-\theta h) \leq 0.5$. For the second term, we can compute
\begin{align*}
    \frac{\partial}{\partial z_v} (Q(z_v) - L(z_v,z_v;\bar{\rho})) =   [2\Phi(-\theta z_v)-1]\, \phi(z_v).
\end{align*}
The derivative is always non-positive for $z_v \geq 0$ and is equal to 0 at $z_v = 0$, so the maximum value must be
\begin{align*}
    Q(0) - L(0,0;\bar{\rho}) = \frac{1}{2}- \left(\frac{1}{4} + \frac{1}{2\pi}\arcsin(\bar{\rho})\right).
\end{align*}
 Together, this implies
\begin{align*}
U(v;\rho)  &\leq  \phi(1)\left(\frac{\sqrt{2-\rho^2}}{\rho} - 1\right) + \frac{1}{4} - \frac{1}{2\pi}\arcsin\left(\frac{1}{\sqrt{2-\rho^2}}\right) \\
    &=\phi(1)\left(\frac{\sqrt{2-\rho^2}}{\rho} - 1\right) + \frac{1}{2\pi}\arctan\sqrt{1-\rho^2}
\end{align*}
For $x \geq 0$, we have $\arctan x \leq x$ as
\begin{align*}
    \frac{d }{dx}\arctan x = \frac{1}{1+x^2} \leq \frac{d }{dx}x=1
\end{align*}
and $\arctan x' = x'$ at $x' = 0$. We thus have
\begin{align*}
U(v;\rho)  &\leq \phi(1)\left(\frac{\sqrt{2-\rho^2}}{\rho} - 1\right) + \frac{1}{2\pi}\sqrt{1-\rho^2}
\end{align*}
Finally, for sufficiently large $\rho$ we have that 
$$
\frac{\sqrt{2-\rho^2}}{\rho} - 1 \leq \sqrt{1-\rho^2}.
$$
This follows because the roots of $g(\rho) = 
{\sqrt{2-\rho^2}}/{\rho} - 1 - \sqrt{1-\rho^2}$ occur at $\rho \approx 0.7184$ and $\rho = 1$, $g(\rho)$ is continuous on $\rho \in (0,1]$ and $g(\rho)$ is negative for some point in between the two roots. For $\rho$ sufficiently close to $1$, we thus have
\begin{equation}
    \begin{aligned}
        U(v; \rho) \leq K^2\sqrt{1-\rho^2}
    \end{aligned}
\end{equation}
for some finite and positive $K$. If $\rho_n \to 1$, this thus gives
$$
\sqrt{\kappa_n} \leq K (1-\rho_n^2)^{1/4}
$$
for sufficiently large $n$. 
\end{proof}

\subsubsection{Almost supermartingale}
If $\alpha_n = a(n+1)^{-1}$, {Lemma \ref{app:lem_kappan}} implies that setting
$$
1-\rho_n^2 = O(n^{-k})
$$
for some $k > 0$ is sufficient for $\sum_{i=1}^\infty \alpha_i \sqrt{\kappa_i} < \infty$.
Given {Assumptions \ref{as:alpha}} and {\ref{as:bandwidth}}, we have that $L_{n}$ is an almost supermartingale, so we can apply {Theorem \ref{app:thm_asmart}}. This implies $L_n \to L_\infty <\infty$ a.s. and more importantly,
$$
\sum_{i=1}^\infty \alpha_i T(Q_{i-1}^{\dagger}) < \infty \quad \text{a.s.}
$$
As $\sum_{i = 1}^\infty \alpha_i = \infty$, one can verify that the above implies
$$
\liminf_{n} T(Q_n^{\dagger}) = 0 \quad \text{a.s.}
$$
which implies there is a subsequence $n_j$ on which $\lim_{n_j \to \infty}T(Q_{n_j}^{\dagger}) = 0$.
We will now use the fact that $Q^*$ is $M$-Lipschitz on $[0,1]$ from {Assumption \ref{as:consistency_truth}}, which gives
\begin{align*}
  |Q^*(u) - Q(u)| \leq M|u - P^*(Q(u))|
\end{align*}
for all $ u \in [0,1]$. We thus have
\begin{align*}
       T(Q_{n}^{\dagger})&= \int \, (Q^*(u) - Q^\dagger_{n}(u))(u- P^*( Q^\dagger_{n}(u)))\, du
       \\&= \int \, |Q^*(u) - Q^\dagger_{n}(u)|\, |u- P^*( Q^\dagger_{n}(u))|\,du \\
&\geq M^{-1}\int \, (Q^*(u) - Q^\dagger_{n}(u))^2\, du\, = M^{-1} L_n
\end{align*}
where the second equality follows from the positivity of $T(Q)$. Applying this to the subsequence $n_j$ gives $ L_{n_j}\leq M T(Q_{n_j}^{\dagger})$, so $L_{n_j}\to 0$ along this subsequence. Since $L_n \to L_\infty$ a.s., we have $L_\infty = 0$ a.s.

\subsection{Theorem \ref{th:consistency}}

Consider now the same setting as the proof of {Theorem \ref{th:consistency_mean}}. 
To begin, we first show that $Q_n^{\dagger}$ satisfies {Assumptions \ref{as:L2}} and {\ref{as:weak_deriv}} for each $n$. We note that the update function $(u-H_\rho(u,v))$ has continuous partial derivative in $u$ for all $v \in [0,1]$ (as shown in (\ref{app:eq_partial})). As a result, $Q_1$ is a sum of a Lipschitz function $Q_0$ and a continuously differentiable function, so $Q_1\in H^1((0,1))$, and $Q_1^{\dagger}\in H^{1}((0,1))$ by {Theorem \ref{app:thm_sobolev}}. Repeating the argument gives $Q_n^{\dagger}\in H^1((0,1))$, thus satisfying {Assumptions \ref{as:L2}} and {\ref{as:weak_deriv}}.

We now extend the probability space. For each $n$, define $Q_{n\infty}^{\dagger}$ as the random function with realizations in $H^1((0,1))$ arising from {Algorithm \ref{alg:QMP}} starting from $Q_n^{\dagger}$. The existence of $Q_{n\infty}^{\dagger}\in H^{1}((0,1))$  is guaranteed by {Theorem \ref{th:sobolev_mart}}. There are a few possible constructions of this space, for example we can let $Y_1,Y_2,\ldots$ be i.i.d. r.v.s from $P^*$, and independently let $V_1,V_2,\ldots$ be i.i.d. r.v.s from $\mathcal{U}(0,1)$. We can then {define} the following:
\begin{align*}
    S_n(u) =  \sum_{i = n+1}^\infty \alpha_i[u- H_{\rho_i}(u,V_i)], \quad Q_{n\infty}(u) = Q_n^{\dagger}(u) + S_n(u).
\end{align*}
Another option is to let $V_{n,n+1}, V_{n,n+2},\ldots$ be distinct independent sequences of uniform r.v.s for each $n$. Either way, this does not affect the next step, as the distribution of $Q_{n\infty}^{\dagger}$ for each $n$ is unchanged.

To show posterior consistency, we apply Markov's inequality which gives
 \begin{align*}
        \Pi_n\left(Q^{\dagger}_{n\infty}: d_2\left(Q^{\dagger}_{n\infty},Q^*\right) \geq \varepsilon \mid Y_{1:n}\right) \leq \frac{1}{\varepsilon^2 }\E\left[d^2_2\left(Q^{\dagger}_{n\infty},Q^*\right)\mid Y_{1:n}\right].\end{align*}
As $Q_{n}^{\dagger}$ and $Q_{n\infty}^{\dagger}$ are essentially bounded by {Proposition \ref{app:prop_bounded}}, 
we can apply {Proposition \ref{prop:rearr}} which gives $d^2_2(Q^{\dagger}_{n\infty},Q^*)\leq d^2_2\left(Q_{n\infty},Q^*\right)$, and thus
 \begin{align*}
        \Pi_n\left(Q^{\dagger}_{n\infty}: d_2\left(Q^{\dagger}_{n\infty},Q^*\right) \geq \varepsilon \mid Y_{1:n}\right) \leq \frac{1}{\varepsilon^2 }\E\left[d^2_2\left(Q_{n\infty},Q^*\right)\mid Y_{1:n}\right].\end{align*}
 Applying the triangle inequality gives
\begin{align*}
&E\left[d^2_2\left(Q_{n\infty}, Q^*\right)\mid Y_{1:n}\right] \leq \E\left[\left(d_2\left(Q_{n\infty}, Q_n^{\dagger}\right) + d_2\left(Q_n^{\dagger}, Q^*\right)\right)^2\mid Y_{1:n}\right] \\
  &=\E\left[ d^2_2\left(Q_{n\infty}, Q_n^{\dagger}\right)\mid Y_{1:n}\right] +2d_2\left(Q_n^{\dagger}, Q^*\right)\E\left[ d_2\left(Q_{n\infty}, Q_n^{\dagger}\right)\mid Y_{1:n}\right] + d^2_2\left(Q_n^{\dagger}, Q^*\right). 
\end{align*}
To compute the first term, we have
\begin{align*}
    d^2_2\left(Q_N,Q_n^{\dagger}\right) = d^2_2(Q_{N-1},Q_n^{\dagger}) + d^2_2\left(Q_N,Q_{N-1}\right) - 2 \int \left(Q_N(u)-Q_{N-1}(u)\right)\left(Q_{N-1}(u)-Q_n^{\dagger}(u)\right) \, du.
\end{align*}
As $Q_N(u)$ is a martingale, we have that
\begin{align*}
  \E\left[d^2_2\left(Q_N,Q_n^{\dagger}\right) \mid Y_{1:N-1} \right]&= d^2_2(Q_{N-1},Q_n^{\dagger}) + \alpha_{N}^2 \int \, \int \left(u - H_{\rho_N}(u,v)\right)^2 \, du \, dv\\
   &\leq  d^2_2(Q_{N-1},Q_n^{\dagger}) + \alpha_{N}^2.
\end{align*}
Iterating further, we have that
\begin{align*}
  \E\left[d^2_2\left(Q_{n\infty},Q_n^{\dagger}\right) \mid Y_{1:n} \right]&\leq \sum_{i = n+1}^\infty \alpha_i^2 = O(n^{-1}).
\end{align*}
Let us consider the second term. We have
\begin{align*}
   \E\left[d_2\left(Q_{n\infty},Q_n^{\dagger}\right) \mid Y_{1:n} \right]&\leq \sqrt{E\left[d^2_2(Q_{n\infty},Q_n^{\dagger}) \mid Y_{1:n} \right]}\\
    &= O(n^{-1/2}).
\end{align*}
Putting this together, we have
\begin{align}\label{app:eq_on}
    \Pi_n \left(Q^{\dagger}_{n\infty}: d_2\left(Q_{n\infty}^{\dagger}, Q^*\right) \geq \varepsilon  \mid Y_{1:n}\right) \leq \frac{1}{\varepsilon^2}\left[O(n^{-1}) + 2O(n^{-1/2}) \, d_2(Q_n^{\dagger},Q^*) + d^2_2(Q_n^{\dagger},Q^*)\right].
\end{align}
We thus have the above going to 0 as long as $d^2_2(Q_n^{\dagger},Q^*) \to 0$ $P^*$-a.s., which follows from {Theorem \ref{th:consistency_mean}}.

\subsection{Theorem \ref{th:contraction_mean} }
This proof fortunately recycles many steps from the proof of {Theorem \ref{th:consistency_mean}}, with additional steps inspired by \cite{Aboubacar2014}. Let us begin again with the almost supermartingale construction:
\begin{align*}
   \E[L_{n} \mid \mathcal{F}_{n-1}] \leq \left(1 + 2\alpha_{n}\sqrt{\kappa_{n}}\right)L_{n-1}-2\alpha_{n}T(Q_{n-1}^{\dagger}) + \alpha_{n}^2 + 2\alpha_{n}\sqrt{\kappa_{n}},
\end{align*}
where  $\sqrt{\kappa_n}  = O(n^{-k/4})$, which is the error term controlled by the bandwidth with $k > 0$. From the Lipschitz condition in {{Assumption \ref{as:consistency_truth}}}, we again have
\begin{align*}
    T(Q_{n-1}^{\dagger}) \geq M^{-1} L_{n-1}.
\end{align*}
Putting this together, we get
\begin{align*}
   \E[L_{n} \mid \mathcal{F}_{n-1}] &\leq \left[1 - 2 \alpha_n\left(M^{-1} -\sqrt{\kappa_n}\right)\right]L_{n-1} + \alpha_{n}^2 + 2\alpha_{n}\sqrt{\kappa_{n}}\\
    &= \left[1 - 2 \alpha_n (M^{-1}+ O(n^{-k/4}))\right]L_{n-1} + \alpha_{n}^2 + 2\alpha_{n}\sqrt{\kappa_{n}}.
\end{align*}
Premultiplying by $n^{\delta}$ for $\delta < 1$, we have
\begin{align*}
   \E[n^{\delta} L_{n} \mid \mathcal{F}_{n-1}] &\leq \left[1 - 2 an^{-1} (M^{-1}+ O(n^{-k/4}))\right][1+\delta n^{-1} + O(n^{-2})] (n-1)^{\delta }L_{n-1} \\
    &+ O\left(n^{-(2-\delta)}\right) + O(n^{\delta - 1 - k/4})
\end{align*}
where we have used the fact that
\begin{align*}
  \left(\frac{n}{n-1}\right)^\delta = 1 + \frac{\delta}{n} + O(n^{-2})
\end{align*}
which is also used in \citet{Aboubacar2014}.
Simplifying, we get
\begin{align*}
   \E[n^{\delta} L_{n} \mid \mathcal{F}_{n-1}] &\leq \left[1 - n^{-1}(2aM^{-1} -\delta + {o}(1))\right] (n-1)^{\delta }L_{n-1} +O\left(n^{-(2-\delta)}\right) + O(n^{\delta - 1 - k/4})
\end{align*}
If we choose $a$ such that $a> \frac{M\delta}{2}$, then $2aM^{-1} - \delta  +o(1)> 0$ eventually for sufficiently large $n$. If we further assume $k > 4$, then the last terms is $O(n^{-(2-\delta)})$, so we just need $\sum_{n = 1}^\infty n^{-(2-\delta)} < \infty$ for the bounded variance condition, which holds if $\delta < 1$. Summarizing, we have that under {Assumption \ref{as:contraction_init}}, for sufficiently large $n$, there exists some positive constant $0 < K <\infty$ such that:
\begin{align*}
   \E[n^{\delta} L_{n} \mid \mathcal{F}_{n-1}] &\leq (n-1)^\delta L_{n-1} +\underbrace{O\left(n^{-(2-\delta)}\right)}_{C_n} - \underbrace{n^{-1}K\, (n-1)^\delta L_{n-1} }_{D_n}
\end{align*}
where $B_n = 0$ and $n^{\delta}L_n$ is an almost supermartingale, as $C_n, D_n \geq 0$ and $\sum_{i = 1}^\infty C_i <\infty$ a.s. As a result, {Theorem \ref{app:thm_asmart}} gives us $n^{\delta} L_n \to X_\infty$ a.s. under {Assumption \ref{as:contraction_init}}. Finally, we have that $X_\infty = 0$ a.s. which follows from $\sum_{i=1}^\infty i^{-(1-\delta)} L_{i} < \infty$ a.s. \\

\subsection{Theorem  \ref{th:contraction}}

We start again have from (\ref{app:eq_on})
\begin{align*}
    \Pi_n \left(Q^{\dagger}_{n\infty}: d_2\left(Q^{\dagger}_{n\infty}, Q^*\right) \geq \varepsilon  \mid Y_{1:n}\right)
    &\leq \frac{1}{\varepsilon^2}\left[O(n^{-1}) + 2O(n^{-1/2}) \, d_2(Q^{\dagger}_n,Q^*) + d_2^2(Q^{\dagger}_n,Q^*)\right].
\end{align*}
From {Theorem \ref{th:contraction_mean}}, for each $0 < \delta <1$, we have that
\begin{align*}
    \Pi_n \left(Q^{\dagger}_{n\infty}: d_2\left(Q^{\dagger}_{n\infty}, Q^*\right) \geq \varepsilon  \mid Y_{1:n}\right) 
    &\leq \frac{1}{\varepsilon^2}\left[O(n^{-1}) + O(n^{-(\delta+1)/2}) + O(n^{-\delta})\right] \quad \text{a.s.}[P^*]\\
    &\leq  \frac{1}{\varepsilon^2} O(n^{-\delta}) \quad \text{a.s.}[P^*]
\end{align*}
The above means that for all $\varepsilon > 0$, for any $0<\delta < 1$, there exists some constant $B_{\delta} < \infty$ such that  we have 
\begin{align*}
    \Pi_n \left(Q^{\dagger}_{n\infty}: d_2(Q_{n\infty}^{\dagger},Q^*) \geq \varepsilon \mid Y_{1:n}\right) \leq \frac{B_\delta n^{-\delta}}{\varepsilon^2} \quad \text{a.s.}[P^*]
\end{align*}
for sufficiently large $n$.
Now choose an arbitrary $0 < \delta < 1$, and also choose $\delta < \delta' < 1$ with corresponding $B_{\delta'}$. If we plug-in $\varepsilon_n = K n^{-\delta/2}$ as in {Theorem \ref{th:contraction}} for an arbitrary finite positive constant $K$, then we have
\begin{align*}
     \Pi_n (Q^{\dagger}_{n\infty}: d_2(Q_{n\infty}^{\dagger},Q^*) \geq K\varepsilon_n \mid Y_{1:n}) &\leq \frac{B_{\delta'}n^{-\delta'}}{\varepsilon_n^2}  \quad \text{a.s.}[P^*] \\
     &= \frac{B_{\delta'}}{K^2}n^{-\left(\delta' - \delta\right)} \quad \text{a.s.}[P^*]
\end{align*}
for sufficiently large $n$. Since $\delta' > \delta$ can always be chosen, we have that the above goes to 0 with $n$ for any $K>0$.

\subsection{Proposition \ref{prop:mu}}
We showed in the proof of {Theorem \ref{th:consistency}} that $Q_n^{\dagger}$ satisfies {Assumptions \ref{as:L2}} and {\ref{as:weak_deriv}}, so we can apply {Proposition \ref{prop:L2_mart}} or {Theorem \ref{th:sobolev_mart}}. The probability space can be constructed in the same way as {Theorem \ref{th:consistency}}.

Fix $n\geq 1$ and define the filtration $\{\mathcal{F}_N\}_{N \geq n}$ with $\mathcal{F}_N = \sigma(Y_1,\ldots,Y_n,V_{n+1},\ldots,V_N)$ and $\mathcal{F}_n =\sigma(Y_1,\ldots,Y_n)$. We will write $\E[\cdot \mid \mathcal{F}_n]$ as $\E[\cdot \mid Y_{1:n}]$ to make it clear that it is conditioned on `real' data.
To begin, we highlight the very useful property that
\begin{align*}
  \mu_{nN} :=  \int_0^1 Q_{nN}^{\dagger}(u)\, du  = \int_0^1 Q_{nN}(u)\, du
\end{align*}
which follows directly from {Lemma \ref{app:lem_equi}} with $h(x) = x$. The above also holds for $N = \infty$, where the existence of $Q_{n\infty}$ is guaranteed by {Proposition \ref{prop:L2_mart}}.

This is particularly convenient as we do not need to consider the rearrangement procedure to study the distribution of $\mu_{nN}$.
Note that this property is not unique to the mean functional. Another subtle but important point is that $\{\mu_{nN}\}_{N\geq n}$ is a martingale, even if $Q_{nN}^{\dagger}$ is not, due to the above property, which follows from
\begin{align*}
    \E\left[\mu_{nN} \mid \mathcal{F}_{N-1}\right]
    &=\int_0^1\E\left[Q_{nN}(u) \mid \mathcal{F}_{N-1}\right]\, du\\&=\int_0^1 Q_{n,N-1}(u) du\\&=\mu_{n,N-1}.
\end{align*}
This arises from the linearity of the mean, which\textit{ is} unique to the mean and does not apply for other functionals. We can then directly show that $\mu_N$ is bounded in $L^2$, as $E[X]^2 \leq\E[X^2]$ gives 
\begin{align*}
\sup_{N\geq n} \E\left[\mu_{nN}^2 \mid Y_{1:n}\right] \leq \sup_{N\geq n} \E\left[\|Q_{nN}\|^2_2\mid Y_{1:n} \right] < \infty, 
\end{align*}
where the boundedness was shown in {Proposition \ref{prop:L2_mart}}. As a result, $\mu_{nN}$ is a martingale bounded in $L^2$, so there exists a finite $\tilde{\mu}_{n\infty}$ such that  $\mu_{nN} \to \tilde{\mu}_{n\infty}$ a.s. and $E[\tilde{\mu}_{n\infty} \mid Y_{1:n}] = \mu_n$. 
Finally, we have
\begin{align*}
    (\mu_{n\infty} - \mu_{nN})^2 \leq d_2^2(Q_{n\infty},Q_{nN}) \to 0 \quad \textnormal{a.s.}
\end{align*}
from {Proposition \ref{prop:L2_mart}}, so $\tilde{\mu}_{n\infty} = \mu_{n\infty}$ a.s.
We thus have the first part of {Proposition \ref{prop:mu}}, that is $E[\mu_{n\infty} \mid Y_{1:n}] = \mu_n$ a.s. for each $n$. 

For the posterior variance, we note that
\begin{align*}
  \E[(\mu_{n\infty} - \mu_{n})^2\mid Y_{1:n}]=\sum_{i =n+1}^\infty  \sum_{j =n+1}^\infty \alpha_i \alpha_j \E[Z_i Z_j]
\end{align*}
where $Z_i = \int_0^1 (u- H_{\rho_i}(u,V_i))\, du$. 
As $Z_i$ is independent of $Z_j$ for $i \neq j$ and are zero-mean, the cross-terms are zero, so we just have
\begin{align*}
     \E[(\mu_{n\infty} - \mu_{n})^2\mid Y_{1:n}]=\sum_{i =n+1}^\infty \alpha_i^2 \, \E[Z_i^2].
\end{align*}
We can then show from {Lemma \ref{app:lem_covariance}} that
\begin{align*}
    \E[Z_i^2] = \int_0^1\int_0^1 C_{\rho_i^2}(u,v)\,du\, dv - \frac{1}{4}\cdot
\end{align*}
As $C_\rho(u,v) \leq 1$ and $\lim_{\rho \to 1} C_{\rho^2}(u,v) = \min(u,v)$, dominated convergence gives 
\begin{align*}
    \lim_{i \to \infty}\E[Z_i^2] &= \int_0^1 \int_0^1 \min(u,v)\,du \, dv - \frac{1}{4}\\
    &= \frac{1}{12}
\end{align*}
Scaling by ${n}$ gives us
\begin{align*}
  {n}\, \E[(\mu_{n\infty} - \mu_{n})^2\mid Y_{1:n}]&=a^2 \sum_{i = n+1}^\infty\frac{{n}}{(i+1)^2}\,  \E[Z_i^2].
\end{align*}
Now consider
\begin{align*}
   \left|  \sum_{i = n+1}^\infty\frac{{n}}{(i+1)^2}\,  \E[Z_i^2]  - \frac{1}{12}\right| \leq    \left|  \sum_{i = n+1}^\infty\frac{{n}}{(i+1)^2}\,  \left(\E[Z_i^2]  - \frac{1}{12}\right) + \frac{1}{12}\left(\sum_{i = n+1}^\infty \frac{{n}}{(i+1)^2} - 1\right)\right|\\
   \leq \sum_{i = n+1}^\infty\frac{{n}}{(i+1)^2}\,  \left|\E[Z_i^2]  - \frac{1}{12}\right| + \frac{1}{12}\left|\sum_{i = n+1}^\infty \frac{{n}}{(i+1)^2} - 1 \right|. 
\end{align*}
For each $\varepsilon > 0$, for sufficiently large $n$, we have $\sup_{i \geq n+1}\left|\E[Z_i^2]- \frac{1}{12} \right|< \varepsilon$. Furthermore, we have
\begin{align*}
    \left|n\sum_{i = n+1}^\infty \frac{1}{(i+1)^2} - 1 \right| \to 0,
\end{align*}
which follows as $(n+2)^{-1} \leq\sum_{i = n+1}^\infty (i+1)^{-2} \leq (n+1)^{-1}$ from the integral test.
We thus have
\begin{align*}
     {n}\, \E[(\mu_{n\infty} - \mu_{n})^2\mid Y_{1:n}] \to \frac{a^2}{12} \quad \textnormal{a.s.}[P^*]
\end{align*}
Note that the result does not depend on the convergence of $\mu_n$, as this will always act as the center of the posterior.

\subsection{Theorem \ref{th:approx_GP}}
We begin with showing the inequality. Once again, we have
\begin{align*}
    S_n(u) = \sum_{i = n+1}^\infty \alpha_i(u - H_{\rho_i}(u,V_i))
\end{align*}
where $V_i \iid \mathcal{U}(0,1)$ and the existence of the random function $S_n$ with realizations in $H^1((0,1))$ is guaranteed by {Theorem \ref{th:sobolev_mart}} as $\alpha_i$ and $\rho_i$ satisfy {Assumptions \ref{as:alpha}}  and {\ref{as:bandwidth}} respectively. It is clear that this function has mean 0. To begin, we have the following lemma.
\begin{lem}
    The covariance function $k_n(u,u') := \E[S_n(u)\, S_n(u')]$ takes the form
    \begin{align*}
        k_n(u,u') = \sum_{i = n+1}^\infty \alpha_i^2 \,[C_{\rho_i^2}(u,u') - uu'].
    \end{align*}
\end{lem}
\begin{proof}
A direct calculation gives
    \begin{align*}
        k_n(u,u') &= \sum_{i = n+1}^\infty \alpha_i \alpha_j \E\left[\left(u - H_{\rho_i}(u,V_i)\right)\left(u' - H_{\rho_j}(u',V_j)\right)\right]\\
        &=  \sum_{i = n+1}^\infty \alpha_i^2 \left[C_{\rho_i^2}(u,u')- u u'\right].
    \end{align*}
    The cross-terms are zero, so we can invoke {Lemma \ref{app:lem_covariance}} which gives the last line. 
\end{proof}
To get the inequalities, we note that $\rho_i^2 \to 1$ and $\rho_i^2 \leq \rho_j^2$ for $i \leq j$, so we can 
 apply the ordering property of the bivariate Gaussian copula from Section \ref{app:sec_copula}, where the ordering holds uniformly over $u,u'\in(0,1)$. Note that  $\lim_{\rho \to 1} C_{\rho^2}(u,u') = \min\{u,u'\}$, so the squeeze theorem gives the convergence to the Brownian bridge covariance function as $\rho_ns \to 1$ for both $k_{\rho_n}$ and $r_n^{-1} k_n$.

We now turn our focus to the Gaussian process $\mathbb{G}_{\rho} \sim \mathcal{GP}(0,C_{\rho^2}(u,u') - uu')$ for some $\rho \in (0,1)$ as in the approximate sampling scheme for the QMP. We will show that sample paths of $\mathbb{G}_{\rho}$ are in $H^1((0,1))$ using \citet[Theorem 1]{Scheuerer2010}. This depends on properties of the partial derivatives of the kernel function, which exists and is equal to
\begin{align*}
    \frac{\partial^2}{\partial u \partial v}[C_{\rho}(u,v) - uv] = c_\rho(u,v) - 1.
\end{align*}
We then have the following lemma.
\begin{lem}
The bivariate copula density $c_{\rho}(u,v)$ satisfies
 \begin{align*}
     \int_0^1[c_\rho(u,u) - 1]\, du < \infty.
 \end{align*}
\end{lem}
\begin{proof}
   We have the following from a change of variables $u \to z_u = \Phi^{-1}(u)$:
    \begin{align*}
        \int_0^1 [c_{\rho}(u,u)-1] \, du &= \frac{1}{\sqrt{1-\rho^2}}\int_{-\infty}^{\infty}\exp\left(\frac{2\rho z_u^2 -2\rho^2z_u^2}{2(1-\rho^2)}\right)\phi(z_u)\, dz_u-1\\
        &=\frac{1}{\sqrt{2\pi(1-\rho^2)}}\int_{-\infty}^{\infty}\exp\left(\frac{ z_u^2 \left(2\rho -1 - \rho^2\right)}{2(1-\rho^2)}\right)\, dz_u-1\\
        &=\frac{1}{\sqrt{2\pi(1-\rho^2)}}\int_{-\infty}^{\infty}\exp\left(-\frac{z_u^2(1-\rho)}{2(1+\rho)}\right)\, dz_u-1 < \infty
    \end{align*}
    where the finiteness follows as $(1-\rho)/(1+\rho) > 0$. 
\end{proof}
Continuity of $c_\rho(u,u)$ for all $u \in (0,1)$ and the above lemma means that the covariance function $k_{\rho}$ satisfies the conditions of \citet[Theorem 1]{Scheuerer2010}, so sample paths of $\mathbb{G}_{\rho}$ are in $H^1((0,1)) = W^{1,2}((0,1))$ a.s.

Finally, we show the weak convergence to the Brownian motion. Marginal convergence is quite obvious as the covariance function approaches $\min\{u,u'\} - uu'$ with $n\to \infty$. However, showing tightness of the sequence of GPs requires a bit more work.
\begin{lem}
    The sequence $\mathbb{G}_{\rho_n}$ is asymptotically tight in $\ell^\infty((0,1))$.
\end{lem}
\begin{proof}
We begin by computing the standard deviation semimetric of $\mathbb{G}_{\rho_n}$:
\begin{align*}
    \E\left[\left(\mathbb{G}_{\rho_n}(u) - \mathbb{G}_{\rho_n}(v) \right)^2\right] &=     \E\left[\mathbb{G}_{\rho_n}(u)^2\right] + \E\left[\mathbb{G}_{\rho_n}(v)^2\right] - 2\E\left[\mathbb{G}_{\rho_n}(u)\, \mathbb{G}_{\rho_n}(v)\right]\\
    &=C_{\rho_n^2}(u,u) - u^2 +  C_{\rho_n^2}(v,v) - {v}^2 - 2\left(C_{\rho_n^2}(u,v) - uv \right)
\end{align*}
From \cite{Meyer2013}, we have the following  property:
\begin{align*}
    \min\{u,v\} - C_{\rho}(u,v) = \int_{\rho}^1 \phi_{2}\left(z_u, z_v; \rho\right)\, dr
\end{align*}
where
\begin{align*}
    \phi_2(z_u,z_v) &= \frac{1}{2\pi\sqrt{1-\rho^2}}\exp\left(-\frac{z_u^2 + z_v^2 - 2\rho z_u z_v}{2(1-\rho^2)}\right).
\end{align*}
For $\mathbb{G} \sim \mathcal{GP}(0,\min\{u,u'\} - uu')$ as the Brownian bridge, we have
\begin{align*}
  &\E\left[\left(\mathbb{G}(u) -\mathbb{G}(v) \right)^2\right] - \E\left[\left(\mathbb{G}_{\rho_n}(u) - \mathbb{G}_{\rho_n}(v) \right)^2\right] =\frac{1}{2\pi} \int_{\rho_n^2}^1 \frac{g_r(z_u,z_v) }{\sqrt{1- r^2}}\, dr
\end{align*}
where
\begin{align*}
   g_r(z_u,z_v) =  \exp\left(- \frac{z_u^2}{1+r}\right) +\exp\left(- \frac{z_v^2}{1+r}\right)- 2\exp\left(-\frac{z_u^2 + z_v^2 - 2r z_u z_v}{2(1-r^2)}\right).
\end{align*}
Completing the square gives
\begin{align*}
   g_r(z_u,z_v) &= \left[\exp\left(- \frac{z_u^2}{2(1+r)}\right) -\exp\left(- \frac{z_v^2}{2(1+r)}\right)\right]^2 +2\left[\exp\left(-\frac{z_u^2 + z_v^2}{2(1+r)}\right)-\exp\left(-\frac{z_u^2 + z_v^2 - 2r z_u z_v}{2(1-r^2)}\right) \right]
\end{align*}
The second term can be written as
\begin{align*}
    &\exp\left(-\frac{(z_u^2 + z_v^2)(1-r)}{2(1-r^2)}\right)-\exp\left(-\frac{z_u^2 + z_v^2 - 2r z_u z_v}{2(1-r^2)}\right) \\
    &=\exp\left(-\frac{(z_u^2 + z_v^2)(1-r)}{2(1-r^2)}\right)\left[1 - \exp\left(-\frac{r(z_u^2+ z_v^2-2z_uz_v)}{2(1-r^2)}\right)\right]\\
    &= \exp\left(-\frac{(z_u^2 + z_v^2)(1-r)}{2(1-r^2)}\right)\left[1 - \exp\left(-\frac{r(z_u - z_v)^2}{2(1-r^2)}\right)\right]
\end{align*}
Since $(z_u - z_v)^2 \geq 0$, we have that the above is non-negative, so $g_r(z_u,z_v) \geq 0$ for all $z_u,z_v \in \R$. This thus gives
\begin{align}\label{app:eq_semimetric_ineq}
\E\left[\left(\mathbb{G}_{\rho_n}(u) - \mathbb{G}_{\rho_n}(v) \right)^2\right]\leq     \E\left[\left(\mathbb{G}(u) - \mathbb{G}(v) \right)^2\right] 
\end{align}
for all $n$. 
Consider the semimetric space $(\mathcal{F} = (0,1),d)$ where $d(u,v) = |u-v|^{1/2}$.
It is clear that $\mathcal{F}$ is totally bounded under this semimetric.
Let $d_{n}$ and $d_\infty$ denote the standard deviation semimetrics of $\mathbb{G}_{\rho_n}$ and $\mathbb{G}$ respectively, which are
\begin{align*}
    d_n^2(u,v) &= \E\left[\left(\mathbb{G}_{\rho_n}(u) - \mathbb{G}_{\rho_n}(v) \right)^2\right]\\
    d_\infty^2(u,v) &= \E\left[\left(\mathbb{G}(u) - \mathbb{G}(v) \right)^2\right]=|u-v|(1-|u-v|).
\end{align*}
As $|u-v| < 1$, we have $d_\infty(u,v)\leq d(u,v)$, which combined with (\ref{app:eq_semimetric_ineq}) gives
\begin{align*}
    d_{n}(u,v) \leq d_\infty(u,v) \leq d(u,v)
\end{align*}
for all $u,v \in(0,1)$ and $n$.

Let $D(\varepsilon,d)$ denote the packing number of the space $(\mathcal{F},d)$.
\citet[Corollary 2.2.9]{vanderVaart2023} states that for $X$ as a separable Gaussian process with $d_X$ as its standard deviation semimetric, we have for every $\delta > 0$:
\begin{align*}
    \E\left[\sup_{d_X(u,v)\leq \delta}|X(u)-X(v)|\right] \leq K \int_0^{\delta} \sqrt{\log D(\varepsilon,d_X)}\, d\varepsilon
\end{align*}
for a universal constant $K$. 

In particular, as we have $d_n(u,v) \leq d(u,v)$, this implies
that $\{u,v:d(u,v) \leq \delta\}\subseteq \{u,v:d_n(u,v) \leq \delta\}$, which gives
\begin{align*}
    \E\left[\sup_{d(u,v)\leq \delta}|\mathbb{G}_{\rho_n}(u)-\mathbb{G}_{\rho_n}(v)|\right] \leq     \E\left[\sup_{d_n(u,v)\leq \delta}|\mathbb{G}_{\rho_n}(u)-\mathbb{G}_{\rho_n}(v)|\right] \leq K \int_0^{\delta} \sqrt{\log D(\varepsilon,\rho_n)}\, d\varepsilon.
\end{align*}
Furthermore, since $d_n(u,v)\leq d(u,v)$ where $d(u,v) = |u-v|^{1/2}$, the packing numbers similarly satisfy
\begin{align*}
    D(\varepsilon, d_n)\leq D(\varepsilon, d). 
\end{align*}
Under the semimetric $d$, the packing number for any $\varepsilon > 0$ satisfies
\begin{align*}
    D(\varepsilon,d) \leq \frac{C}{\varepsilon^2}
\end{align*}
for a universal constant $C$. We thus have
\begin{align*}
    \int_0^{\delta} \sqrt{\log D(\varepsilon,\rho)}\, d\varepsilon &=\int_0^\delta \sqrt{\log C + 2\log \varepsilon^{-1}} \, d\varepsilon\\
    &\leq \sqrt{\int_0^\delta(\log C + 2\log \varepsilon^{-1})\, d\varepsilon} \\
    &= \sqrt{\delta \left[\log C + 2(1 + \log \delta^{-1})\right] }
\end{align*}
where the second line follows from Jensen's inequality. The above can be made arbitrarily small by decreasing $\delta$. Finally, Markov's inequality gives
\begin{align*}
    \limsup_n \Pr\left(\sup_{d(u,v) \leq \delta}|\mathbb{G}_{\rho_n}(u)-\mathbb{G}_{\rho_n}(v)|> \varepsilon\right) &\leq \varepsilon^{-1}\limsup_n \E\left[\sup_{d(u,v) \leq \delta}|\mathbb{G}_{\rho_n}(u)-\mathbb{G}_{\rho_n}(v)|\right]\\
    &\leq \varepsilon^{-1}\sqrt{\delta \left[\log C + 2(1 + \log \delta^{-1})\right] }.
\end{align*}
We can make the right hand side less than any $\eta>0$ by sufficiently decreasing $\delta$, so $\mathbb{G}_{\rho_n}$ is asymptotically uniformly $d$-equicontinuous in probability. 

For each $u\in(0,1)$, uniform tightness of the sequence $\mathbb{G}_{\rho_n}(u) \sim \mathcal{N}(0,C_{\rho_n^2}(u,u) - u^2)$  can be verified with
\begin{align*}
    \sup_n \, \Pr\left(|\mathbb{G}_{\rho_n}(u)| > M \right) \leq \sup_n\frac{\E\left[\mathbb{G}_{\rho_n}(u)^2\right]}{M^2} \leq \frac{u(1-u)}{M^2}\cdot
\end{align*}
For any $\varepsilon > 0$, we can choose $M^2 > u(1-u)/\varepsilon$ which gives $ \sup_n \, \Pr\left(|\mathbb{G}_{\rho_n}(u)| > M \right) < \varepsilon$. From \citet[Theorem 1.5.7]{vanderVaart2023}, the sequence $\mathbb{G}_{\rho_n}$ is asymptotically tight in $\ell^\infty((0,1))$. 
\end{proof}
Finally, the marginals  of $\mathbb{G}_{\rho_n}$ are simply zero-mean Gaussian vectors with covariance matrix with entries $C_{\rho_n^2}(u_i,u_j) - u_i u_j$. Each entry converges pointwise from below to $\min\{u_i,u_j \} - u_i u_j$ as $\rho_n \to 1$, so from Lévy's continuity theorem, the marginals converge to a zero-mean Gaussian vector with the appropriate covariance matrix. From \citet[Theorem 1.5.4]{vanderVaart2023}, $\mathbb{G}_{\rho_n}$ converges weakly to $\mathbb{G}$ in $\ell^\infty((0,1))$.
    
\subsection{Theorem \ref{th:qr_sobolev}}
The additional required assumptions of {Theorem \ref{th:qr_sobolev}} are as follows, which is analogous to {Assumptions \ref{as:L2}} and {\ref{as:weak_deriv}} from the unconditional version.
\begin{as}[Bounded in $L^2$]\label{as:beta_L2}
    For each $j \in \{1,\ldots,p\}$, $\beta_{nj}$ satisfies $\|\beta_{nj}\|_2 < \infty$.
\end{as}
\begin{as}[Weak derivatives bounded in $L^2$]\label{as:beta_weak_deriv} 
    For each $j \in \{1,\ldots,p\}$, $\beta_{nj}$ is weakly differentiable with weak derivative $\beta'_{nj}$ which satisfies $\|\beta'_{nj}\|_2 < \infty$, so $\|\beta_{nj}\|_{1,2} <\infty$.
\end{as}
The proof is an extension of {Theorem \ref{th:sobolev_mart}}, with the additional complication of random covariates $X_{n+1:\infty}$ arising from the Bayesian bootstrap. Let $B' = H^1((0,1))$, and consider a single component $j  \in \{1\ldots,p\}$. We then have the update
\begin{align*}
    \beta_{N+1,j}(u) = \beta_{Nj}(u) + \alpha_{N+1}\left[u - H_{\rho_{N+1}}(u,V_{N+1})\right] X_{N+1,j}
\end{align*}
for each $u \in (0,1)$  and $N \geq n$, where $V_{N+1} \iid \mathcal{U}(0,1)$ and $X_{N+1}\mid X_{1:N} \sim \frac{1}{N}\sum_{i = 1}^{N}\delta_{X_i}$. Our filtration now consists of $\mathcal{F}_N = \sigma\left(X_{n+1}, V_{n+1}, \ldots, X_N, V_N\right)$
 for $N \geq n+1$, with $\mathcal{F}_{n} = \{\emptyset, \Omega\}$ again.

The above is again a pointwise martingale, as we have $\int_0^1 H_\rho(u,v)\, dv = u$ which gives
\begin{align*}
    \E\left[\beta_{N+1,j}(u) \mid \mathcal{F}_{N}\right]&=\E\left[\E\left[\beta_{N+1,j}(u)\mid \mathcal{F}_N, X_{N+1}\right] \mid \mathcal{F}_{N}\right]\\
    &=\E\left[\beta_{Nj}(u)\mid \mathcal{F}_{N}\right]\\
    &= \beta_{Nj}(u).
\end{align*}
Note that the (conditional) distribution of $X_{N+1}$ does not affect the martingale. The argument using continuous bounded functionals in {Lemma \ref{app:lem_bochnerh1}} can be repeated here to show that
$\{\beta_{Nj}\}_{N \geq n+1}$ is a $B'$-valued martingale if $\beta_{Nj} \in B'$. We can then upper bound the $L^2(B')$ norm as in {Theorem \ref{th:sobolev_mart}} with
\begin{align*}
    \|\beta_{Nj}\|_{L^2(B')} \leq  \|\beta_{nj}\|_{1,2} + \sup_{k \in \{1,\ldots,n\}} X_{kj}^2\sum_{i = n+1}^N \frac{\alpha_i^2}{1-\rho_i^2},
\end{align*}
where we have used the fact that $X_{n+1:N}$ will be repeats of $X_{1:n}$ and
$\sup_{k \in \{1,\ldots,n\}} X_{kj}^2$ is finite as we only have finitely many (i.e. $n$) covariate observations.
Under {Assumptions \ref{as:alpha}}, {\ref{as:bandwidth}}, {\ref{as:beta_L2}} and {\ref{as:beta_weak_deriv}}, we thus have $\beta_{Nj} \in B'$ for each $N$ and $\sup_{N \geq n} \|\beta_{Nj}\|_{L^2(B')} < \infty$, so we can apply {Theorem \ref{app:thm_hilbert_mart}}.

We can repeat the above for all components $j \in (1,\ldots,p)$, and as $p$ is finite, the union of the null sets on which convergence does not occur for each component has measure 0, so the vector $\beta_N$ converges to $\beta_\infty$ component-wise a.s. 

We now describe the space of vector functions $\beta_N$ and $\beta_\infty$, which we will need for a later proof. Consider the finite product of Banach spaces 
\begin{align*}
    H^1((0,1))^p := H^1((0,1)) \times \ldots \times H^1((0,1)).
\end{align*}
For a vector $f \in H^1((0,1))^p$, we define the norm of this Banach space as 
\begin{align*}
    \|f\|_{1,2,p} := \sum_{j = 1}^p \|f_j\|_{1,2}.
\end{align*}
It is clear that $\beta_N, \beta_\infty \in H^1((0,1))^p$, and as each component converges a.s., we have
\begin{align*}
    \|\beta_N - \beta_\infty\|_{1,2,p} \to 0 \quad \text{a.s.}
\end{align*}

\subsection{Proposition \ref{prop:condit_quant}}
For an arbitrary $x \in \mathcal{X}$, consider the mapping $h_x: H^1((0,1))^p \to H^1((0,1))$ defined by 
\begin{align*}
h_x(\beta)= \sum_{j = 1}^p \beta_j x_j
\end{align*}
where $x_j\in \R$ is the $j$-th component of $x$. This mapping can be shown to be continuous as follows. Consider a sequence $\beta_N \to \beta_\infty$ in $H^1((0,1))^p$, then we have
\begin{align*}
    \|h_x(\beta_N) - h_x(\beta_\infty)\|_{1,2}      &\leq \sum_{j = 1}^p \|(\beta_{Nj}- \beta_{\infty j}) \, x_j \|_{1,2}\\
    &=\sum_{j = 1}^p |x_j|\, \|\beta_{Nj}- \beta_{\infty j} \, \|_{1,2}\\
    &\leq \sup_{j}|x_j|\, \|\beta_N - \beta_\infty\|_{1,2,p} \to 0,
\end{align*}
where we have applied the triangle inequality in the first step. For $Q_N(\cdot \mid x) := h_x(\beta_N)$ and $Q_\infty(\cdot \mid x) := h_x(\beta_\infty)$ from {Theorem \ref{th:qr_sobolev}}, the continuous mapping theorem then gives
\begin{align*}
    d_{1,2}\left(Q_N(\cdot \mid x), Q_\infty(\cdot \mid x)\right) \to 0 \quad \text{a.s.}
\end{align*}
From {Theorem \ref{app:thm_sobolev}}, we have that  $Q_N^{\dagger}(\cdot \mid x)$ and $Q_\infty(\cdot \mid x)$ are in $H^1((0,1))$. We can thus apply {Proposition \ref{app:prop_L2}} to give 
\begin{align*}
        d_{2}(Q^{\dagger}_N(\cdot \mid x), Q^{\dagger}_\infty(\cdot \mid x)) \to 0 \quad \text{a.s.}
\end{align*}

\subsection{Proposition \ref{prop:linear}}
    For each $x \in \mathcal{X}$,  we can once again apply {Lemma \ref{app:lem_equi}} with $h(x) = x$, which gives
\begin{align*}
\int_0^1 Q^{\dagger}_\infty(u \mid x)\, du&= \int_0^1 Q_\infty(u \mid x)\, du\\
&=\int_0^1 \beta_\infty(u)^T x \, du \quad \textnormal{a.s.}
\end{align*}
Linearity of expectation then gives $\E_\infty\left[Y \mid x\right] = \left[\int_0^1 \beta_\infty(u) \, du \right]^T x$, which is a (random) linear function in $x$ a.s.

\subsection{Theorem \ref{th:reg_gp}}

We follow the same strategy as the proof for {Theorem \ref{th:gp}}. We will use {Theorem \ref{th::bracketing_clt}} to verify asymptotic tightness, and then we establish marginal convergence. To reduce clutter, we will suppress the conditioning on the weights in the notation. Define
\begin{equation*}
    C_{X} = \max_{\substack{i \in \{1,\ldots, n\}\\j \in \{1, \ldots, p\}}}|X_{ij}|.
\end{equation*}

\begin{thm}\label{th:tight}
The sequence of functions $\sqrt{N}S_N$ is  asymptotically tight in $\ell^{\infty}(\mathcal{F})$ with probability 1.
\end{thm}
\begin{proof}
    We verify the assumptions in {Theorem \ref{th::bracketing_clt}}. Our semimetric space is $(\mathcal{F} = (0,1)\times \{1,\ldots, p\}, d)$, where $d((u_{1},j_{1}),(u_{2},j_{2})) = |u_{1}-u_{2}|^{1/2} + \mathbbm{1}(j_{1} \neq j_{2})$. We have used the discrete metric on $\{1,\ldots, p\}$ and then specified the sum of the two semimetrics to define the semimetric product space. Clearly, this semimetric space is totally bounded.

First define $Z_{Ni}(u, j) = \sqrt{N}\alpha_{i}(H_{\rho_{i}}(u,V_{i})-u)X_{ij}$ for all $N$ and $i \geq N$. Note that this definition differs in nature to the non-regression case (where the $-u$ term can be omitted) because the randomness in the covariates must be accounted for. A trivial envelope function for $Z_{Ni}$ is $F_{Ni} = 2\sqrt{N}\alpha_{i}C_{X}$. We need to verify the Lindeberg condition
\begin{equation*}
    \sum_{i=N}^{\infty} F_{Ni}\mathbbm{1}\{F_{Ni} > \eta\} \rightarrow 0
\end{equation*}
for every $\eta > 0$. Since $F_{Ni} < 2cC_{X}/\sqrt{N}$, we will have $F_{Ni} < \eta$ for all sufficiently large $N$. Thus, the Lindeberg condition holds.

Next we need 
\begin{equation*}
    \sup_{d((u_{1},j_{1}),(u_{2},j_{2})) < \delta_{N}} \sum_{i=N}^{\infty} \E\left\{Z_{Ni}(u_{1})-Z_{Ni}(u_{2})\right\}^{2} \rightarrow 0
\end{equation*}
for every $\delta_{N} \downarrow 0$. For all sufficiently large $N$, we must have $\delta_{N} < 1$, in which case
\begin{equation*}
    d((u_{1},j_{1}),(u_{2},j_{2})) < \delta_{N} \implies j_{1} = j_{2} \quad \text{and} \quad |u_{1} - u_{2}|^{1/2} < \delta_{N}.
\end{equation*}
Let $u_{1} > u_{2}$ with $|u_{1} - u_{2}|^{1/2}  < \delta_{N}$. We have
\begin{align*}
     \E\left\{Z_{Ni}(u_{1}, j)-Z_{Ni}(u_{2},j)\right\}^{2}&\leq  2\E\left\{\sqrt{N}\alpha_{i}X_{ij}H_{\rho_{i}}(u_{1},V_{i})-\sqrt{N}\alpha_{i}X_{ij}H_{\rho_{i}}(u_{2},V_{i})\right\}^{2} \\
     &+ 2\E\left\{\sqrt{N}\alpha_{i}X_{ij}u_{1}-\sqrt{N}\alpha_{i}X_{ij}u_{2}\right\}^{2}\\
     &\leq  2N\alpha_{i}^{2}C_{X}^{2}\E\left\{H_{\rho_{i}}(u_{1},V_{i})-H_{\rho_{i}}(u_{2},V_{i})\right\}^{2} \\
     &+ 2N\alpha_{i}^{2}C_{X}^{2}\delta_{N}^{2}
\end{align*}

Note that $H_{\rho_{i}}(u,V_{i})$ is non-decreasing in $u$ with $H_{\rho_{i}}(0,V_{i}) = 0$ and $H_{\rho_{i}}(1,V_{i}) = 1$. So
\begin{align*}
    \E\left\{H_{\rho_{i}}(u_{1},V_{i})-H_{\rho_{i}}(u_{2},V_{i})\right\}^{2} &\leq \E\left\{H_{\rho_{i}}(u_{1},V_{i})-H_{\rho_{i}}(u_{2},V_{i})\right\} \\
    &= (u_{1}-u_{2})\\
    &< \delta_{N}^{2}.
\end{align*}
Thus, for all sufficiently large $N$ such that $\delta_{N} < 1$, we have
\begin{equation*}
    \sup_{d((u_{1},j_{1}),(u_{2},j_{2})) <\delta_{N}} \sum_{i=N}^{\infty}\E\left\{Z_{Ni}(u_{1}, j_{1})-Z_{Ni}(u_{2}, j_{2})\right\}^{2} < 4N\delta_{N}^{2} C_{X}^{2}\sum_{i=N}^{\infty} \alpha_{i}^{2}.
\end{equation*}
Since $\limsup_{N \rightarrow \infty} N \sum_{i=N}^{\infty} \alpha_{i}^{2} \leq c^{2}$, the right-hand side of the above display tends to zero for any $\delta_{N} \downarrow 0$.

Finally, we need to verify the bracketing entropy integral condition. Fix $j \in \{1,\ldots, p\}$ for the time being. For all sufficiently large $n$, we will have $N \sum_{i=N}^{\infty} \alpha_{i}^{2} \leq 2c^{2}$. Given $\varepsilon > 0$, choose a partition $0 = u_{0} < u_{1} < \ldots < u_{M} = 1$ such that $u_{k} - u_{k-1} < \varepsilon^{2}/(8C_{X}^{2}c^{2})$ for every $k$. The number of points in the partition can be chosen to be smaller than a constant times $1/\varepsilon^{2}$. Then
\begin{align*}
    \E \sup_{u_{k-1} \leq s,t < u_{k}} |Z_{Ni}(s,j) - Z_{Ni}(t,j)|^{2} &< 2N\alpha^{2}_{i}C_{X}^{2}\E\left\{H_{\rho_{i}}(u_{k},V_{i})-H_{\rho_{i}}(u_{k-1},V_{i})\right\}^{2}\\
    &+ 2N\alpha^{2}_{i}C_{X}^{2}\E \sup_{u_{k-1} \leq s,t < u_{k}} |s-t|^{2} \\
    &< 2N\alpha^{2}_{i}C_{X}^{2} \E\left\{H_{\rho_{i}}(u_{k},V_{i})-H_{\rho_{i}}(u_{k-1},V_{i})\right\}\\
    &+ 2N\alpha^{2}_{i}C_{X}^{2}(u_{k} - u_{k-1}) \\
    &< \frac{N \alpha_{i}^{2}\varepsilon^{2}}{2c^{2}}.
\end{align*}
We deduce that for all sufficiently large $N$,
\begin{equation*}
    \sum_{i=N}^{\infty}\E \sup_{u_{j-1} \leq s,t < u_{j}} |Z_{Ni}(s) - Z_{Ni}(t)|^{2} < \varepsilon^{2}.
\end{equation*}
In other words, $\mathfrak{N}_{[]}(\varepsilon, \mathcal{F}, L_{2,n}) \lesssim \frac{1}{\varepsilon^{2}}$, which verifies the entropy condition. 

Thus, {Theorem \ref{th::bracketing_clt}} implies that
\begin{equation*}
    -\sum_{i=N}^{\infty} Z_{Ni}(u,j) = \sqrt{N}\sum_{i=N}^{\infty}\alpha_{i}(u - H_{\rho_{i}}(u,V_{i}))X_{ij}
\end{equation*}
is asymptotically tight in $\ell^{\infty}(\mathcal{F})$.
\end{proof}

\begin{prop}
    If $(u_{1},\ldots,u_{d}) \in (0,1)^{d}$ and $(j_{1},\ldots, j_{d}) \in \{1,\ldots,p\}^{d}$, then
    \begin{equation*}
        \sqrt{N}[S_{N}(u_{1},j_{1}), \ldots, S_{N}(u_{d},j_{d})]^{T} \mid w_{1:n}\xrightarrow[]{d} \mathcal{N}(0,a^{2}\Sigma)
    \end{equation*}
    as $N \rightarrow \infty$, where $    \Sigma_{l,m} = \left[\sum_{k=1}^{n}w_{k}X_{kj_l}X_{kj_m}\right](\min\{u_l,u_m\} - u_l u_m)$.
\end{prop}

\begin{proof}
    Fix an arbitrary vector $\mathbf{t} = (t_{1},\ldots, t_{d}) \in \mathbb{R}^{d}$ and we study the convergence of
\begin{align*}
     \sqrt{N}\sum_{l=1}^d \, t_l \, S_N(u_l,j_l).
\end{align*}
Consider
\begin{align*}
    Z_{Ni}[\mathbf{t}] = \sqrt{N}\,\alpha_{i} \sum_{l=1}^d \, t_l \, (H_{\rho_{i}}(u_l,V_{i})-u_{l})X_{ij}
\end{align*}
which has expectation $0$. We have that
\begin{align*}
    |Z_{Ni}[\mathbf{t}]| \leq 2aC_{X}\frac{\sqrt{N}}{i}\sum_{l=1}^d  |t_j |,
\end{align*}
so for every $\varepsilon > 0$, $|Z_{Ni}[\mathbf{t}]| <\varepsilon$ eventually, which verifies the Lindeberg condition for the Lindeberg-Feller CLT. For the limiting variance, we first have
\begin{align*}
   \E[Z_{Ni}[\mathbf{t}]^2] &= a^2 \frac{N}{i^{2}} \sum_{l=1}^d \sum_{m = 1}^d t_l t_m \left[\sum_{k=1}^{n}w_{k}X_{kj_l}X_{kj_m}\right]\\
    &\quad \left(\E\left[H_{\rho_{i}}(u_j,V_{i})\, H_{\rho_{i}}(u_k,V_{i})\right] - u_{l}u_{m}\right)
\end{align*}
If we take $N/m_{N} = o(1)$, then
\begin{align*}
    \sum_{i=N}^{N+m_{N}}\E[Z_{Ni}[\mathbf{t}^2] \to \sigma^2 =  a^2  \sum_{l=1}^d \sum_{m = 1}^d t_l t_m \left[\sum_{k=1}^{n}w_{k}X_{kj_l}X_{kj_m}\right]\left(\min\{u_{l},u_{m}\} - u_{l}u_{m}\right)
\end{align*}
via similar computations to the non-regression case. So we can apply the Lindeberg-Feller CLT to the sequence of sums up to $m_{N}$. We check that the tail sums from $m_{N}+1$ onwards are asymptotically negligible: 
\begin{align*}
    \sum_{i = m_{N}+1}^{\infty} \text{Var}[Z_{Ni}[\mathbf{t}]] &= N \sum_{i =  m_{N} + 1}^\infty \alpha_{i}^2 \sum_{l=1}^d \sum_{m = 1}^d t_l t_m \left[\sum_{k=1}^{n}w_{k}X_{kj_l}X_{kj_m}\right] \left(C_{\rho_{i}^{2}}(u_l,u_m) - u_l u_m\right)\\
    & \leq \left|\sum_{l=1}^d \sum_{m = 1}^d t_l t_m C_{X}^{2}\right| N \sum_{i = m_{N} + 1}^\infty \alpha_{i}^2 \to 0.
\end{align*}
Finally, the Cram\'{e}r-Wold device gives us
\begin{align*}
    \sqrt{N}[S_{N}(u_{1},j_{1}), \ldots, S_{N}(u_{d},j_{d})]^{T} \overset{d}{\to} \mathcal{N}(\mathbf{0}, a^2 \Sigma)
\end{align*}
where $    \Sigma_{l,m} = \left[\sum_{k=1}^{n}w_{k}X_{kj_l}X_{kj_m}\right](\min\{u_l,u_m\} - u_l u_m)$.
\end{proof}

\subsection{Proposition \ref{prop:QR_coverage}}
We require the following assumption on the covariance matrix of the covariates.
\begin{as}[Covariance matrix of covariates]\label{app:as_covariates_cov} 
The covariance matrix of the covariate distribution $\Sigma_x = \int_{\mathcal{X}} x x^T \, dP^*(x)$ is positive definite, and all elements are finite. 
\end{as}

We follow a similar approach to {Proposition \ref{prop:mu}}. Let $Y_i,X_i \iid P^*(y,x)$, and $\beta_n$ is computed by {Algorithm \ref{alg:reg_fit}}.
Following the same argument as in {Theorem \ref{th:consistency}} for each component of $\beta_n$, one can see that $\beta_n \in H^1((0,1))^p$ for each $n$, thus satisfying {Assumptions \ref{as:beta_L2}} and {\ref{as:beta_weak_deriv}}. This allows us to apply {Theorem \ref{th:qr_sobolev}} giving the existence of $\beta_{n\infty}$. To construct the probability space, we can again consider a single sequence $V_1,V_2,\ldots$ of uniform r.v.s.
For each $n$, consider the Bayesian bootstrap starting with $X_{1:n}$, i.e. we have
\begin{align*}
    w_{1:n}&\sim \text{Dir}(1,\ldots,1), \quad 
    X_{nN} \mid w_{1:n}, X_{1:n} \iid \sum_{i = 1}^n w_i \delta_{X_i}
\end{align*}
for $N \geq n+1$, where the additional subscript $n$ on $X$ indicates how many `real' observations we start predictive resampling from. 
We can then define for each $n \geq 1$
\begin{align*}
    \beta_{n\infty}(u) = \beta_n(u) + \sum_{i = n+1}^\infty \alpha_i \left(u - H_{\rho_i}(u,V_i)\right)\, X_{ni}.
\end{align*}

We now consider the posterior distribution of the mean functional $\bar{\beta}_{n\infty}$ where
\begin{align*}
\bar{\beta}_{n\infty}, = \int_0^1 \beta_{n\infty}(u) \, du.
\end{align*}
Fubini's theorem gives 
\begin{align*}
\E[\bar{\beta}_{n\infty}\mid Y_{1:n},X_{1:n}] = \bar{\beta}_n.
\end{align*}
The posterior covariance matrix is thus
\begin{align*}
    \E[b_n b_n^T \mid Y_{1:n},X_{1:n}] = \E\left[\left(\sum_{i = n+1}^\infty \alpha_i Z_i\, X_{ni}\right)\left(\sum_{i = n+1}^\infty \alpha_i Z_i\, X_{ni}\right)^T \Big| Y_{1:n},X_{1:n}\right]
\end{align*}
where we write $b_n = \bar{\beta}_{n\infty}- \bar{\beta}_n$ and $Z_i = \int_0^1\left(u - H_{\rho_{i}}(u,V_i)\right)\, du$ for shorthand.
Let us first condition on $w_{1:n}$, which gives
\begin{align*}
 \E[b_n b_n^T\mid w_{1:n},Y_{1:n},X_{1:n}] = 
\sum_{i = n+1}^\infty \sum_{j = n+1}^\infty \alpha_i\alpha_j \E\left[ Z_iZ_j\, X_{ni}X_{nj}^T \mid w_{1:n}, Y_{1:n},X_{1:n}\right].
\end{align*}
    
As $Z_i$ and $Z_j$ are independent for $i \neq j$, and the covariates are independent from the uniform r.v.s, the cross-terms are all 0, so the above simplifies to
\begin{align*}
    \sum_{i = n+1}^\infty \alpha_i^2\,  \E[Z_i^2] \,\E[X_{ni} X_{ni}^T \mid w_{1:n},  Y_{1:n},X_{1:n}].
\end{align*}
As before, we have
\begin{align*}
    \E[Z_i^2] = \int_0^1\int_0^1 \, C_{\rho_i^2}(u,v)\, du \,dv - \frac{1}{4}
\end{align*}
and now we have the additional term
\begin{align*}
   \E[X_{ni} X_{ni}^T \mid w_{1:n},X_{1:n}] = \sum_{i = 1}^n w_i X_{i}X_i^T.
\end{align*}
This gives
\begin{align*}
     \E[b_nb_n^T \mid w_{1:n}, Y_{1:n},X_{1:n}] = \left[\sum_{i = 1}^n w_i X_i X_i^T\right]\sum_{i = n+1}^\infty \alpha_i^2 \, \E[Z_i^2].
\end{align*}
The tower property then gives
\begin{align*}
    \E[b_nb_n^T\mid Y_{1:n},X_{1:n}] = \left[\frac{1}{n}\sum_{i = 1}^n X_iX_i^T\right]\sum_{i = n+1}^\infty \alpha_i^2 \, \E[Z_i^2].
\end{align*}
We can then scale this by ${n}$ and take the limit, giving us
\begin{align*}
    \lim_{n\to \infty}{n}\,\E[b_nb_n^T\mid Y_{1:n},X_{1:n} ] 
    &= \lim_{n\to \infty}\left[\frac{1}{n}\sum_{i = 1}^n X_iX_i^T\right] \lim_{n \to \infty} \left[n\sum_{i = n+1}^\infty \alpha_i^2 \, \E[Z_i^2]\right]\\    
    &= \frac{a^2}{12}\Sigma_x \quad P^*\text{-a.s.}
\end{align*}

\section{Additional results}

\subsection{Asymptotic distribution of mean functional}\label{app:sec_munormal}
For the mean functional, we only computed its posterior mean and asymptotic variance in {Proposition \ref{prop:mu}}. We can actually extend this in the unconditional case and quantify the asymptotic distribution of
$\mu_{n\infty}- \mu_n$ due to it being a sum of independent terms. This could potentially lead the way to future Bernstein-von Mises for functionals of the QMP. However, as mentioned in the main paper, quantifying the distribution of $\mu_n$ is more challenging.
\begin{prop}\label{prop:mu_normal}
   Let $\mu_n = \int Q^{\dagger}_n(u) \, du$ for $\{Q^{\dagger}_n\}_{n \geq 1}$ from {Algorithm \ref{alg:fit}},  and $\mu_{n\infty} = \int_0^1 Q_{n\infty}^{\dagger}(u)\, du$ where $Q_{n\infty}^{\dagger}$ arises from {Algorithm \ref{alg:QMP}} starting from $Q^{\dagger}_n$. Under {Assumptions \ref{as:alpha}} and {\ref{as:bandwidth}}, we have
 \begin{align*}
   \sqrt{n}\left(\mu_{n\infty} - \mu_n\right)\overset{d}{\to}\mathcal{N}(0, a^2/12)
\end{align*}
\end{prop}
\begin{proof}
We will extend the proof of {Proposition \ref{prop:mu}} by applying {Theorem \ref{th:gp}}, although we highlight that one can also prove the above using the standard Lindeberg-Feller CLT for scalar r.v.s.
In order to apply {Theorem \ref{th:gp}}, we will leverage the specific construction for $Q_{n\infty}^{\dagger}$ as in the proof of {Theorem \ref{th:consistency}}. Let us define
 \begin{align*}
     S_N(u) = \sum_{i = N+1}^\infty \alpha_i \left[u - H_{\rho_i}\left(u,V_i\right)\right]
 \end{align*}
 for $n \geq 1$, where $V_i \iid \mathcal{U}(0,1)$. We then define
 \begin{align*}
     Q_{n\infty} = Q_n^{\dagger} + S_n,
 \end{align*}
 for $n \geq 1$. For each $n$, $Q_{n\infty}$ has the same distribution as that induced by {Algorithm \ref{alg:QMP}}, but is not independent across $n$. However, this is inconsequential as we are studying a weak limit. This gives a corresponding sequence of random means:
 \begin{align*}
     \mu_{n\infty} = \int_0^1 Q_{n\infty}(u) \, du.
 \end{align*}
 We then clearly have
 \begin{align*}
    \sqrt{n}(\mu_{n\infty} - \mu_n) &=\int_0^1 \sqrt{n}\, S_n(u)\, du.
\end{align*}
 Since the weak limit of $\sqrt{N}S_N$ is independent of the initial estimate $Q_n^{\dagger}$, we can apply {Theorem \ref{th:gp}} directly. This gives $\sqrt{n} S_n \to \mathbb{G}_a$ weakly in $\ell^\infty((0,1))$, where $\mathbb{G}_a$ is a zero-mean GP with covariance function
$\E\left[\mathbb{G}_a(u)\, \mathbb{G}_a(u^{\prime})\right] = a^2 (\min\{u,u^{\prime}\} - uu^{\prime})$.

Consider now the integral operator $h: \ell^\infty((0,1)) \to \R$ where $h(f) = \int_0^1 \, f(u) \, du$. It is not too hard to see that this is a continuous function, as for any sequence $f_n \to f_\infty$ for $f_n,f_\infty \in \ell^\infty((0,1))$, we have
\begin{align*}
    |h(f_n) - h(f_\infty)| \leq \int_0^1 |f(u) - f_\infty(u)|\, du \leq \|f_n - f_\infty\|_{\infty} \to 0.
\end{align*}
The continuous mapping theorem then gives us
\begin{align*}
\sqrt{n}(\mu_{n\infty} - \mu_n) \overset{d}{\to} h(\mathbb{G}_a).
\end{align*}
We now show that $h(\mathbb{G}_a) \sim \mathcal{N}(0,a^2/12)$. It is clear that $h(\mathbb{G}_a)$ has mean zero, and its variance is
\begin{align*}
   \E\left[h(\mathbb{G}_a)^2 \right] 
   &=
    \E\left[\int_0^1 \int_0^1\mathbb{G}_a(u) \, \mathbb{G}_a(u')\, du \, du'\right]
   \\&=\int_0^1 \int_0^1 \E\left[\mathbb{G}_a(u) \, \mathbb{G}_a(u')\right]\, du \, du'\\
  &= a^2\int_0^1 \int_0^1\left[\min\{u,u'\}- uu'\right] du \, du' \\
  &= \frac{a^2}{12}\cdot
\end{align*}
The normality of $h(\mathbb{G}_a)$ then follows from an approximating Riemann sum argument (e.g. \cite[Chapter 8.3]{Hassler2016}) as sample paths of the Brownian motion are continuous a.s. 
\end{proof}

\subsection{Frequentist consistency for quantile regression}\label{app:sec:QR_consistency}
In this section, we outline a posterior consistency result for the QMP for quantile regression with $\rho = 1$. This setting lends itself more easily to a consistent initial estimate. Consider the QMP with the updates
\begin{align}\label{eq:beta_consist1}
    \beta_{N+1}(u) &=\beta_N(u) + \alpha_{N+1} \left[u - G_{N+1}(\beta_N,Y_{N+1},X_{N+1})\right]\, X_{N+1}\\
        G_{N}(\beta,Y,X)&= \begin{cases}
         \mathbbm{1}\left(Y \leq \beta(u)^TX\right)\quad &\text{for } N \leq n\\
            \mathbbm{1}\left(Y \leq Q^{\dagger}(u \mid X)\right)\quad &\text{for } N \geq n+1
       \end{cases}\label{eq:beta_consist2}
\end{align}
where $Q^{\dagger}(u \mid X)$ is the increasing rearrangement of $\beta(u)^TX$.
The difference between the update for the initial estimate and predictive resampling is subtle but important for both consistency and the martingale.
To derive the latter form, note that
\begin{align*}
   \lim_{\rho \to 1} H_\rho(u,V_{N+1}) = \mathbbm{1}\left(V_{N+1}\leq u\right) = \mathbbm{1}\left(Y_{N+1} \leq Q_{N}^{\dagger}(u \mid X_{N+1})\right)  
\end{align*}
which is obtained by applying the proper quantile function $Q_{N}^{\dagger}(\cdot \mid X_{N+1})$ to both sides of the inequality. The martingale under predictive resampling is thus preserved in this case. 
For the `real data' update however, we opt to use the standard stochastic approximation estimate of $\beta_n$, as it is non-trivial to derive a copula-smoothed version of the above initial estimate (i.e. an equivalent version of (\ref{eq:rearr_quantile_copula}) for $\beta_n^{\dagger}(u)$). Consider now the following assumptions on the data generating distribution.
\begin{as}[Covariate distribution]\label{ass:covariate}
 $P^*(x)$ has compact support and the covariance matrix $\Sigma_x = \int_{\mathcal{X}} x x^T \, dP^*(x)$ is positive definite. 
\end{as}
\begin{as}[Linear quantiles]\label{ass:linear}
 There exists some true function $\beta^*(u)$ such that the quantile function corresponding to $P^*(\cdot \mid x)$ takes the form $Q^*(u \mid x)= \beta^*(u)^Tx$ for all $x$ in the support of $P^*(x)$.
\end{as}
\begin{as}[Lipschitz continuity]\label{ass:reg_lipschitz}
There exists a finite $L$ such that
\begin{align*}
    \sup_{j \in \{1,\ldots,p\}} |\beta^*_j(u) - \beta^*_{j}(u^{\prime})| \leq L|u - u^{\prime}|
\end{align*}
where $\beta^*_j(u)$ is the $j$-th component of the vector $\beta^*(u)$. Assume that the initial vector function $\beta_0$ also satisfies the above Lipschitz condition.
\end{as}
 We now define the norms to study the conditional quantile function
\begin{align*}
    {d}^2_{2,x}(Q^*(\cdot \mid x),Q(\cdot \mid x)) = \int_{\mathcal{X}} d^2_2(Q^*(\cdot \mid x),Q(\cdot \mid x))\, dP^*(x)
\end{align*}
which is the covariate average $L^2$ distance between the conditional quantiles, and also
\begin{align*}
    {d}^{2}_{2,p}(\beta^*,\beta) = \int_0^1 \left(\beta^*(u) - \beta(u)\right)^T\left(\beta^*(u) - \beta(u)\right)\, du.
\end{align*}
We have a standard result from stochastic approximation arguments:
\begin{prop}
         Under {Assumptions \ref{ass:covariate}}, {\ref{ass:linear}} and {\ref{ass:reg_lipschitz}}, we have that
             ${d}_{2,p}(\beta^*,\beta_n) \to 0$ a.s.$[P^*]$
          as $n \to \infty$ under (\ref{eq:beta_consist1}) and (\ref{eq:beta_consist2}) with the $N\leq n$ form of $G_{N}$.
\end{prop}
\begin{proof}
The $L^2$ distance can be expanded recursively:
\begin{align*}
    d^2_{2,p}\left(\beta^*, \beta_{n+1}\right) &= d^2_{2,p}\left(\beta^*, \beta_{n}\right) - 2 \alpha_{n+1} \int_0^1  \left(\beta^*(u) - \beta_n(u)\right)^T X_{n+1}\left(u - \mathbbm{1}\left(Y_{n+1} \leq \beta_n(u)^T X_{n+1}\right)\,\right) \, du\\
    &+\alpha_{n+1}^2 \int_0^1 \left(u - \mathbbm{1}\left(Y_{n+1} \leq \beta_n(u)^T X_{n+1}\right)\right)  du \, X_{n+1}^T X_{n+1} \
\end{align*}
Taking the conditional expectation of the above given $\mathcal{F}_n = \sigma(Z_{1},\ldots,Z_n)$ for $Z_i = (Y_i,X_i)$ gives
\begin{align*}
  \E\left[d^2_{2,p}\left(\beta^*, \beta_{n+1}\right) \mid \mathcal{F}_n \right]  &\leq d^2_{2,p}\left(\beta^*, \beta_{n}\right) -  2 \alpha_{n+1}  g(\beta^*,\beta_n)+ \alpha_{n+1}^2 V^2
\end{align*}
where $V^2 = \E[X^T X]$ for $X \sim P^*(x)$ which is finite by {{Assumption \ref{ass:covariate}}}, and 
\begin{align*}
    g(\beta^*,\beta) &=\int_{\mathcal{X}} \, \int_0^1  \left(\beta^*(u) - \beta(u)\right)^T x\left(u - P^*\left(\beta(u)^T x \mid x\right)\right) \, du\, dP^*(x)\\
    &= \int_{\mathcal{X}}\int_0^1 \left(Q^*(u \mid x) -Q(u \mid x)\right)\left(u - P^*\left(Q(u \mid x)\mid x\right)\right)\,du \, dP^*(x).
\end{align*}
For each value of $(u,x)$, we have that 
\begin{align*}
    \left(Q^*(u \mid x) -Q(u \mid x)\right)\left(u - P^*\left(Q(u \mid x) \mid x\right)\right) \geq 0.
\end{align*}
which follows as $P^*(\cdot  \mid x)$ is monotonic, so we have
\begin{align*}
     \left(Q^*(u \mid x) -Q(u \mid x)\right)\geq 0 \implies \left(u - P^*\left(Q(u \mid x) \mid x\right)\right) \geq 0.
\end{align*}
We thus have $g(\beta^*, \beta) \geq 0$  and $V^2\sum_{n =1}^\infty{\alpha_{n+1}^2} < \infty$ which gives us the almost supermartingale from {Theorem \ref{app:thm_asmart}}. We thus have
\begin{align*}
    d^2_{2,p}\left(\beta^*, \beta_n\right) \to d_\infty \quad \text{a.s.}, \quad 
    \sum_{n = 1}^\infty \alpha_n g\left(\beta^*,\beta_n\right) < \infty \quad \text{a.s }
\end{align*}

We now seek to show $d_\infty = 0$ a.s.
Let $C_X$ denote the magnitude of the maximum value of $\mathcal{X}$ in all dimensions which is finite by {{Assumption \ref{ass:covariate}}}. Note that we have
\begin{align*}
    \left|Q^*(u \mid x) - Q^*(u' \mid x)\right|  &= \left|\left(\beta^*(u) - \beta^{*}(u')\right)^T x\right| \\
    &\leq C_X\sum_{j = 1}^p \left|\beta^*_j(u) - \beta^*_j(u')\right|\\
    &\leq C_X\, p\,L \, |u - u'|
\end{align*}
where $L$ is the Lipschitz constant from {Assumption \ref{ass:reg_lipschitz}}.
We thus have $M= C_XpL<\infty$ such that
\begin{align*}
   |Q^*(u \mid x) - Q(u \mid x)| \leq   M \left|u - P^*\left(Q(u \mid x) \mid x\right)\right|
\end{align*}
where we have plugged in $u' = P^*\left(Q(u \mid x)\mid x\right)$, and $M$ is chosen uniformly over $x$. Then this gives
\begin{align*}
    g(\beta^*,\beta) \geq M^{-1}\int_{\mathcal{X}}\int_0^1  \left(Q^*(u \mid x) - Q(u \mid x)\right)^2 \, \, du \, dP^*(x)
\end{align*}
With the above, we have
\begin{align*}
g(\beta^*,\beta) 
    &\geq M^{-1}\int_{\mathcal{X}} \int_0^1\left(\beta^*(u) - \beta(u)\right)^Tx x^T\left(\beta^*(u) - \beta(u)\right) \,du \,  dP^*(x)\\
    &= M^{-1}\int_0^1 \left(\beta^*(u) - \beta(u)\right)^T\Sigma_x\left(\beta^*(u) - \beta(u)\right) \, du
\end{align*}
where $\Sigma = \int x x^T dP^*(x)$ and we have used Tonelli's theorem.
As $\Sigma_x$ is positive definite by {Assumption \ref{ass:covariate}}, we have that 
\begin{align*}
    \frac{x^T \Sigma_x x}{x^Tx} \geq \lambda_{\text{min}} > 0
\end{align*}
where $\lambda_{\text{min}}$ is the minimum eigenvalue of $\Sigma_x$. Therefore, we have that there exists $\varepsilon = \lambda_{\text{min}}M^{-1} >0$ such that
\begin{align*}
    g(\beta^*,\beta_n) \geq \varepsilon \,  d^2_{2,p}\left(\beta^*,\beta_{n}\right).
\end{align*}
As $\sum_{n = 1}^\infty \alpha_n g(\beta^*, \beta_n) < \infty$ a.s., this then ensures that $d_\infty = 0$ a.s. by the usual argument.
\end{proof}

We can show that consistency of $\beta_n$ implies consistency of the conditional quantiles.
\begin{cor}
 Under {Assumptions \ref{ass:covariate}}, {\ref{ass:linear}} and {\ref{ass:reg_lipschitz}}, we have that
             ${d}_{2,x}(Q^*(\cdot \mid x), Q_n(\cdot \mid x)) \to 0$ a.s.$[P^*]$
          as $n \to \infty$ under (\ref{eq:beta_consist1}) and (\ref{eq:beta_consist2}) with the $N\leq n$ form of $G_{N}$.
\end{cor}
\begin{proof}
First, we write
\begin{align*}
    {d}^2_{2,x}\left(Q^*(\cdot \mid x),Q_n(\cdot \mid x)\right) &= \int_{\mathcal{X}} \int_0^1 \left(Q^*(u \mid x) - Q_n(u \mid x)\right)^2 \, du\,  dP^*(x)\\
    &= \int_{\mathcal{X}} \int_0^1 \left(\beta^*(u) - \beta_n(u)\right)^T xx^T \left(\beta^*(u) - \beta_n(u)\right) \, du\,  dP^*(x)\\
    &= \int_0^1 \left(\beta^*(u) - \beta_n(u)\right)^T \Sigma_x \left(\beta^*(u) - \beta_n(u)\right) \, du
\end{align*}
which looks familiar. Using the other side of the inequality for Rayleigh's quotient, we have that
\begin{align*}
    \int_0^1 \left(\beta^*(u) - \beta_n(u)\right)^T \Sigma_x \left(\beta^*(u) - \beta_n(u)\right) \, du \leq  \lambda_{\text{max}}\int_0^1 \left(\beta^*(u) - \beta_n(u)\right)^T \left(\beta^*(u) - \beta_n(u)\right) \, du 
\end{align*}
where $\lambda_{\text{max}}$ is the maximum eigenvalue of the covariance matrix, which is bounded due to compact support from {Assumption \ref{ass:covariate}}. We then have the desired result as ${d}^2_{2,x}\left(Q^*(\cdot \mid x),Q_n(\cdot \mid x)\right) \leq \lambda_{\text{max}} \, d_{2,p}^2(\beta^*,\beta_n) \to 0$ a.s.
\end{proof}

Posterior consistency can then be showed as follows, where we work directly with $Q^{\dagger}(u \mid x)$ instead of $\beta_\infty(u)$ due to need to take into account the rearrangement operator. The setup is the same as {Theorem \ref{th:consistency}}, where we extend the probability space as before, with the additional ingredients of a vector of weights $w_{1:n}$ and random covariates $\{X_{n,n+1},X_{n,n+1},\ldots\}$ for each $n$ arising from the Bayesian bootstrap.
\begin{thm}
    Under {Assumptions \ref{ass:covariate}}, {\ref{ass:linear}} and {\ref{ass:reg_lipschitz}}, the QMP with $\rho = 1$ as in (\ref{eq:beta_consist1}) and (\ref{eq:beta_consist2}) is consistent, that is for any $\varepsilon > 0$, we have that
    \begin{align*}
    \Pi\left({d}^2_{2,x}(Q^*(\cdot \mid x),Q_{n\infty}^{\dagger}(\cdot \mid x)) \geq \varepsilon \mid Y_{1:n}, X_{1:n}\right) \to 0 \quad \textnormal{a.s.}[P^*]
    \end{align*}
 \end{thm}
 \begin{proof}
Once again, Markov's inequality gives us
\begin{align*}
    \Pi\left({d}^2_{2,x}(Q^*(\cdot \mid x),Q_{n\infty}^{\dagger}(\cdot \mid x)) \geq \varepsilon \mid Y_{1:n}, X_{1:n}\right)\leq \frac{1}{\varepsilon^2}\E\left[{d}^2_{2,x}(Q^*(\cdot \mid x),Q_{n\infty}^{\dagger}(\cdot \mid x)) \mid Y_{1:n}, X_{1:n}\right].
\end{align*}
Expanding out the triangle inequality, we have
\begin{equation}
\begin{aligned}\label{app:eq_triangle}
    &E\left[{d}^2_{2,x}(Q^*(\cdot \mid x),Q_{n\infty}^{\dagger}(\cdot \mid x)) \mid Y_{1:n}, X_{1:n}\right]  \\
    &\leq\E\left[{d}^2_{2,x}(Q_{n\infty}^{\dagger}(\cdot \mid x), Q^{\dagger}_n(\cdot \mid x))\mid Y_{1:n}, X_{1:n}\right] \\
    &+ 2E\left[{d}_{2,x}(Q_{n\infty}^{\dagger}(\cdot \mid x), Q^{\dagger}_n(\cdot \mid x))\mid Y_{1:n}, X_{1:n}\right]{d}_{2,x}(Q^*(\cdot \mid x), Q^{\dagger}_n(\cdot \mid x))\\
&+ {d}^2_{2,x}(Q^*(\cdot \mid x), Q^{\dagger}_n(\cdot \mid x)).
\end{aligned}
\end{equation}

For each $x \in \mathcal{X}$, $Q^*(\cdot \mid x)$ and $Q_n(\cdot \mid x)$ has compact range from {Assumption \ref{ass:reg_lipschitz}}. We can thus apply {Proposition \ref{prop:rearr}} to show  
${d}^2_{2}(Q^*(\cdot \mid x), Q^{\dagger}(\cdot \mid x)) \leq {d}^2_{2}(Q^*(\cdot \mid x), Q(\cdot \mid x))$ for each $x \in \mathcal{X}$, which gives
\begin{align*}
    {d}^2_{2,x}(Q^*(\cdot \mid x), Q^{\dagger}(\cdot \mid x)) \leq {d}^2_{2,x}(Q^*(\cdot \mid x), Q(\cdot \mid x)).
\end{align*}

For the final term then, we have
\begin{align*}
 {d}^2_{2,x}(Q^*(\cdot \mid x), Q^{\dagger}_n(\cdot \mid x)) \leq  {d}^2_{2,x}(Q^*(\cdot \mid x), Q_n(\cdot \mid x)) \to 0 \quad \textnormal{a.s.}[P^*]
\end{align*}
For the first term , we also apply the rearrangement inequality to get
\begin{align*}
   \E\left[{d}^2_{2,x}(Q_{n\infty}^{\dagger}(\cdot \mid x), Q^{\dagger}_n(\cdot \mid x))\mid Y_{1:n}, X_{1:n}\right] \leq   \E\left[{d}^2_{2,x}(Q_{n\infty}(\cdot \mid x), Q_n(\cdot \mid x))\mid Y_{1:n}, X_{1:n}\right]
\end{align*}
which we now bound. First, we look at the inner term
\begin{align*}
    &{d}^2_{2,x}(Q_{n\infty}(\cdot \mid x), Q_n(\cdot \mid x)) = \int_{\mathcal{X}} \int_0^1 \left[\left(\beta_{n\infty}(u) - \beta_{n}(u)\right)^T x\right]^2 \, du \, dP^*(x)\\
    &=\sum_{i = n+1}^\infty \sum_{j = n+1}^\infty \alpha_i\, \alpha_j \int_0^1\left(u - \mathbbm{1}\left(V_i \leq u\right)\right)\, \left(u - \mathbbm{1}\left(V_j \leq u\right)\right) \, du\, X_i^T \int_{\mathcal{X}} x x^T dP^*(x)\,  X_j
\end{align*}
where we have applied Tonelli's theorem. Taking the expectation conditional on $\{w_{1:n}, Y_{1:n}, X_{1:n}\}$ gives
\begin{align*}
    &\E\left[{d}^2_{2,x}(Q_{n\infty}(\cdot \mid x), Q_n(\cdot \mid x))\mid w_{1:n}, Y_{1:n}, X_{1:n}\right]\\
    &=\sum_{i = n+1}^\infty \alpha_i^2 \int_0^1 (u - \mathbbm{1}\left(V_i  \leq u\right))^2\, du \, E[X_i^T   \Sigma_x X_i\mid w_{1:n}]\,
\end{align*}
where the cross-terms disappear as $V_i$ is independent of $V_j$ (and both are independent of $X_i,X_j$) for $i\neq j$, and the terms have mean 0.
We can upper bound the above term by
\begin{align*}
   E\left[X_{n+1}^T   \Sigma_x X_{n+1}\mid w_{1:n}\right] \sum_{i = n+1}^\infty \alpha_i^2  =O(n^{-1}) \left[\sum_{i = 1}^n w_i X_i^T \Sigma_x X_i\right]
\end{align*}
where we have used the fact that $X_{n+1:\infty}$ are i.i.d. conditional on $w_{1:n}$. Taking the expectation over the weights then gives
\begin{align*}
    \E\left[{d}^2_{2,x}(Q_{n\infty}(\cdot \mid x), Q_n(\cdot \mid x))\mid Y_{1:n}, X_{1:n}\right]\leq O(n^{-1}) \left[\frac{1}{n}\sum_{i= 1}^n X_i^T \Sigma_x X_i\right]
\end{align*}
As $\Sigma_x$ has finite eigenvalues from {Assumption \ref{ass:covariate}}, we have that
\begin{align*}
\E\left[{d}^2_{2,x}(Q_{n\infty}(\cdot \mid x), Q_n(\cdot \mid x))\mid Y_{1:n}, X_{1:\infty}\right]\leq O(n^{-1})\left(\frac{1}{n}\sum_{i = 1}^n X_i^T X_i\right)
\end{align*}
We then have $\frac{1}{n}\sum_{i = 1}^n X_i^T X_i \to V^2 < \infty$ a.s., which gives
\begin{align*}
\E\left[{d}^2_{2,x}(Q_{n\infty}(\cdot \mid x), Q_n(\cdot \mid x))\mid Y_{1:n}, X_{1:\infty}\right]\leq O(n^{-1}) \quad \text{a.s.}[P^*]
\end{align*}
A similar argument as in the proof of {Theorem \ref{th:consistency}} can be used to handle the second cross-term in (\ref{app:eq_triangle}), so we have the result.
 \end{proof}
 We thus have posterior consistency of the QMP for linear regression for $\rho = 1$.
A similar result can  likely be derived for the posterior contraction rate. Like in the quantile estimation case, we suspect that the QMP with the smoothed update (\ref{eq:beta_freq}) satisfies a similar result on posterior consistency and contraction, but it is not immediately obvious due to the non-linearity of the increasing rearrangement operator. In practice, we see that the rearrangement is negligible for the initial estimate for a sequence $\rho_N$ which approaches $1$ sufficiently slowly. We leave a thorough investigation of this for future work.

\section{Practical details}

\subsection{Implementation}

In this section, we outline some computation details that were not mentioned in the main paper. All methods were implemented in the \texttt{JAX} package in Python, which is efficient and competitive with \texttt{C++} in terms of computational speed. The bivariate copula term $C_{\rho}(u,u')$ can be computed efficiently using standard approximations to the bivariate normal CDF; we utilize the implementation in \texttt{scipy}. 

As the quantile function is scalar on bounded support $(0,1)$, it is efficient to implement methods based on a uniform discrete grid of size $n_U$. Rearrangement is particularly straightforward in this case, as it just involves sorting the values in increasing order \citep{Chernozhukov2009}. For selecting $c$, we compute estimates of $q_n^{\dagger}$ by taking finite differences of $Q_n^{\dagger}$ on the grid. We suspect it is possible to compute this more accurately using the derivatives of the update function but the rearrangement step makes it nontrivial. Finally, we outline the computational complexity of the main algorithms. Estimation of $Q_n^{\dagger}$ ({Algorithm \ref{alg:fit}}) has an average time complexity of  $O(n \times n_U\log n_U )$ due to the sorting required for rearrangement, but in practice rearrangement is not required for each step.
Exact quantile predictive resampling ({Algorithm \ref{alg:QMP}}) has a time complexity of $O(B\times (N-n) \times n_U)$, where $N$ governs the truncation of predictive resampling. In practice, we select $N = n + 5000$, although $N$ can likely shrink with $n$ \citep{Fong2023a}. Approximate predictive resampling ({Algorithm \ref{alg:approx_QMP}}) has time complexity $O(B \times n_U^2)$, which is much faster in practice.
For all examples, we set the grid size to $n_U = 200$, which does not need to grow with $n$. Finally, the algorithms for quantile regression ({Algorithms \ref{alg:reg_fit}}, {\ref{alg:reg_approx_QMP}} and {\ref{alg:QMP_reg}})  have the same complexity as the original unconditional algorithms
multiplied by a factor of $p$.

\subsection{Algorithms}\label{app:sec:alg}
In this section, we outline a few algorithms that were omitted from the main paper due to space constraints. {Algorithm \ref{alg:QMP_reg}} outlines the exact quantile predictive resampling method for quantile regression, where we carry out the exact Bayesian bootstrap for the covariates for expediency.
{Algorithm \ref{alg:GP_samp}} then illustrates how to draw a sample from a GP with kernel $C_{\rho}(u,u')- uu'$ on a finite grid of size $n_U$, which is essentially just equivalent to drawing a Gaussian vector. {Algorithm \ref{alg:reg_GP_samp}} is then a natural extension to generate the GP from the covariate-dependent kernel, which we highlight is conditional on the randomly drawn $w_{1:n}$, so technically we would only want to draw $B = 1$ for each sample of $w_{1:n}$. This involves drawing $p$ independent GPs and carrying out an affine transformation to induce the covariate dependence.
\begin{figure}[!h]
\small
\center
\begin{minipage}{.9\linewidth}
\begin{algorithm}[H]
{Initialize $\beta_n$ from {Algorithm \ref{alg:reg_fit}}}\\
\For{$b \gets 1$ \textnormal{\textbf{to}} $B$}{
{Draw $w_{1:n}^{(b)} \sim \text{Dirichlet}(1,\ldots,1)$ and $X^{(b)}_{n+1:N} \iid \sum_{i = 1}^n w_i \delta_{X_i}$}\\
\For{$i \gets n+1$ \textnormal{\textbf{to}} $N$}{
{Draw $V^{(b)}_i \sim \mathcal{U}(0,1)$}\\
{{$\beta^{(b)}_{i}(u) = \beta^{(b)}_{i-1}(u) + \alpha_{i}\left[u - H_{\rho_{i}}\left(u, V^{(b)}_i\right)\right]\, X^{(b)}_i$}   
}}}
 {Return $\left\{{\beta_N}^{(1)},\ldots,{\beta_N}^{(B)}\right\}$}
\caption{QMP sampling for regression}\label{alg:QMP_reg}
\end{algorithm}
\end{minipage}
\end{figure}\vspace{-5mm}
\begin{figure}[!h]
\small
\center
 \begin{minipage}{.9\linewidth}
\begin{algorithm}[H]
Initialize uniform grid ${U}$ on $[0,1]$ of size $n_U$  \\
Compute the $n_U \times n_U$ matrix $\Sigma$ where $\Sigma_{ij} = C_\rho(u_i,u_j) - u_i u_j$ for $(u_i,u_j) \in U\times U$\\
Compute Cholesky decomposition $\Sigma = LL^T$\\
\For{$b \gets 1$ \textnormal{\textbf{to}} $B$}{
    {Draw $Z^{(b)}\sim \mathcal{N}(0,I_{n_U})$} \\
    {Compute $S^{(b)} = LZ^{(b)}$}
}
 {Return $\left\{S^{(1)},\ldots, S^{(B)} \right\}$}
\caption{Sampling from GP with kernel $C_\rho(u,u') - uu'$}\label{alg:GP_samp}
\end{algorithm}
\end{minipage}
\end{figure}\vspace{-5mm}
\begin{figure}[!h]
\small
\center
 \begin{minipage}{.9\linewidth}
\begin{algorithm}[H]
Initialize uniform grid ${U}$ on $[0,1]$ of size $n_U$  \\
Compute the $p \times p$ matrix $\Sigma_{x} = \sum_{i = 1}^n w_i X_i X_i^T$\\
Compute Cholesky decomposition $\Sigma_{x} = L_{x}{L_{x}}^T$\\
\For{$b \gets 1$ \textnormal{\textbf{to}} $B$}{
    {Draw $Z_{1:p}^{(b)}\iid \mathcal{GP}(0, C_{\rho^2} (u,u')- uu')$ on grid $U \times U$} \\
    {Compute $S_{1:p}^{(b)} = L_xZ_{1:p}^{(b)}$}
}
 {Return $\{S_{1:p}^{(1)},\ldots, S_{1:p}^{(B)} \}$}
\caption{Sampling from GP with covariate-dependent kernel $k_{\rho}(\{u,j\}, \{u',j'\}; w_{1:n})$}\label{alg:reg_GP_samp}
\end{algorithm}
\end{minipage}
\end{figure}

\section{Additional experiments \& discussion}
\subsection{Simulations}\label{app:sec_sim}
In this section, we include additional results for the simulations. {Figure \ref{fig:simulation_exact}} illustrates the equivalent to {Figure \ref{fig:simulation}} but with the exact sampler. {Figure \ref{fig:simulation_samples}} additionally shows sample paths for the exact and approximate QMP. We see that there is little difference between the exact and approximate sampler, even with $n = 50$.
\begin{figure}[!h]
\begin{center}
\includegraphics[width=\textwidth]{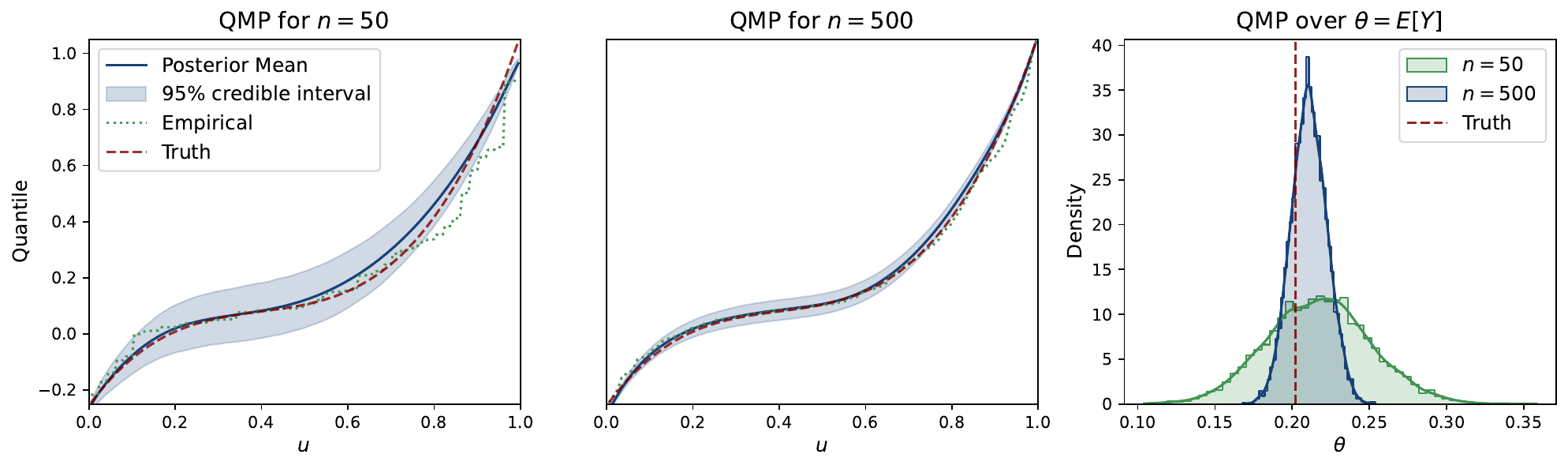}
\end{center}
\caption{QMP over $Q_\infty^{\dagger}$ with exact sampling for (Left) $n = 50$; (Middle) $n = 500$; (Right) QMP over  $\theta = \E[Y]$; predictive resampling is truncated at $N = n + 5000$ } 
\label{fig:simulation_exact}
\end{figure}\vspace{-2mm}
\begin{figure}[!h]
\begin{center}
\includegraphics[width=\textwidth]{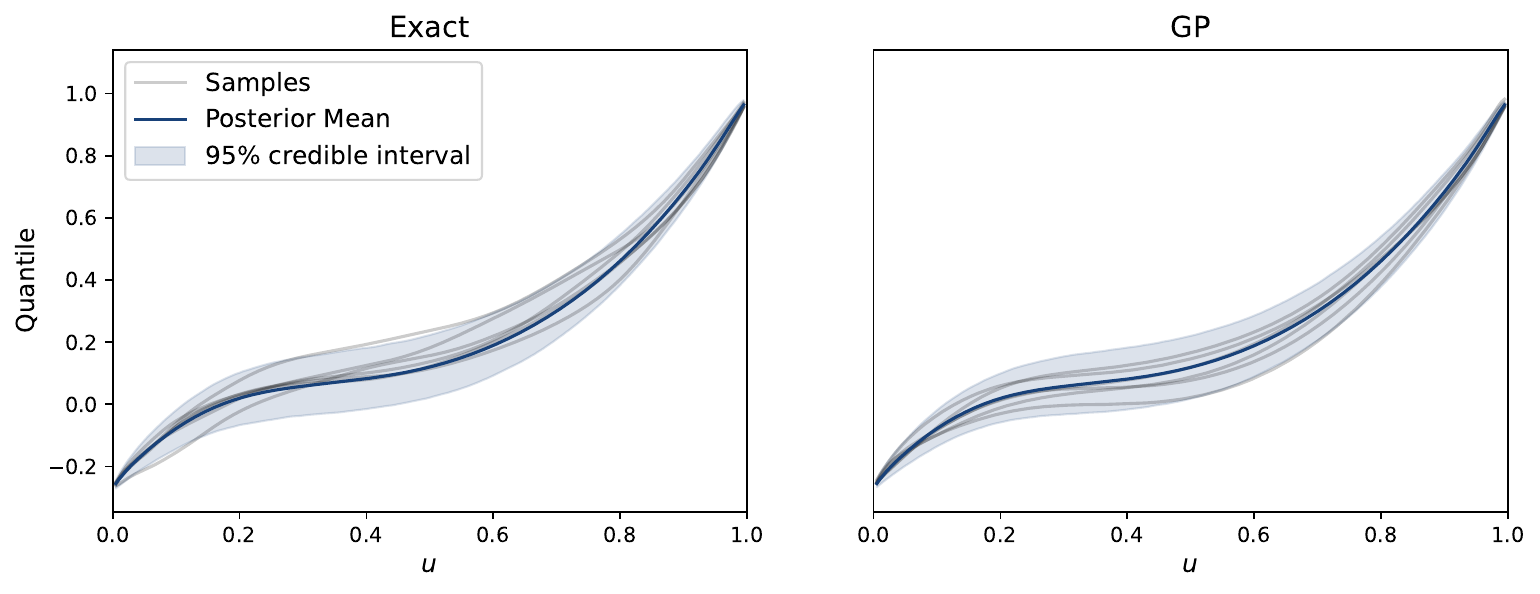}
\end{center}
\caption{QMP over $Q_\infty^{\dagger}$ in the $n = 50$ setting with posterior samples for (Left) Exact sampling; (Right) GP approximation} 
\label{fig:simulation_samples}
\end{figure}

\subsection{Cyclone dataset experiment}\label{app:sec_cyc}
In this section, we include additional results for the cyclone data experiment with $n = 291$ in the NA basin. {Figures \ref{fig:reg_small_exact_samples}} and {\ref{fig:reg_small_samples}} illustrate the QMP for the conditional quantile functions and quantile regression curves for the exact and approximate QMPs respectively. Once again, we see that the two sampling schemes are visually indistinguishable.
\begin{figure}[!h]
\begin{center}
\includegraphics[width=\textwidth]{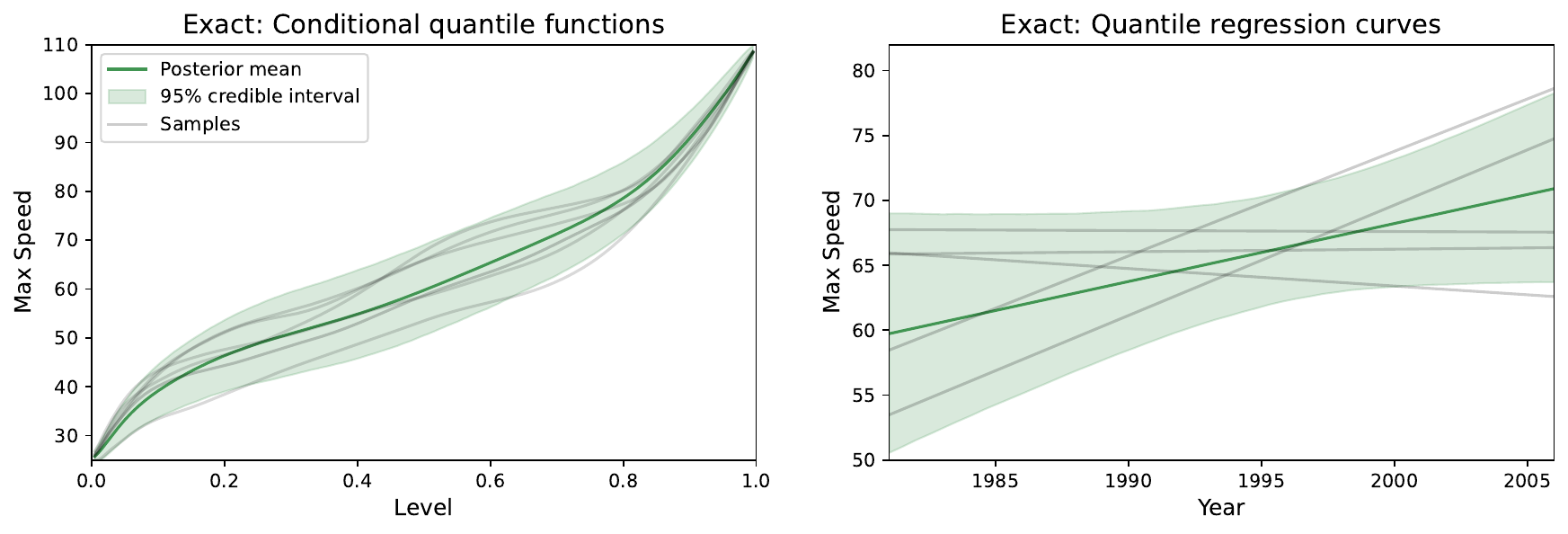}
\end{center}
\caption{Tropical cyclone maximum speeds in the NA basin ($n = 291$): (Left) Posterior mean, 95\% credible intervals and samples for $Q(u \mid x = 1981)$ from the exact QMP; (Right)  Posterior mean, 95\% credible intervals and samples for $Q^{\dagger}_\infty(u = u^* \mid x)$ for $u^* =0.5 $ from the exact QMP } 
\label{fig:reg_small_exact_samples}
\end{figure}\vspace{5mm}
\begin{figure}[!h]
\begin{center}
\includegraphics[width=\textwidth]{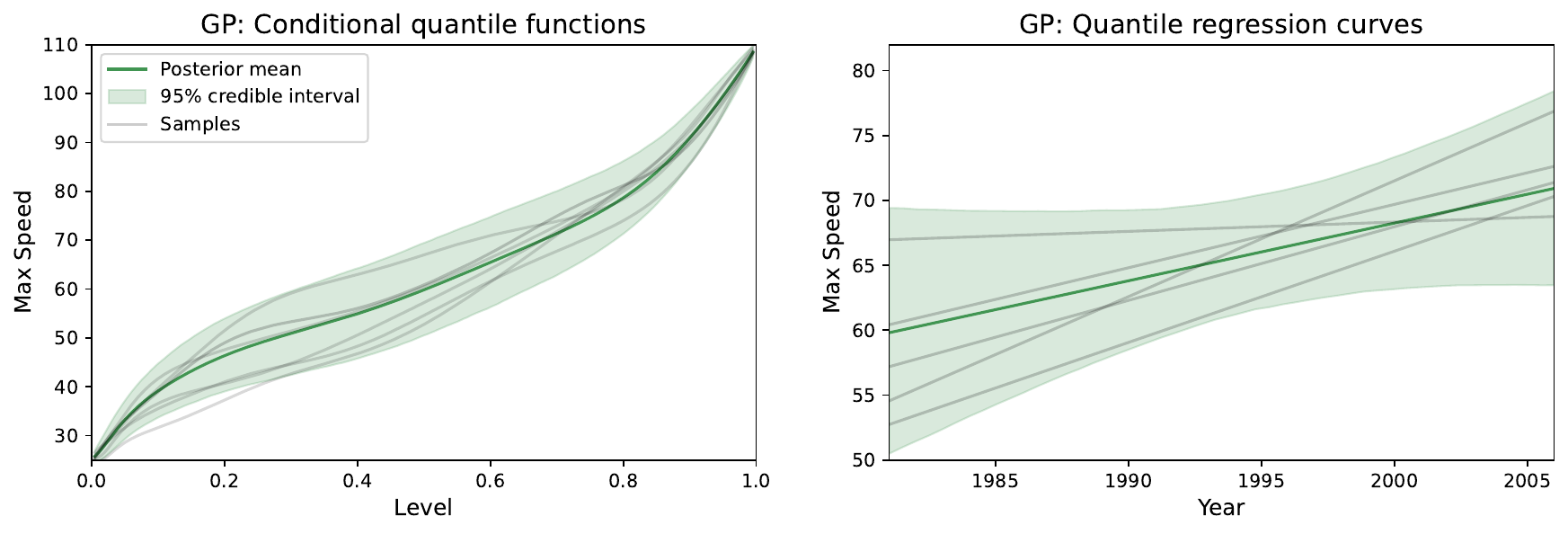}
\end{center}
\caption{Tropical cyclone maximum speeds in the NA basin ($n = 291$): (Left) Posterior mean, 95\% credible intervals and samples for $Q(u \mid x = 1981)$ from the approximate QMP; (Right)  Posterior mean, 95\% credible intervals and samples for $Q^{\dagger}_\infty(u = u^* \mid x)$ for $u^* =0.5 $ from the approximate QMP} 
\label{fig:reg_small_samples}
\end{figure}

\subsection{Functional learning rates}\label{app:sec_func}
As discussed in the main paper, we can consider a functional learning rate $a(u)$. It is not too difficult to extend {Theorem \ref{th:consistency_mean}} if $a(u)$ satisfies the following.
\begin{as}[Functional learning rate]\label{as:alpha_func}
    The learning rate sequence takes the form $\alpha_i(u)= a(u)(i+1)^{-1}$, where
 $L \leq a(u) \leq U$ for some constants $0 < L < U < \infty$. 
\end{as}
\begin{thm}
   Under {Assumptions  \ref{as:alpha_func}}, {\ref{as:bandwidth}} and {\ref{as:consistency_truth}}, we have that $d_2( Q_n^{\dagger},Q^*) \to 0$ a.s.$[P^*]$ under a variant of {Algorithm \ref{alg:fit}} with a functional learning rate $a(u)$.
\end{thm}
\begin{proof}
    Following the proof of {Theorem \ref{th:consistency_mean}}, let $L_n := d_2^2(Q_n^{\dagger},Q^*)$. We once again have
    \begin{align}\label{app:eq_asmart_func}
        \E\left[L_n \mid \mathcal{F}_{n-1}\right] \leq L_{n-1} - 2(n+1)^{-1}\widetilde{T}(Q_{n-1}^{\dagger}) + U^2 (n+1)^{-2}+ 2 (n+1)^{-1} \widetilde{\zeta}_n
    \end{align}
    where
    \begin{align*}
         \widetilde{T}(Q^{\dagger}_{n-1}) &= \int a(u) \left(Q^*(u) - Q^\dagger_{n-1}(u)\right)\left(u- P^*\left(Q^{\dagger}_{n-1}(u)\right)\right)\, du\\
         \widetilde{\zeta}_n &=\int \, a(u)\,  \left(Q^*(u) - Q^\dagger_{n-1}(u)\right)\left( K_n(u) - P^*\left(Q_{n-1}^{\dagger}(u)\right)\right)du
    \end{align*}
    and $K_n(u) = \int  H_{\rho_n}(u,P_{n-1}(y))\, p^*(y) \, dy$ as before.
Once again, we have $\widetilde{T}(Q) \geq 0$ as the integrand is always positive. We can upper bound $|\zeta_n|$ again with
\begin{align*}
    |\zeta_n| &\leq \sqrt{L_{n-1}} \sqrt{\int a^2(u) \left(K_n(u) - P^*\left(Q_{n-1}^{\dagger}(u)\right)\right)^2 \, du}\\
    &\leq (L_{n-1} +1)\,  U \sqrt{\kappa_n},
\end{align*}
where $\kappa_n$ is defined in (\ref{app:eq_kappan}). It is thus again sufficient to show $\sum_{i=1}^\infty (i+1)^{-1}\sqrt{\kappa_i} < \infty$ which occurs under the same assumptions as before (i.e. {Assumption \ref{as:bandwidth}}).

Once again, we have $L_{n} \to L_\infty$ a.s. and $\sum_{i=1}^\infty (i+1)^{-1} \widetilde{T}(Q^{\dagger}_{i-1}) < \infty$ a.s. from {Theorem \ref{app:thm_asmart}}. As the integrand in $\widetilde{T}(Q)$ is positive, we can further lower bound
\begin{align*}
    \widetilde{T}(Q) \geq LT(Q),
\end{align*}
so we also have $\sum_{i=1}^\infty (i+1)^{-1} {T}(Q^{\dagger}_{i-1}) < \infty$ a.s. The same argument based on the Lipschitz constant can then be applied to show $L_\infty = 0$ a.s.
\end{proof}

 In practice, an intuitive choice for the functional learning rate is to set 
 \begin{align*}
     a(u) = \frac{1}{\hat{p}(\hat{Q}(u))},
 \end{align*}
 where $\hat{p}$ and $\hat{Q}$ are estimates of $p^*$ and $Q^*$ respectively. This can be motivated by optimal learning rates for attaining efficient stochastic approximation of pointwise quantiles which is also suggested by \cite{Aboubacar2014}.  In the quantile regression case, under appropriate assumptions, this would involve estimating the residuals via linear regression, then estimating $\hat{p}$ and $\hat{Q}$ from the residuals. As discussed in the main paper however, it is unsatisfying that a separate density estimate is required, and the results will also be quite sensitive to this density estimate.

{Figure \ref{fig:simulation_func}} illustrates the same experiment as Section \ref{sec:simulation} but instead with $a(u)$ as above, where we estimate $\hat{p}$ and $\hat{Q}$ with the Gaussian kernel density estimate and empirical quantile function respectively. We set $c = 0.7$ to match the settings of the main paper, with all other settings the same. We can see that the center and tails have slightly less and more uncertainty respectively compared to the main paper,
 due to the adaptive $a(u)$. However, the estimates and intervals are quite non-smooth despite setting a large value of $c = 0.7$.
This suggests that while an adaptive $a(u)$ may help with estimating $Q^*$, it may not be better for estimating the quantile density function $q^*$. Finally, we see that the posterior of the mean functional looks quite similar to the fixed $a$ result. 
\begin{figure}[!h]
\begin{center}
\includegraphics[width=\textwidth]{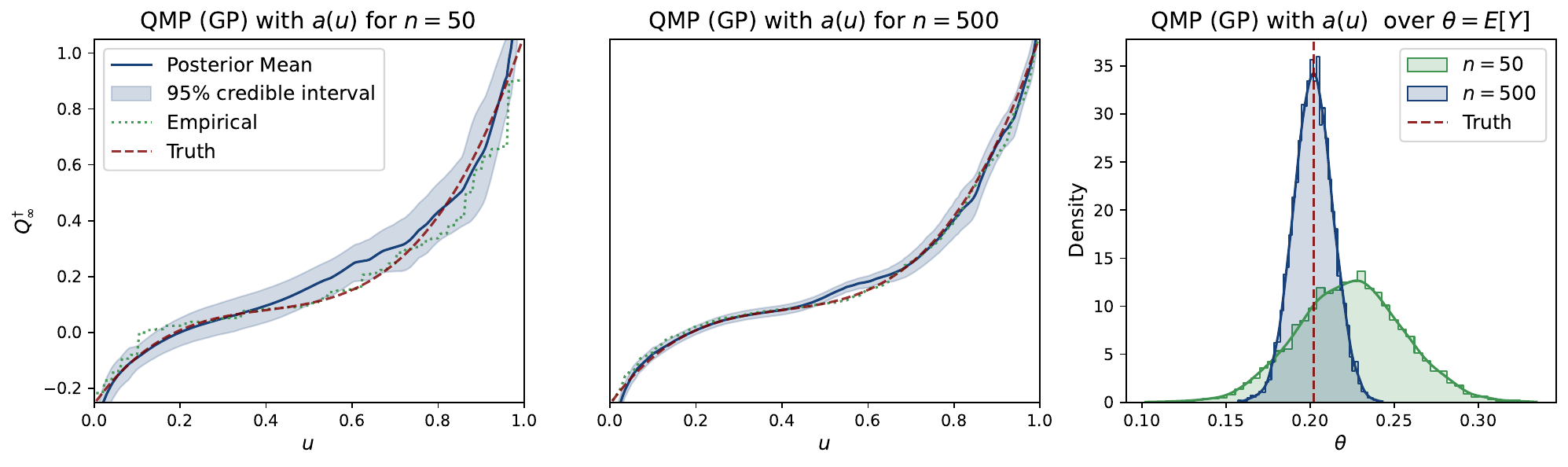}
\end{center}
\caption{QMP over $Q_\infty^{\dagger}$ with functional learning rate $a(u)$ for (Left) $n = 50$; (Middle) $n = 500$; (Right) QMP over  $\theta = \E[Y]$;  we only show the GP approximation as it is visually indistinguishable from exact sampling. } 
\label{fig:simulation_func}
\end{figure}\vspace{-5mm}

\subsection{Comparison to the Bayesian bootstrap for quantile regression}
We now draw comparisons between the QMP for quantile regression with $\rho = 1$ with the Bayesian bootstrap.
Consider now a new test point $x$, which is distinct from $X_{1:n}$. Under the BB, the posterior distribution over $\E[Y \mid x]$ is always 0 in this case, as the BB only allocates mass to $x = X_i$. However, the QMP will be the distribution of
\begin{align*}
   \E_\infty[ Y \mid x] =   \int Q^{\dagger}_{\infty}(u\mid x) \, du &=  \int Q_{\infty}(u\mid x) \, du=x^T\int \beta_{\infty}(u)\, du.
\end{align*}
This is thus non-zero for all values of $x$. The QMP thus allows posterior inference on $\E[Y \mid x]$ for the whole covariate space, which the Bayesian bootstrap is unable to do. Of course the same argument also applies if we are interested in the posterior over $Q(u \mid x)$ for some $x$ not in the support of the data.

    
\end{appendices}
\end{document}